\documentclass[aps,prx,reprint,superscriptaddress,longbibliography,floatfix,
               amsmath,amssymb]{revtex4-2}
\usepackage[T1]{fontenc}
\usepackage[utf8]{inputenc}
\usepackage{amsthm,mathtools,bm,mathrsfs}
\usepackage{microtype,needspace}
\usepackage{booktabs}
\usepackage{algpseudocode}
\usepackage{xcolor}
\usepackage[normalem]{ulem}
\usepackage[colorlinks,linkcolor=blue,citecolor=blue]{hyperref}
\usepackage{enumitem}
\usepackage{orcidlink}
\setlist{nosep,leftmargin=*}
\allowdisplaybreaks[2]
\newtheorem{theorem}{Theorem}
\newtheorem*{informaltheorem}{Informal theorem}
\newtheorem{lemma}{Lemma}
\newtheorem{proposition}{Proposition}
\newtheorem{corollary}{Corollary}
\theoremstyle{definition}
\newtheorem{definition}{Definition}
\newtheorem{assumption}{Assumption}
\newtheorem*{informalassumption}{Assumptions (informal)}
\DeclareMathOperator{\Tr}{Tr}
\DeclareMathOperator{\tr}{tr}
\DeclareMathOperator{\diam}{diam}
\DeclareMathOperator{\supp}{supp}
\DeclareMathOperator{\id}{id}
\DeclareMathOperator{\ad}{ad}
\DeclareMathOperator{\sech}{sech}
\DeclareMathOperator{\poly}{poly}
\DeclareMathOperator{\polylog}{polylog}
\newcommand{\I}{\mathbf 1}
\newcommand{\E}{\mathbb E}

\newcommand{\He}{\widetilde H}
\newcommand{\PP}{\mathcal P}

\newcommand{\norm}[1]{\left\lVert#1\right\rVert}

\newcommand{\trn}[1]{\tr_{#1}}
\newcommand{\wt}{w_{\lambda,\mu}}
\newcommand{\Csh}[1]{C_{\mathrm{sh}}(#1)}
\newcommand{\fail}{p_{\mathrm{fail}}}

\hypersetup{pdftitle={Polylog-depth Quantum Thermal Simulation via Local Recovery Channels},pdfauthor={}}

\newcounter{algorithm}
\makeatletter
\newenvironment{algorithm}[1][tbp]{%
 \begin{figure}[#1]\small
 \let\c@figure\c@algorithm

 \def\fnum@figure{\textbf{Algorithm~\thealgorithm}}%
 \hrule height .6pt\smallskip
}{\par\smallskip\hrule height .6pt\end{figure}}
\makeatother
\begin{document}
\title{Polylog-depth Quantum Thermal Simulation via Local Recovery Channels}
%\title{From local Petz recovery to polylogarithmic-depth Gibbs-state purification}
\author{Hideaki Hakoshima\,\orcidlink{0000-0003-1295-9801}}
\email{hakoshima.hideaki.qiqb@osaka-u.ac.jp}
%\affiliation{Graduate School of Engineering Science, University of Osaka, 1-3 Machikaneyama, Toyonaka, Osaka 560-8531, Japan}
\affiliation{Center for Quantum Information and Quantum Biology, University of Osaka, 1-2 Machikaneyama, Toyonaka, Osaka 560-0043, Japan}

\author{Atsushi Iwaki\,\orcidlink{0000-0002-5345-8045}}
\affiliation{Department of Physics, University of Tokyo, 7-3-1 Hongo, Bunkyo-ku, Tokyo 113-0033, Japan}

\author{Nobuyuki Yoshioka\,\orcidlink{0000-0001-6094-8635}}
\affiliation{International Center for Elementary Particle Physics, University of Tokyo, 7-3-1 Hongo, Bunkyo-ku, Tokyo 113-0033, Japan}
\date{\today}
\begin{abstract}
We establish sufficient conditions for preparing quantum thermal states of noncommuting local Hamiltonians with polylogarithmic circuit depth in arbitrary fixed spatial dimension. Our conditions combine locality and stability bounds on the effective interactions of reduced density matrices of a Gibbs state with a quantitative high-temperature condition and access to coarse classical reference Hamiltonians. When these references can be generated locally, the classical preprocessing cost is $N^{1+o(1)}$ for $N$ sites at inverse-polynomial global accuracy. We construct the preparation circuit from spatially localized Petz recovery maps. Quantum corrections derived from the microscopic Hamiltonian enable accurate recovery while the classical references remain coarse. After preprocessing, the resulting circuit prepares the canonical purification using $N\operatorname{polylog}(N/\varepsilon)$ gates and qubits, where $\varepsilon$ is the preparation error. These results link two fundamental questions: how correlations are organized in thermal equilibrium, and how efficiently the corresponding states can be realized through operations allowed by quantum mechanics. By translating static equilibrium structure into explicit preparation circuits, they give equilibrium locality a constructive computational interpretation and a physically grounded role in quantum algorithm design.
\end{abstract}
\maketitle
\section{Introduction}
\label{sec:introduction}

Simulating quantum many-body physics is a foundational motivation for
quantum computing, dating back to Feynman's proposal to simulate quantum
systems with quantum devices~\cite{Feynman1982}. Gibbs states are central
to this goal: they describe thermal equilibrium and underpin the study of
finite-temperature phases, thermodynamic properties, and equilibrium
correlations~\cite{Alhambra2023Thermal}. Two complementary research
directions address these states. One investigates their correlations and
information-theoretic structure~\cite{WolfEtAl2008,Alhambra2023Thermal};
the other develops quantum algorithms for preparing them
\cite{PoulinWocjan2009,TemmeEtAl2011,ChenKastoryanoBrandaoGilyen2025}.
A fundamental question connecting these directions is how the physical
structure of thermal equilibrium can be translated into efficient
quantum preparation circuits.

Conditional independence provides an operational connection between state
structure and reconstruction. For strictly positive classical
distributions, the Hammersley--Clifford theorem identifies Gibbs
factorization over graph cliques with the Markov
property~\cite{Besag1974}. In particular, for a tripartition $A,B,C$ of a
finite classical Gibbs system, $A$ and $C$ are conditionally independent whenever
$B$ separates them in the interaction graph. Equivalently, their
conditional mutual information (CMI) satisfies $I(A:C\mid B)=0$.
This is the Markov property of classical Gibbs distributions, valid
at every finite temperature. Noncommuting quantum Gibbs states need not
satisfy this exact condition, making its approximate counterpart a
nontrivial question about thermal correlations
\cite{ChenRouze2025LocallyMarkovian}.
The operational meaning of small CMI is established: it guarantees
approximate reconstruction of the tripartite state from its marginal on
$AB$ by a channel acting only on $B$ and restoring $C$
\cite{FawziRenner2015,JungeEtAl2018}.

The spatial organization of correlations in thermal equilibrium has long
been a central subject of many-body physics. Traditionally, two-point
correlation functions have played a prominent role, while
information-theoretic approaches also characterize multipartite
correlations and entanglement~\cite{Alhambra2023Thermal}.
Recent progress on CMI includes exponential decay in one-dimensional
finite-range systems at every fixed finite temperature
\cite{Kuwahara2025Markov}, high-temperature decay obtained through quantum
belief-propagation channels~\cite{KatoKuwahara2025BP}, and quasi-local
recovery of shielded fixed-size regions in higher dimensions at arbitrary
fixed temperature~\cite{ChenRouze2025LocallyMarkovian}.
Complementing these results on conditional correlations,
Bakshi~et al.~\cite{BakshiLiuMoitraTang2024} established full separability
and efficient classical sampling of a product-state decomposition for
sufficiently high-temperature Gibbs states. This connects a structural property
of thermal states to a constructive simulation method.

A complementary algorithmic approach designs dynamics with the desired
Gibbs state as a stationary state. Its development has benefited from
broader advances in quantum matrix algorithms, notably quantum singular
value transformation (QSVT), which implements polynomial transformations
of encoded singular values~\cite{GilyenEtAl2019}.
Recent work provides efficiently implementable Lindbladians satisfying
Kubo--Martin--Schwinger (KMS) detailed balance
\cite{ChenKastoryanoGilyen2023,ChenKastoryanoBrandaoGilyen2025,DingLiLin2025,GilyenChenDoriguelloKastoryano2024},
alongside microscopic derivations and local-circuit implementations
\cite{ScandiAlhambra2026,HahnSwekeDeshpandeShtanko2026}.
Thermal-state preparation then proceeds by simulating the corresponding
dynamics. Its cost depends on both the implementation of the generator
and its mixing time, whose analysis remains difficult in general.
At sufficiently high temperatures, rapid-mixing bounds yield efficient
Gibbs state preparation~\cite{RouzeFrancaAlhambra2026Thermalization},
with algorithms achieving nearly
linear preparation costs~\cite{RouzeFrancaAlhambra2026}.
This route establishes efficient preparation through the design and
convergence analysis of thermalization dynamics.

A complementary route reconstructs Gibbs states from their
spatial correlation structure. Brand{\~a}o and Kastoryano
\cite{BrandaoKastoryano2019} used uniform clustering and approximate
Markov conditions to construct shallow circuits of local channels.
Scalet~et al.~\cite{ScaletEtAl2025} connected decaying effective
interactions, local marginal reconstruction, and averages of rotated Petz
recovery maps to constructive Gibbs-sampling algorithms.
These approaches assemble thermal states from their spatial conditional
structure, providing a quantum counterpart to classical reconstruction
from local probabilities. Table~\ref{tab:comparison} compares their
assumptions and resource guarantees with other preparation methods.

The central question is how the local structure of thermal
equilibrium can yield an efficient preparation circuit built from
elementary quantum gates.
Local recoverability identifies where to act, but spatial support
alone does not determine the implementation cost: the Hilbert-space
dimension of a recovery region grows exponentially with its number
of sites
\cite{ScaletEtAl2025}.
This calls for a representation of the recovery operators that
makes their thermal-state structure directly usable in quantum
circuits.
%A remaining obstacle is that spatially local recovery need not be
%inexpensive to implement. Composing many approximate recoveries requires
%control of their accumulated error, so exponential locality bounds lead
%to neighborhoods whose radii grow logarithmically with system size and
%inverse accuracy. Their Hilbert-space dimensions are still exponential
%in the number of enclosed sites. Dense-matrix constructions can therefore
%remain costly, particularly in higher dimensions~\cite{ScaletEtAl2025}.
%General Petz algorithms also have costs governed by small nonzero
%eigenvalues of the reduced density operators~\cite{GilyenEtAl2022}.
%The existence of local effective Hamiltonians does not supply their
%interaction coefficients as accessible data. Thus the implementation
%problem remains even after the spatial support of recovery is controlled.

In this paper, we establish sufficient conditions under which
this equilibrium structure yields an efficient preparation circuit.
Our construction combines coarse classical references for the
effective Hamiltonians with quantum corrections derived from
the microscopic Hamiltonian.
These corrections retain the information in the reduced density
operators that the classical references omit.
The essential change is to make this information available directly
to the quantum circuit, so that accurate recovery does not require
equally accurate classical references.
The efficiency relies on controlling the spread and growth of local
operators during the required imaginary-time evolution and on an
integral representation involving bounded functions.

Under the preparation assumptions summarized in Sec.~II, our
construction prepares the canonical purification of an $N$-site
Gibbs state to error $\varepsilon$.
In any fixed spatial dimension, the preparation unitary has
$\operatorname{polylog}(N/\varepsilon)$ depth and uses
$N\operatorname{polylog}(N/\varepsilon)$ gates and qubits
after preprocessing.
When the classical references can be generated locally, the total
classical preprocessing cost is $N^{1+o(1)}$ at inverse-polynomial
global accuracy.

The result connects a static question---how correlations are
organized in thermal equilibrium---with an operational one:
how efficiently the corresponding many-body state can be built.
By realizing the local recovery maps with elementary quantum gates,
it makes this connection constructive and gives the local structure
of effective Hamiltonians a role in quantum algorithm design.

The preparation unitary and its inverse also supply coherent state
access for established multiple-observable estimation algorithms
\cite{HugginsEtAl2022,WadaYamamotoYoshioka2025}, enabling the estimation
of thermal energies and correlations. Under the same preparation
hypotheses, the construction also applies to finite-range fermionic
chains through the Jordan--Wigner transformation
\cite{VerstraeteCirac2005}. Further possibilities include
Gibbs-based optimization and local-model learning when the required
Hamiltonians remain within the regime of the preparation theorem
\cite{BrandaoSvore2016,AnshuEtAl2020Learning,CoopmansBenedetti2024}.

Section~\ref{sec:setting-summary} specifies the target and inputs.
Section~\ref{sec:theoretical-overview} reduces global preparation to
finite-patch recovery, presents the two local circuit constructions, and
combines them into the global resource bounds.
Sections~\ref{sec:application} and~\ref{sec:comparison} describe
applications and further directions. The appendices follow the proof:
locality, coherent composition, one-patch implementation, and global
error and resource estimates.

\begin{table*}[t]
\caption{Selected Gibbs-preparation methods (upper block) and ordinary-Petz
implementations (lower block), at fixed physical parameters. Here $N$ is the
system size and $\varepsilon$ is the global preparation error. Depth counts
elementary-gate layers unless marked as channel layers; dashes mean not
compared. The Scalet~et al. bound uses logarithmic recovery radius.
Our local-step row evaluates Theorem~\ref{cr:thm:local-main} with the
patch size and error allocation of Algorithm~\ref{alg:global-petz}.
Our unitary depth is after preprocessing; classical reference generation
separately costs $N^{1+o(1)}$ at inverse-polynomial accuracy under
Assumption~\ref{cr:ass:classical}(ii).}
\label{tab:comparison}
\centering
\small
\setlength{\tabcolsep}{4pt}
\renewcommand{\arraystretch}{1.16}
\begin{tabular}{@{}p{0.19\textwidth}p{0.22\textwidth}p{0.10\textwidth}p{0.26\textwidth}p{0.15\textwidth}@{}}
\toprule
Work & \raggedright Regime / input & Target & \raggedright Quantum resource bound & Depth\\
\midrule
\raggedright Ge~et al. \cite{GeMolnarCirac2016}
 & \raggedright Commuting local $H$; uniform parent gap
 & $|\Gamma_\beta\rangle$
 & \raggedright $N\polylog(N/\varepsilon)$ gates
 & $\polylog(N/\varepsilon)$\\
\raggedright Brand{\~a}o and Kastoryano \cite{BrandaoKastoryano2019}
 & \raggedright Uniform clustering and Markov conditions
 & $\gamma_\Lambda$
 & \raggedright Local-channel construction
 & $O(1)$ channel layers\\
\raggedright Scalet~et al. \cite{ScaletEtAl2025}
 & \raggedright Decaying effective interactions; reconstructed marginals
 & $\gamma_\Lambda$
 & \raggedright $\poly(N,1/\varepsilon)$\newline $\times\exp(O([\log(N/\varepsilon)]^D))$
 & ---\\
\raggedright Rouz\'e~et al. \cite{RouzeFrancaAlhambra2026}
 & \raggedright Local $H$; sufficiently high temperature
 & $\gamma_\Lambda$
 & \raggedright $N\polylog(N/\varepsilon)$ gates
 & ---\\
\raggedright Bergamaschi and Chen \cite{BergamaschiChen2025}
 & \raggedright 1D open chain; fixed finite $\beta$
 & $|\Gamma_\beta\rangle$
 & \raggedright $N\polylog(N/\varepsilon)$ gates
 & $\polylog(N/\varepsilon)$\\
\raggedright \textbf{This work}
 & \raggedright Locality and stability of effective Hamiltonians; high-temperature condition; explicit classical references
 & $|\Gamma_\beta\rangle$
 & \raggedright $N\polylog(N/\varepsilon)$ gates
 & $\polylog(N/\varepsilon)$\\
\midrule
\raggedright Gily\'en~et al. \cite{GilyenEtAl2022}
 & \raggedright State encodings and channel access
 & Petz\newline isometry
 & \raggedright Spectral- and input-cost dependent
 & ---\\
\raggedright \textbf{This work: local step}
 & \raggedright Logarithmic-radius patch; explicit classical references
 & Petz\newline isometry
 & \raggedright $\polylog(N/\varepsilon)$ gates
 & ---\\
\bottomrule
\end{tabular}
\end{table*}

\section{Setup}
\label{sec:setting-summary}

Consider $N$ sites of fixed local dimension $q\ge2$ on a bounded-degree graph
$\Lambda$. Its shortest-path distance is $d$. The balls
$B_r(x)=\{y\in\Lambda:d(x,y)\le r\}$ satisfy
$|B_r(x)|\le c_{\mathrm{vol}}(1+r)^D$ for constants $D$ and $c_{\mathrm{vol}}$, independent of $N$. The Hamiltonian has
bounded finite-range interactions,
\begin{equation}\label{m:eq:thermal-target}
 H_\Omega=\sum_{\varnothing\ne S\subseteq\Omega}h_S,
 \qquad
 \gamma_\Omega=\frac{e^{-\beta H_\Omega}}{\Tr e^{-\beta H_\Omega}},
 \qquad \Omega\subseteq\Lambda.
\end{equation}
The inverse temperature $\beta>0$ and the bounds on interaction
strengths, ranges, and support sizes
are fixed independently of $N$. Access to the microscopic interaction terms $h_S$ is specified in
Assumption~\ref{ass:microscopic-access}.
For $A\subseteq\Omega$, write
${\rho_A^\Omega=\Tr_{\Omega\setminus A}\gamma_\Omega}$ for the reduced
density operator of $\gamma_\Omega$ on $A$. %At finite $\beta$ these reduced density operators are positive definite, also called faithful.
For each region $A$, we introduce a purification register
$\mathcal E_A$, a copy of its Hilbert space with a fixed product basis.
For a density operator $\omega_A$, its canonical purification is
\begin{equation}\label{m:eq:canonical}
 |\Gamma(\omega_A)\rangle
 :=(\omega_A^{1/2}\otimes\I_{\mathcal E_A})
             \sum_b|b\rangle_A|b\rangle_{\mathcal E_A}.
\end{equation}
The vector in Eq.~\eqref{m:eq:canonical} is normalized, and tracing out its
purification register gives $\omega_A$
\cite{DuttaFaulkner2021}. Write
$|\Gamma_\beta\rangle=|\Gamma(\gamma_\Lambda)\rangle$. Our target is a
unitary $\mathscr U_{\mathrm{prep}}$
whose zero-input output approximates
$|\Gamma_\beta\rangle|0\rangle_{\mathsf A}$, where $\mathsf A$ is its work
register. Each site is encoded in $n_q=\lceil\log q/\log2\rceil$ qubits.

The preparation result uses the following assumptions.
Their quantitative formulations are given in the appendices.
\begin{informalassumption}\label{informalassumptions}
    \emph{Locality and stability of effective Hamiltonians.}
For each reduced density operator $\rho_A^\Omega$, we introduce
an effective Hamiltonian $\widetilde H_A^\Omega$ through
\begin{equation}\label{m:eq:effective-Hamiltonian}
 \rho_A^\Omega
 =
 \frac{e^{-\beta\widetilde H_A^\Omega}}
      {\operatorname{Tr}e^{-\beta\widetilde H_A^\Omega}}.
\end{equation}
Let $H_A$ contain the microscopic interaction terms supported
entirely within $A$.
We assume that $\widetilde H_A^\Omega-H_A$ is concentrated near
the boundary between $A$ and the traced region $\Omega\setminus A$,
with contributions decaying exponentially into the interior.
Contributions involving many sites or spanning large distances
are also exponentially suppressed.
At fixed $\Omega$, enlarging $A$ changes interaction terms
far from the added sites only exponentially weakly
(Assumption~\ref{an:ass:effective}).

\emph{High temperature.}
The inverse temperature is sufficiently small relative to the
local interaction strengths, as quantified in
Assumption~\ref{an:ass:strip}.

\emph{Classical references.}
To use these effective Hamiltonians in the circuit construction,
we assume access to coarse classical references satisfying
Assumption~\ref{cr:ass:classical}(i).
Their accuracy is determined by the patch size, independently
of the final recovery tolerance.
The classical preprocessing bound further assumes local generation
of these same references
(Assumption~\ref{cr:ass:classical}(ii)).
\end{informalassumption}

The compiler takes the microscopic Hamiltonian, classical reference
lists, and accuracy parameters, and returns a gate description of
the local recovery unitary.
Its preprocessing, including normalization measurements on separate
test inputs, is completed before the recovery input is supplied.

%We use \emph{compilation} for the circuit-construction procedure that
%takes the microscopic Hamiltonian, the interaction lists for the classical references, and
%accuracy parameters and produces an elementary-gate description of a
%unitary approximating the specified Petz isometry. This procedure,
%called the \emph{compiler}, is independent of the state later supplied
%to the recovery. Its preprocessing includes independent quantum
%experiments for normalization estimates. Microscopic access and finite-precision operations are specified in
%Secs.~\ref{sec:microscopic-model} and~\ref{sec:finite-matrix-model}.

%Unsubscripted operator norms are operator norms; $\|\cdot\|_1$ and
%$\|\cdot\|_2$ denote the trace norm and vector Euclidean norm, respectively.
%The notation $\polylog u$ denotes a polynomial
%of fixed degree in $\log u$; constants and degrees may depend on the fixed
%physical parameters, including $D$, but not on $N$ or the desired accuracy.
%Unions of disjoint regions may be written by juxtaposition, and identities
%on unaffected registers are understood. 
%Table~\ref{tab:notation}, at the
%beginning of the appendices, collects the principal symbols and error budgets.

\section{Theoretical overview}
\label{sec:theoretical-overview}

\begin{figure*}[t]
\centering
\includegraphics[ width=\linewidth]{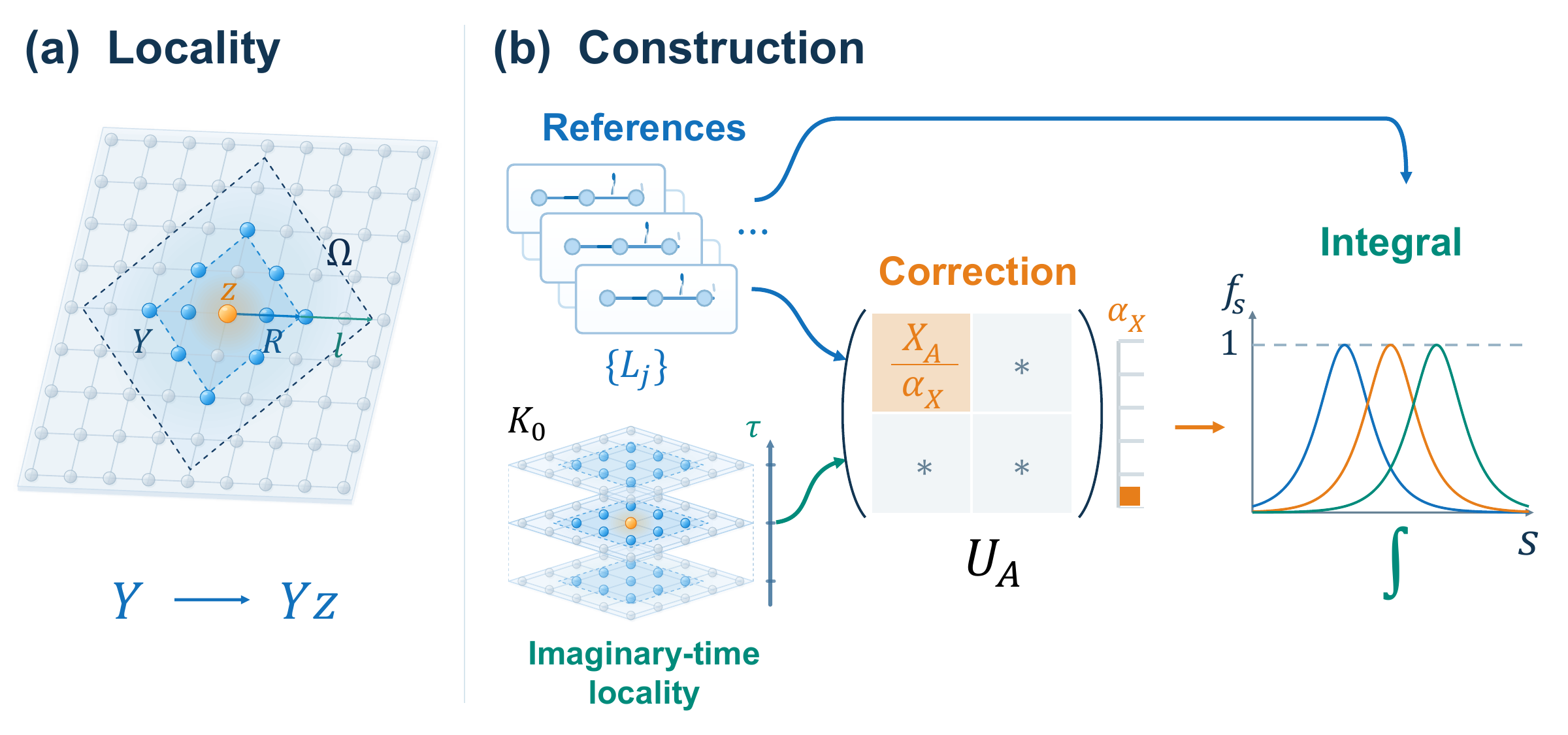}
\caption{
\textbf{Locality and constructive implementation of one-site Petz recovery.}
\textbf{(a)}
The Petz map in Eq.~\eqref{m:eq:petz} restores $z$ (orange) from $Y$ (blue).
The inner and outer dashed boundaries indicate $B_R(z)$ and the
patch $\Omega=B_{R+\ell}(z)$, respectively, where $R$ is the recovery
radius and $\ell$ is the additional buffer.
%Sites in $\Omega\setminus(Yz)$ enter the defining partial trace
%but are not added to the recovery output.
\textbf{(b)}
The cards represent fixed coarse classical references $\{L_j\}$.
Together with the microscopic input $K_0$
in Eq.~\eqref{m:eq:microscopic-input}, they yield the quantum
corrections $X_A$ in Eq.~\eqref{m:eq:reference-factorization},
for $A=Y,Yz$.
The stacked lattices depict local operators remaining near their
initial support during imaginary-time evolution
(Sec.~\ref{sec:overview:corrections}).
%while the exact microscopic input $K_0$ determines the target reduced density operators.
%The stacked lattices schematically illustrate imaginary-time locality
%of initially local operators, with $\tau$ denoting imaginary time.
%This locality enables efficient encodings of exponential products
%and coherent one-site contractions, yielding corrections $X_A$,
%for $A=Y,Yz$, that retain the exact information of the reduced density operators through
%$\rho_A^\Omega\propto e^{-L_A/2}(I_A+X_A)e^{-L_A/2}$.
The orange block of $U_A$ represents $X_A/\alpha_X$;
the scale indicates the small encoding normalization $\alpha_X$.
%The upper blue arrow shows that the same references are reused,
Combining the classical references with the correction encodings, in the
integral representation in Eq.~\eqref{m:eq:integral}, gives an efficient construction of Eq.~\eqref{m:eq:T-definition}.
%The curves illustrate $f_s(x)=\operatorname{sech}(s+x/2)$
%for representative eigenvalues $x$ of $L_A$.
%Here $s$ is the integration variable, and the dashed line marks
%the upper bound $1$.
}
\label{fig:gibbs-overview}
\end{figure*}

\subsection{Main result}

Local recovery specifies where to act; our main constructive result
specifies how to implement that recovery efficiently. %Following the
%recovery-based preparation viewpoint of
%Refs.~\cite{BrandaoKastoryano2019,ScaletEtAl2025}, we first reduce global
%preparation to finite-patch Petz isometries. We then construct circuits
%to encode the marginal residuals and implement the square-root correction
%factors needed for these isometries. Imaginary-time locality is used to
%encode exponential products and one-site contractions, obtaining exact
%marginal residuals from microscopic input and coarse classical references.
%An ordered bounded-function integral turns those residuals into the
%noncommutative square-root factors required by the same Petz operator.
Combining this implementation with the finite-patch reduction yields
the following preparation guarantee.
Figure~\ref{fig:gibbs-overview} illustrates the local construction.

The following summarizes Theorem~\ref{thm:global},
Corollary~\ref{cor:sequential-workspace}, and
Proposition~\ref{prop:classical-budget}.

\begin{informaltheorem}[Gibbs state preparation]
Under the preparation assumptions summarized in Sec. \ref{sec:setting-summary}, input-independent preprocessing
produces, with high probability, a unitary preparing the canonical
purification of an $N$-site Gibbs state to vector error $\varepsilon$.
The circuit has $\operatorname{polylog}(N/\varepsilon)$ depth and
$N\operatorname{polylog}(N/\varepsilon)$ gates and qubits.
With the additional local-reference-generation condition, the total
classical preprocessing cost is $N^{1+o(1)}$ when
$\log(1/\varepsilon)=O(\log N)$.
Sequential execution prepares the mixed Gibbs state with
$\operatorname{polylog}(N/\varepsilon)$ reusable ancillary qubits.
\end{informaltheorem}

The depth bound applies to the preparation unitary and its inverse after preprocessing.
Quantum gate counts include both normalization-estimation trials and state preparation.
%Preprocessing fixes the interaction lists for the classical references, the normalization estimates
%obtained from independent measurements, and finite gate descriptions. The depth bound refers to
%the resulting preparation unitary and its inverse; the quantum gate
%budget also includes the normalization-estimation trials.

\subsection{Reduction to finite-patch recovery}
\label{sec:overview:one-site}
\label{sec:overview:assembly}

A recovery restores one site $z$ from an already prepared region $Y$.
For a patch $\Omega\supseteq Yz$ of $M=|\Omega|$ sites, its Petz factor
and channel are
\begin{equation}\label{m:eq:petz}
 \begin{aligned}
 Q^\Omega_{z|Y}
  &=(\rho^\Omega_{Yz})^{1/2}
             ((\rho^\Omega_Y)^{-1/2}\otimes\I_z),\\
 \PP^\Omega_{z|Y}(\sigma_Y)
  &=Q^\Omega_{z|Y}(\sigma_Y\otimes\I_z)(Q^\Omega_{z|Y})^\dagger.
 \end{aligned}
\end{equation}
Petz recovery here denotes the unrotated map. Sites in
$\Omega\setminus(Yz)$ determine its coefficients through the defining
partial trace; only $z$ is added to the physical input. A purification register
$\mathcal E_z$ in the fixed site basis gives the isometry
\begin{equation}\label{m:eq:isometry}
 V^\Omega_{z|Y}|\psi\rangle
   =\sum_{b=1}^{q}Q^\Omega_{z|Y}
       (|\psi\rangle_Y\otimes|b\rangle_z)|b\rangle_{\mathcal E_z},
\end{equation}
with $V^\dagger V=\I_Y$~\cite{Petz1986Sufficiency,Wilde2015,GilyenEtAl2022}.

Under Assumptions~\ref{an:ass:effective} and~\ref{an:ass:strip},
Theorem~\ref{thm:two-cut} compares the full-system and finite-patch
Petz isometries on purifications of the reduced density operator
$\rho_{WY}^{\Lambda}$.
Take $Y$ to be the previously prepared sites within distance $R$
of $z$, let $W$ contain the remaining prepared sites, and set
$\Omega=B_{R+\ell}(z)$.
For any purification $|\psi\rangle_{WY\mathcal F}$ of
$\rho_{WY}^{\Lambda}$,
\begin{equation}\label{m:eq:locality}
\begin{split}
 &\left\|
 \left[
 \left(V^{\mathrm{full}}-\I_W\otimes V^\Omega_{z|Y}\right)
 \otimes \I_{\mathcal F}
 \right]|\psi\rangle
 \right\|_2\\
 &\quad\le C_{\mathrm{loc}}^Q
 \left[
 e^{-\kappa R}
 +(1+R)^{3D/4}e^{-\kappa\ell/16}
 \right],
\end{split}
\end{equation}
where $V^{\mathrm{full}}=V^\Lambda_{z|WY}$ and $C_{\mathrm{loc}}^Q,\kappa>0$ are independent of the region volumes and the accuracy.
The first term bounds the error from restricting the recovery input
to $Y$.
The second bounds the additional error from defining the recovery
using the Gibbs state on $\Omega$ instead of the full system.
%Under (E1)--(E3) and Assumption~\ref{an:ass:strip},
%Theorem~\ref{thm:two-cut} uses imaginary-time locality to compare this
%specified isometry with the one defined by full-system marginals.
%Taking $Y$ to be the previously prepared sites within distance $R$
%of $z$, $W$ to contain the remaining prepared sites, and
%$\Omega=B_{R+\ell}(z)$, it gives
%\begin{equation}\label{m:eq:locality}
% \|V^{\mathrm{full}}-\I_W\otimes V^\Omega_{z|Y}\|
% \le C_{\mathrm{loc}}^Q(e^{-\kappa R}+e^{-\kappa\ell}),
%\end{equation}
%where $V^{\mathrm{full}}=V^\Lambda_{z|WY}$ and
%$C_{\mathrm{loc}}^Q,\kappa>0$ are volume-independent constants.
%The two terms control the restriction of the retained input and the
%replacement of the Gibbs volume by $\Omega$, respectively. The common
%environment basis allows this estimate to be composed coherently.

For any full-rank density operator $\omega_{Yz}$ with $\omega_Y=\Tr_z\omega_{Yz}$,
Petz composition gives
\begin{equation}\label{m:eq:telescoping}
 (V_{z|Y}\otimes\I_{\mathcal E_Y})|\Gamma(\omega_Y)\rangle
       =|\Gamma(\omega_{Yz})\rangle.
\end{equation}
Adjacent square roots cancel, so the isometries defined by the target Gibbs state prepare the
canonical purification from the empty system in any site order
\cite{Wilde2015,ParzygnatBuscemi2023}. 
%The isometries defined by the exact reduced density operators
%of $\gamma_\Lambda$ prepare the canonical purification in any
%site order.
Replacing them by the finite-patch isometries and telescoping
along this exact preparation sequence, while retaining all
purification registers, gives
%Replacing them by the maps in
%Eq.~\eqref{m:eq:locality}, while retaining all environments, produces
%$|\Psi_{\mathrm{loc}}\rangle$ with
\begin{equation}\label{m:eq:analytic-assembly-error}
\begin{split}
 \bigl\||\Psi_{\mathrm{loc}}\rangle-|\Gamma_\beta\rangle\bigr\|_2
 &\le NC_{\mathrm{loc}}^Q
 \left[
 e^{-\kappa R}
 +(1+R)^{3D/4}e^{-\kappa\ell/16}
 \right].
\end{split}
\end{equation}
Thus $R,\ell=O(\log(N/\varepsilon))$ give the required accuracy, with
$M=O([\log(N/\varepsilon)]^D)$ sites per patch.

Color the graph joining sites at distance at most $2(R+\ell)$. Greedy
coloring partitions the sites into $n_{\mathrm{col}}$ color classes
$\Lambda^{(c)}$, $c=1,\ldots,n_{\mathrm{col}}$. The number of colors satisfies
\begin{equation}\label{m:eq:colors}
 n_{\mathrm{col}}\le\max_x|B_{2(R+\ell)}(x)|
           =O([\log(N/\varepsilon)]^D).
\end{equation}
For $x\in\Lambda^{(c)}$, the input is
$Y_x=\Lambda^{(<c)}\cap B_R(x)$, where
$\Lambda^{(<c)}=\bigcup_{a<c}\Lambda^{(a)}$.
Equal-colored patches are disjoint and their inputs stay fixed within
the layer, permitting parallel recovery. Appendix~\ref{sec:analytic-composition}
proves the composition and scheduling statements. This fixes the
recoveries to be implemented; their elementary-gate cost is the
remaining problem.

\subsection{Factorization using classical references}
\label{sec:overview:references}

The reduction of global preparation to finite-patch Petz isometries
identifies the recovery maps, but does not by itself provide an
efficient circuit implementation. A dense construction of the reduced density operators in Eq.~\eqref{m:eq:petz} can
retain exponential dependence on $M$, and implementing their fractional
powers separately can introduce large spectral normalization factors
\cite{ScaletEtAl2025,GilyenEtAl2022}. Our solution is to change the
operator representation: use coarse, accessible classical references for the
large Gibbs factors, and construct the information they omit as small
quantum-encoded corrections.

For $A=Y,Yz$, let $L_A$ be a fixed classical reference for the
traceless part of $\beta\widetilde H_A^\Omega$. We seek a Hermitian quantum correction
$X_A$ such that
\begin{equation}\label{m:eq:reference-factorization}
 \rho_A^\Omega\propto e^{-L_A/2}(\I_A+X_A)e^{-L_A/2}.
\end{equation}
An encoding represents $X_A/\alpha_X$ as a block of a unitary,
where $\alpha_X>0$ is the known encoding normalization.
The correction is constructed from $H_\Omega$ and encoded to the
precision required for recovery, while $L_A$ remains fixed.
The two accuracy scales therefore use different neighborhoods: the
recovery patch $\Omega$ controls the accumulated recovery error, while the
neighborhood used to construct the classical references need only
be large enough to make the correction small. Under the
local-generation guarantee, radius $\ell_{\mathrm{cl}}=O(\log(M+1))$
suffices for the precision of the classical references
(Proposition~\ref{prop:classical-budget}).

For Hermitian $L,X$ with $\|X\|<1$, define the square-root factor
\begin{equation}\label{m:eq:T-definition}
 T(L,X):=e^{-L/2}
     \bigl[e^{-L/2}(\I+X)e^{-L/2}\bigr]^{-1/2},
\end{equation}
and put $T_A=T(L_A,X_A)$. A positive scalar $a$ fixes the relative
normalization: the partial trace over $z$ of the unnormalized $Yz$
operator in Eq.~\eqref{m:eq:reference-factorization} is $aq$ times the
$Y$ operator. Then
\begin{equation}\label{m:eq:central-factorization}
 Q^\Omega_{z|Y}
 =\frac1{\sqrt{aq}}\,
 T_{Yz}^{-1}e^{-L_{Yz}/2}
 (e^{L_Y/2}\otimes\I_z)(T_Y\otimes\I_z).
\end{equation}
The middle exponential product $e^{-L_{Yz}/2}
 (e^{L_Y/2}\otimes\I_z)$ uses the classical references,
while the two outer factors $T_{Yz}^{-1}$ and $T_{Y}$ incorporate the quantum corrections.
%This is an ordered product of operators on $Yz$. The middle
%product uses the classical references; the outer square-root factors correct its
%output and input to the same Petz operator
%(Proposition~\ref{cr:prop:V}). The following construction addresses two
%remaining tasks: obtaining $X_Y,X_{Yz}$ from microscopic data by
%one-site contractions, and converting them into $T_Y,T_{Yz}^{-1}$
%without separately implementing the inverse square root in
%Eq.~\eqref{m:eq:T-definition}.

\subsection{Constructing the quantum corrections}
\label{sec:overview:corrections}
%\label{sec:overview:endpoints}

The exact starting point is the traceless part of $\beta H_\Omega$:
\begin{equation}\label{m:eq:microscopic-input}
 K_0=\beta H_\Omega-
      \frac{\Tr(\beta H_\Omega)}{q^M}\I_\Omega.
\end{equation}
Choose $A_0=\Omega$ and $A_j=A_{j-1}\setminus B_j$, deleting one site
at a time: remove $\Omega\setminus(Yz)$ first and $B_m=\{z\}$ last.
Thus $A_{m-1}=Yz$, $A_m=Y$, and $m\le M$,
where $M=|\Omega|$ is the recovery patch size.
For each $j$, let $K_j$ denote the traceless part of
$-\log\rho_{A_j}^\Omega$.
For each $j\ge1$, let $L_j$ be a fixed coarse classical reference
for $K_j$.
%Only $j\ge1$ requires a coarse classical reference $L_j$.
%Assumption~\ref{cr:ass:classical}(i) and
%Lemma~\ref{cr:lem:sparsify} supply these fixed, compressed interaction lists with
%$\|L_j-K_j\|\le2e_M$, where $e_M$ is inverse-polynomial in $M$.
The exact input $K_0$ and the coarse classical references therefore have distinct
roles: the former fixes the target, while the latter make its
representation inexpensive to encode.

For each deletion step $j$, define the following operator on $A_{j-1}$,
with identities on smaller regions understood:
\begin{equation}\label{m:eq:evolution}
 \begin{aligned}
 E_1(t)&=e^{-tK_0/2}(e^{tL_1/2}\otimes\I_{B_1}),\\
 E_j(t)&=e^{-tL_{j-1}/2}(e^{tL_j/2}\otimes\I_{B_j}),
       \quad 2\le j\le m.
 \end{aligned}
\end{equation}
Write $E_j=E_j(1)$ and embed all factors in $\Omega$ when taking their
product. Adjacent exponentials of the classical references cancel exactly, giving
$E_1\cdots E_j=e^{-K_0/2}
(e^{L_j/2}\otimes\I_{\Omega\setminus A_j})$.
Writing $\E_B=q^{-|B|}\Tr_B$ for the normalized partial trace,
we obtain, for the positive normalizations $a_i$ fixed below,
\begin{equation}\label{m:eq:square-correction}
 \I_{A_j}+X_j=
 \frac{\E_{\Omega\setminus A_j}
  \!\left[(E_1\cdots E_j)^\dagger(E_1\cdots E_j)\right]}
      {\prod_{i=1}^j a_i}.
\end{equation}
This identity retains the exact microscopic Gibbs information rather
than replacing it by the Gibbs state defined by the coarse classical reference. The normalized partial
trace splits as
$\E_{\Omega\setminus A_j}=\E_{B_j}\circ\cdots\circ\E_{B_1}$,
so the corrections satisfy
\begin{equation}\label{m:eq:contraction-step}
 \I_{A_j}+X_j=\frac1{a_j}\E_{B_j}
 \!\left[E_j^\dagger(\I_{A_{j-1}}+X_{j-1})E_j\right],
\end{equation}
where we define $X_0=0$.
%Only its two endpoints are needed:
The quantum corrections are given by
\begin{equation}\label{m:eq:endpoint-corrections}
\begin{aligned}
 X_{Yz}&=X_{m-1},\\
 X_Y&=X_m,
 \end{aligned}
\end{equation}
and the normalization is $a=a_m$.

The first constructive ingredient makes both the exponential products and
their contractions gate efficient. For $j\ge2$,
\begin{equation}\label{m:eq:evolution-generator}
 E_j'(t)=-\tfrac12E_j(t)e^{-tL_j/2}
                 (L_{j-1}-L_j)e^{tL_j/2},
\end{equation}
with $K_0$ replacing $L_{j-1}$ in the first step.
The generator in Eq.~\eqref{m:eq:evolution-generator} is the imaginary-time transform of
the change $L_{j-1}-L_j$ associated with deleting one site.
Imaginary-time locality~\cite{BluhmCapelPerezHernandez2025}, together with the effective-Hamiltonian
assumptions, controls its norm independently of the patch volume.
We encode the generator through its nested-commutator expansion
using a linear combination of unitaries (LCU), amplify the encoding,
and implement the time-ordered evolution to obtain $E_j$
(Proposition~\ref{cr:prop:E}).
%The generator involves the difference $L_{j-1}-L_j$ between the
%classical references before and after removing one site.
%Although the individual reference norms can grow with the patch size $M$,
%imaginary-time locality and stability of the classical references
%bound the generator norm independently of $M$.
%We turn the associated nested-commutator expansion into a coherent
%linear combination of unitaries (LCU) circuit by selecting overlapping
%interaction supports, and amplify
%its initial $O(M)$ encoding normalization to a constant before applying
%the time-ordered construction. This gives each $E_j$ with constant
%normalization and polynomial gate cost, not the exponentially large
%LCU weight of separately treated extensive factors
%(Proposition~\ref{cr:prop:E}). The analytic estimates are established
%tools~\cite{BluhmCapelPerezHernandez2025}; the constructive step is to
%realize their controlled expansions from the explicit interaction lists, using
%LCU, QSVT, and truncated-Dyson primitives
%\cite{GilyenEtAl2019,LowWiebe2018,KieferovaSchererBerry2019,BerryCosta2024}.

To encode the one-site contractions in Eq.~\eqref{m:eq:contraction-step}, prepare a
maximally entangled pair between the deleted site and an auxiliary
register, apply $E_j$, and rescale by $a_j^{-1/2}$.
Let $W_j$ denote the resulting auxiliary map.
The positive scalars $a_j$ are determined during preprocessing
(Lemma~\ref{lem:normalization-estimation}).
Composing the independently compiled maps gives encodings of
the quantum corrections in Eq.~\eqref{m:eq:endpoint-corrections}.
The normalization is controlled throughout the composition,
yielding small constant encoding normalization for $X_Y$ and
$X_{Yz}$ (Proposition~\ref{cr:prop:residual}).

%To compose the contractions in Eq.~\eqref{m:eq:contraction-step},
%a maximally entangled pair between the site and an auxiliary register implements each normalized
%partial trace coherently. Let $W_j$ denote the map obtained by
%appending this pair, applying $E_j$ to the physical registers, and
%rescaling by $a_j^{-1/2}$.
%Independent normalization experiments determine $a_j$ so that $W_j$
%has isometry defect
%$O(1/m)$ and encoding normalization $1+O(1/m)$
%(Lemma~\ref{lem:normalization-estimation}). By independent normalization
%experiments, we mean preprocessing trials that estimate $a_j$ from the
%success probability of the encoded map before rescaling, using freshly
%prepared maximally entangled test states. These trials use inputs
%independent of the state later supplied to the recovery.
%The normalization of the flat product of the $W_j$ maps therefore
%stays bounded. This contraction-aware implementation,
%shown in Fig.~\ref{fig:gibbs-overview}(b), avoids paying for the norm of
%the uncontracted full-space product in Eq.~\eqref{m:eq:square-correction}.
%Proposition~\ref{cr:prop:residual} proves the operator identities
%and supplies the endpoint encodings with $\|X_A\|<1/128$ and
%normalization $1/32$. Here an encoding is a projected unitary block
%representing the target divided by a known normalization; controlling
%that normalization is essential to the cost.

%The circuit uses independently compiled one-site factors, not recursive
%calls to earlier $X_j$ or $T_j$ circuits. For $m=1$, we use the exact $K_0$
%on $Yz$, and $X_{Yz}=0$.

\subsection{Integral representation and local implementation}
\label{sec:overview:local-implementation}
%\label{sec:overview:integral}

%An encoding of the quantum correction does not yet implement the
%Petz square-root factors. The classical reference $L$ and quantum correction $X$ need not
%commute, so $T(L,X)$ is generally not $(\I+X)^{-1/2}$.
%The second constructive ingredient implements the ordered combination

We implement $T(L,X)$  using an integral representation without separately implementing the inverse square root
appearing in Eq.~\eqref{m:eq:T-definition}.
%to avoid directly constructing its potentially ill-conditioned inverse square root. 
The fractional-power resolvent
representation~\cite{Balakrishnan1960} yields
\begin{equation}\label{m:eq:integral}
 \begin{aligned}
 T(L,X)&=\frac1\pi\int_{\mathbb R}
       (\I+P_sX)^{-1}\sech(s+L/2)\,ds,\\
 P_s&=(\I+e^{2s+L})^{-1}.
 \end{aligned}
\end{equation}
For real $s$, both $P_s$ and $\sech(s+L/2)$ have operator norm at most one.
We encode these functions using quantum singular value transformation
(QSVT) and construct $(1+P_sX)^{-1}$ through a Neumann expansion.
The small encoding normalization of $X$ keeps the coefficient sum
of this expansion bounded.
A quadrature rule with positive weights then combines these
operators into an encoding of $T$ (Theorem~\ref{lem:inverse-circuit}).
%These bounded functions, together with the small encoding
%normalization of $X$, make the integral efficiently implementable (Theorem~\ref{lem:inverse-circuit}).
%The small encoding
%normalization of $X$, not only its small operator norm, keeps the
%Neumann-series coefficient sum for $(\I+P_sX)^{-1}$ bounded.
%Positive quadrature and QSVT thus implement $T$ from the encodings of
%$L$ and $X$ with normalization $O(M+1)$ and cost
%$\poly(M,\log(1/\zeta))$ at operator error $\zeta$.
%The bound depends on fixed convergence margins, not on the smallest
%eigenvalue of $\rho_A^\Omega$ 
%(Theorem~\ref{lem:inverse-circuit}). 
%This is what makes the
%factorization computationally useful even though the exact effective
%Hamiltonian is unavailable at the final precision.
%The encodings of $L$ and $X$ also yield an encoding of $T^{-1}$, because
%The definition of $T$ in Eq.~\eqref{m:eq:T-definition} also gives
The same encodings provide $T^{-1}$, because
\begin{equation}\label{m:eq:condition}
 T^{-1}=T^\dagger(\I+X).
\end{equation}
%Amplification reduces its normalization factor to a constant independent
%of the patch size $M$ (Lemma~\ref{cr:lem:T}). Inserting the two endpoint
%square-root factors into Eq.~\eqref{m:eq:central-factorization} completes the
%specified Petz factor. Appending the purification register in the fixed site basis as in
%Eq.~\eqref{m:eq:isometry} and amplifying the complete encoded isometry
%then gives the recovery unitary. The imaginary-time circuits and
%coherent one-site contractions encode the quantum corrections $X_Y$ and
%$X_{Yz}$. 
%Circuits based on the bounded integral representation then use
%these encodings to implement the square-root factors required
%for the Petz isometry.
We combine $T_Y$ and $T_{Yz}^{-1}$ with the reference exponential
product in Eq.~\eqref{m:eq:central-factorization}, prepare the purification register as in
Eq.~\eqref{m:eq:isometry}, and amplify the encoded isometry to obtain $U_m$.

For local channel tolerance $\eta\in(0,1/4)$ and preprocessing failure
budget $p_{\mathrm{fail}}\in(0,1/4)$, write the initialized action as
\begin{equation}\label{m:eq:initialized}
 \widetilde V_m|\psi\rangle
 :=\mathscr U_m(|\psi\rangle_Y|0\rangle_{z\mathcal E_z\mathsf A_m}),
\end{equation}
where $\mathcal A_m$ denotes its work register.
Theorem~\ref{cr:thm:local-main} gives, with probability at least
$1-p_{\mathrm{fail}}$,
\begin{equation}\label{m:eq:local-guarantee}
 \|\widetilde V_m-V^\Omega_{z|Y}\otimes|0\rangle_{\mathsf A_m}\|
 \le\eta/2,
\end{equation}
including work-register leakage. %The induced channel has diamond-norm
%error at most $\eta$. 
Its gate count, including independent normalization
experiments, and workspace are bounded by
\begin{equation}\label{m:eq:local-cost}
 \poly(M,m,\log N,\log(1/\eta),\log(1/p_{\mathrm{fail}})).
\end{equation}
Classical compilation has the same polynomial dependence, in addition
to the cost of obtaining and reading the classical reference lists.
Algorithm~\ref{alg:local-petz} summarizes the construction; Appendix~\ref{sec:classical-construction}  gives the
detailed implementation and resource bounds.

%Classical compilation has the same type of bound in addition to the
%cost of generating and scanning the interaction lists for the classical references.
%Classical preprocessing reduces each interaction list for the classical
%references to polynomially many terms through truncation and weighted
%sampling, while preserving the required accuracy. Each retained term acts
%on $O(\log(M+1))$ sites and therefore admits an encoding with polynomial
%gate cost in $M$. Selecting and applying these encodings under the control
%of a register in a superposition of term labels also has polynomial gate
%cost in $M$.
%The fixed number of nested circuit transformations preserves polynomial
%cost in the precision logarithms. Every preprocessing outcome fixes a
%unitary before it is applied to the recovery input.
%Algorithm~\ref{alg:local-petz} summarizes the construction; the
%projected maps and finite-precision implementation are given in
%Appendix~\ref{sec:classical-construction} and Algorithm~\ref{alg:local-petz-detailed}.

\begin{algorithm}
\caption{Outline of one-site Petz compilation}
\label{alg:local-petz}
\begin{algorithmic}[1]
\Require Recovery $Y\to Yz$ in $\Omega$, microscopic Hamiltonian and
 classical references satisfying Theorem~\ref{cr:thm:local-main};
 $\eta,p_{\mathrm{fail}}\in(0,1/4)$.
\Ensure A unitary satisfying Eq.~\eqref{m:eq:local-guarantee},
 except with probability at most $p_{\mathrm{fail}}$.
\State Choose a sequence $A_0,\ldots,A_m$ by deleting one site
at a time from $A_0=\Omega$: remove all sites in
$\Omega\setminus(Yz)$ first and $z$ last, so that
$A_{m-1}=Yz$ and $A_m=Y$.
\State Preprocess the classical reference data, keeping $K_0$ exact.
\For{$j=1,\ldots,m$}
 \State Independently encode $E_j$ in Eq.~\eqref{m:eq:evolution}
 using the imaginary-time construction based on Eq.~\eqref{m:eq:evolution-generator}.
 \State Determine $a_j$ during preprocessing and encode the
         auxiliary map $W_j$ for the contraction in Eq.~\eqref{m:eq:contraction-step}.
\EndFor
\State Use a product of these maps to encode only $X_{Yz},X_Y$
 in Eq.~\eqref{m:eq:endpoint-corrections}.
\State Construct $T_Y,T_{Yz}^{-1}$ using
 Eqs.~\eqref{m:eq:integral}--\eqref{m:eq:condition}.
\State Use Eq.~\eqref{m:eq:central-factorization} to encode the isometry
 in Eq.~\eqref{m:eq:isometry}, and amplify its initialized action;
 \Return $\mathscr U_m$.
\end{algorithmic}
\end{algorithm}

\subsection{Global preparation and resources}
\label{sec:parallel-method}
\label{sec:overview:global-resources}

%The local circuits complete the finite-patch reduction of
%Sec.~\ref{sec:overview:assembly}. Each circuit uses its own purification
%and work registers, so circuits in the same color class remain
%parallel after compilation. 
Let $U_x$ be the local recovery unitary for site $x$.
Circuits on disjoint patches can be run in parallel.
With the radii and local error budgets chosen in Theorem~\ref{thm:global},
their composition prepares the canonical purification with
 error at most $\varepsilon$ and success probability
$1-O(\varepsilon)$.
%The output averaged over preprocessing has trace-norm error
%at most $\varepsilon$.
%
%With the radii in
%Eq.~\eqref{eq:global-radii} and local budgets
%$\eta_x=\varepsilon/(4N)$ and
%$p_{\mathrm{fail},x}=\varepsilon/(8N)$, locality and implementation
%errors add over sites, and a union bound controls preprocessing
%failures. Theorem~\ref{thm:global} gives a fixed unitary preparing
%$|\Gamma_\beta\rangle$ to vector error at most $\varepsilon/4$ with
%probability at least $1-\varepsilon/8$, and trace-norm error at most
%$\varepsilon$ after averaging over preprocessing.
%
Both local precision logarithms are $O(\log(N/\varepsilon))$.
Substituting $m\le M=O([\log(N/\varepsilon)]^D)$ into
Eq.~\eqref{m:eq:local-cost} gives polylogarithmic cost per patch.
Depth multiplies by the color count in Eq.~\eqref{m:eq:colors},
whereas gates and separately allocated registers add over sites:
\begin{equation}\label{m:eq:quantum-cost}
\begin{aligned}
 \operatorname{depth}(\mathscr U_{\mathrm{prep}})
 &\le n_{\mathrm{col}}\max_x\operatorname{depth}(\mathscr U_x)
   =\polylog(N/\varepsilon). %,\\
% G_{\mathrm Q},\ n_{\mathrm{qubits}}
% &\le N\polylog(N/\varepsilon).
\end{aligned}
\end{equation}
The total gate count and the number of system, purification,
and work qubits are bounded by
$N\operatorname{polylog}(N/\varepsilon)$.
%Here, $\mathscr U_{\mathrm{prep}}$ is the global preparation unitary after preprocessing,
%$G_{\mathrm Q}$ is the total number of elementary quantum gates, 
%and
%$n_{\mathrm{qubits}}$ is the total number of qubits, 
%including the $N\polylog(N/\varepsilon)$ number of qubits of the total
%system, purification, and work registers.
%The gate budget includes normalization experiments; the 
%The depth refers to
%the unitary after preprocessing. 
Under
Assumption~\ref{cr:ass:classical}(ii), the estimate in
Proposition~\ref{prop:classical-budget} gives a classical preprocessing cost of
$N^{o(1)}$ per patch and $N^{1+o(1)}$ overall when
$\log(1/\varepsilon)=O(\log N)$. 
%Sequential execution of the same
%circuits instead prepares the mixed Gibbs state using
%$Nn_q+\polylog(N/\varepsilon)$ qubits, with
%$N\polylog(N/\varepsilon)$ gates and depth
%(Corollary~\ref{cor:sequential-workspace}). 
Appendix~\ref{sec:global-complexity}
collects the global construction and its error and resource bounds.

\section{Applications}
\label{sec:application}

\subsection{Thermal expectation values}
\label{sec:thermal-expectation-application}
Thermal observables such as local energies and equilibrium
correlations can be estimated from the prepared Gibbs state.
For estimating multiple observables, Theorem~\ref{thm:global} supplies a
preparation unitary $U_{\mathrm{prep}}$ and its inverse that
can be used in the methods of Refs.~\cite{HugginsEtAl2022,WadaYamamotoYoshioka2025}, together with the observable operations
assumed in those works.

\subsection{One-dimensional fermionic systems}
\label{sec:fermionic-application}

For parity-preserving, finite-range fermionic Hamiltonians on an open
one-dimensional chain with a fixed number of modes per site, the
Jordan--Wigner transformation gives a finite-range qubit Hamiltonian and
identifies the normalized Gibbs states in the occupation and computational
bases \cite{VerstraeteCirac2005}. Applying our construction
to the mapped Hamiltonian therefore gives a unitary preparing an encoded
purification of the fermionic Gibbs state whenever the effective-interaction
and input-access hypotheses of Theorem~\ref{thm:global} hold. 

\section{Discussion}
\label{sec:comparison}
\label{sec:discussion-reference-eth}

\subsection{Temperature and quantum advantage}
\label{sec:discussion:temperature-advantage}

A central question is whether the same recovery construction
remains efficient at lower temperatures.
The present high-temperature condition provides uniform control
of the spread and norm growth of local operators under
imaginary-time evolution.
Model-specific bounds may extend this temperature range, provided
the effective-Hamiltonian structure and classical access required
by the construction can also be maintained.

%A central methodological question is how far the same recovery
%construction can be extended towards lower temperatures. The
%high-temperature condition in Assumption~\ref{an:ass:strip} is sufficient
%for the imaginary-time estimates, rather than a necessary condition for
%efficient preparation. Model-specific locality bounds could replace this
%uniform condition by controlling the norms and spatial perturbations of
%the relevant effective-Hamiltonian evolutions
%\cite{PerezGarciaPerezHernandez2023,BluhmCapelPerezHernandez2025}.
%The corresponding estimates for the stability of the classical references,
%encoding normalization, and truncation must remain efficient. Turning such diagnostics into guarantees requires %uniform
%bounds on the uncomputed tails. This offers a model-dependent route to
%extending the present construction without identifying its numerical
%high-temperature condition with a physical or computational boundary.

An important motivation is to reach a regime with quantum sampling
advantage. At sufficiently high temperature, local Gibbs states admit
efficiently samplable mixtures of product states
\cite{BakshiLiuMoitraTang2024,PuttermanZlokapaCotler2026}. Quantum sampling
advantages at constant temperature are nevertheless known for specific
local Hamiltonian families under stated complexity assumptions
\cite{RajakumarWatson2026}. A natural target is therefore an intermediate
temperature regime in which local recovery remains efficiently
implementable while sampling the corresponding measurement distribution
is classically hard. Establishing this overlap for a Hamiltonian family
satisfying the required recovery and input-access conditions would turn
the present construction into a model-specific route to quantum advantage.

\subsection{Correlated initial states and mixed-state phases}
\label{sec:discussion:mixed-state-phases}

A different extension changes the initial state rather than only the
estimates used to justify recovery. Mixed-state phase frameworks relate
local reversibility to uniform exponential CMI decay, expressed by a
finite Markov length: under such bounds, a suitable local dissipative
path admits quasi-local reversal \cite{SangHsieh2025}, and locally reversible channel circuits
provide a refined notion of phase equivalence
\cite{SangEtAl2025LocalReversibility}. This motivates implementing the
recovery operations along such paths with elementary-gate cost polynomial
in their support size.

For low-temperature families carrying long-range logical information,
one could begin with an efficiently preparable quantum state that already
contains the relevant topological correlations. Local recovery would then
adjust thermal excitations while retaining that information, rather than
create it from a product state. Equilibration within logical sectors and
circuit descriptions of finite-temperature phases provide related
settings \cite{BergamaschiGheissariLiu2025,MaKhemaniSang2025}.
Extending the present approach would require conditional-correlation,
locality, and implementation estimates for the relevant regions and
intermediate states along the chosen path, together with control of
the sector weights in the desired purification. 

\subsection{Optimization and learning}
\label{sec:discussion:optimization-learning}

Thermal expectation values are used in quantum optimization and statistical
inference, including Gibbs-based semidefinite programming
\cite{BrandaoSvore2016} and local quantum Gibbs-model learning
\cite{AnshuEtAl2020Learning,CoopmansBenedetti2024}. In the latter,
parameter gradients are determined by differences between prescribed
expectations and those of the current model. The preparation circuit,
combined with the estimation method in
Sec.~\ref{sec:thermal-expectation-application}, supplies the model
expectations needed for these updates. This gives a route to
quantum-assisted inference for local Gibbs models whose parameter
trajectories remain uniformly within the regime of our preparation
theorem.

\subsection{Fermionic encodings in higher dimensions}
\label{sec:discussion:fermionic-encodings}

Beyond the one-dimensional application, locality-preserving
fermion-to-qubit encodings offer a starting point for higher-dimensional
extensions \cite{VerstraeteCirac2005,DerbyEtAl2021}. Such encodings can
introduce a physical code subspace. If its projector $P$ commutes with
the encoded Hamiltonian $H_{\mathrm q}$, the encoded thermal state is
$Pe^{-\beta H_{\mathrm q}}P/\Tr(Pe^{-\beta H_{\mathrm q}})$, rather than
the unconstrained Gibbs state on all qubits. The auxiliary constraints
must therefore be incorporated into the preparation; they do not become
irrelevant when the fermionic state is hot \cite{RamkumarEtAl2025}.
Promising directions are recovery maps that preserve the physical
constraint sector and formulations directly in fermionic local algebras,
with the effective-interaction and circuit-access bounds established for
those settings.

\subsection{Geometry and connected retained regions}
\label{sec:discussion:connected-regions}

Sequential preparation can use connected retained regions on
suitable geometries, reducing the reference data required by the
local compiler.
The global proof nevertheless uses uniform estimates for the wider
region family entering local indistinguishability.
The present theorem therefore retains its region-uniform
effective-Hamiltonian assumptions.

\section*{Use of AI tools.}
In this work, the idea of combining fixed coarse classical
effective-Hamiltonian references with quantum corrections derived
from the microscopic Hamiltonian to implement local Petz recovery
originated with the authors.
OpenAI models GPT-5.6-sol, GPT-5.6-luna, GPT-5.6-terra,
GPT-5.6-sol-wm, and GPT-6-astra, accessed through ChatGPT
in August and September 2026, were used to explore proof approaches,
check intermediate mathematical arguments, error estimates, and
computational resource bounds, and assist with literature searches.
They were also used extensively to draft and revise the entire
manuscript, including the mathematical proofs.
The mathematical arguments and proofs were independently verified
by the authors.
The authors take sole responsibility for the content and correctness
of the final manuscript.

\section*{Acknowledgments}
H.H. was supported by MEXT Quantum Leap Flagship Program (MEXT Q-LEAP) Grant Number JPMXS0120319794, JST PRESTO, Japan, Grant Number JPMJPR25F7, and JSPS KAKENHI Grant No. JP24K16979.
A. I. was supported by the Center of Innovations for Sustainable Quantum AI
(JST Grant No. JPMJPF2221) 
and  JSPS KAKENHI (No. 25K17308).
N.Y. is supported by JST Grant Number JPMJPF2221, JST CREST Grant Number JPMJCR23I4, IBM Quantum, Google Quantum AI, JST ASPIRE Grant Number JPMJAP2316, JST ERATO Grant Number JPMJER2302, JST [Moonshot R\&D] [Grant Number JPMJMS256J], and Institute of AI and Beyond of the University of Tokyo.

\bibliography{bib}

\onecolumngrid
\clearpage
\appendix
\renewcommand{\theequation}{\Alph{section}.\arabic{equation}}
\section*{Notation}
\label{sec:notation-guide}
Table~\ref{tab:notation} summarizes the principal notation for the preparation
argument. Local proof constants are defined where they are used. Physical parameters and locality bounds are
fixed independently of $N$ and the desired accuracy.

\begingroup
\normalsize
\setlength{\tabcolsep}{5pt}
\renewcommand{\arraystretch}{1.10}
\makeatletter\def\@captype{table}\makeatother
\caption{Principal notation. Region and step indices specify members of the
same family. Encoding normalization and approximation error are separate
quantities.}
\label{tab:notation}
\noindent\begin{tabular}{@{}p{0.265\textwidth}p{0.565\textwidth}p{\dimexpr0.17\textwidth-20pt\relax}@{}}
\toprule
\raggedright Symbol & \raggedright Meaning & \raggedright Definition\tabularnewline
\midrule
\raggedright $\Lambda$, $\Omega$; $N$, $M$ & \raggedright Full lattice and recovery patch; their site counts $N=|\Lambda|$, $M=|\Omega|$. & \raggedright Def.~\ref{an:def:geometry}\tabularnewline
\raggedright $D$, $q$, $\beta$ & \raggedright Fixed spatial dimension, local Hilbert-space dimension, and inverse temperature. & \raggedright Def.~\ref{an:def:geometry}\tabularnewline
\raggedright $Y$, $z$; $A_j$, $B_j$, $m$ & \raggedright Retained input and restored site; internal retained regions, deleted blocks, and number of coefficient steps. & \raggedright Def.~\ref{cr:def:chain}\tabularnewline
\raggedright $R$, $\ell$, $\ell_{\mathrm{cl}}$ & \raggedright Recovery radius, additional patch buffer, and classical reference radius; the quantum patch has radius $R+\ell$. & \raggedright Eqs.~\eqref{eq:global-radii},\newline\eqref{eq:classical-radius}\tabularnewline
\raggedright $H_\Omega$, $\gamma_\Omega$, $\rho_A^\Omega$ & \raggedright Microscopic Hamiltonian, its Gibbs state, and its marginal on $A$; $\rho_j=\rho_{A_j}^\Omega$. & \raggedright Eqs.~\eqref{an:eq:gibbs},\newline\eqref{cr:eq:chain}\tabularnewline
\raggedright $\widetilde H_A^\Omega$, $K_A^\circ$, $K_A$ & \raggedright Exact effective Hamiltonian, $\beta\widetilde H_A^\Omega$, and its centered version. & \raggedright Eqs.~\eqref{an:eq:heff},\newline\eqref{an:eq:k-circ}\tabularnewline
\raggedright $\overline L_j$, $L_j$ & \raggedright Supplied and compressed classical references for $j\ge1$; $K_0$ is the exact microscopic input ($L_0:=K_0$ in Appendix~C). & \raggedright Assump.~\ref{cr:ass:classical};\newline Lemma~\ref{cr:lem:sparsify}\tabularnewline
\raggedright $Q$, $\mathcal P$, $V$ & \raggedright Ordinary Petz factor, channel, and specified isometry; subscripts specify the restored and retained regions. & \raggedright Def.~\ref{an:def:petz}\tabularnewline
\raggedright $|\Gamma_\beta\rangle$, $\mathcal E_A$, $\mathsf A$ & \raggedright Target canonical purification, fixed-basis reference copy of $A$, and work register. & \raggedright Def.~\ref{def:canonical-purification};\newline Def.~\ref{def:global-algorithm}\tabularnewline
\raggedright $\mathscr U_x$, $\mathscr U_{\mathrm{prep}}$ & \raggedright Compiled one-site unitary and global preparation unitary. & \raggedright Def.~\ref{def:global-algorithm}\tabularnewline
\raggedright $E_j(t)$, $G_j(t)$ & \raggedright Square exponential product and its generator, $E_j^{\prime}=-E_jG_j$; the first factor uses exact $K_0$. & \raggedright Eqs.~\eqref{cr:eq:E},\newline\eqref{cr:eq:generator}\tabularnewline
\raggedright $W_j$, $\mathbb W_j$ & \raggedright Independently compiled one-step map and its cumulative product, updated along the deletion chain. & \raggedright Eqs.~\eqref{cr:eq:W-S},\newline\eqref{cr:eq:correction-product}\tabularnewline
\raggedright $X_A$, $T_A$; $A=Y,Yz$ & \raggedright Endpoint quantum correction and ordered square-root factor; $X_A$ is encoded via one-site contractions of the microscopic input. & \raggedright Eqs.~\eqref{m:eq:endpoint-corrections},\newline\eqref{m:eq:integral}\tabularnewline
\raggedright $a_j$; $\alpha$ & \raggedright Stored scalar normalization; normalization of a projected operator encoding, distinct from its operator norm. & \raggedright Lemma~\ref{lem:normalization-estimation};\newline Def.~\ref{def:encoding}\tabularnewline
\raggedright $\varepsilon$, $\eta$ & \raggedright Global preparation tolerance and local channel tolerance; local isometry error is at most $\eta/2$. & \raggedright Thms.~\ref{thm:global}, \ref{cr:thm:local-main}\tabularnewline
\raggedright $p_{\mathrm{fail}}$, $\zeta$ & \raggedright Local preprocessing failure budget and internal encoding or synthesis tolerance. & \raggedright Thm.~\ref{cr:thm:local-main};\newline Def.~\ref{def:encoding}\tabularnewline
\raggedright $r_{\mathrm{ref}}$, $e_M$ & \raggedright Supplied reference error and allowed inverse-polynomial threshold, independent of the final $\eta$. & \raggedright Eqs.~\eqref{cr:eq:reference-error},\newline\eqref{cr:eq:theta-p}\tabularnewline
\raggedright $w_{\lambda,\mu}(S)$, $\|\Phi\|_{\lambda,\mu}$ & \raggedright Exponential support-size and diameter weight, and the corresponding per-site interaction norm. & \raggedright Eq.~\eqref{an:eq:weights}\tabularnewline
\raggedright $\mu_0,\mu_1,\mu_2$, $\mu_*$; $\kappa$ & \raggedright Spatial decay rates in the hypotheses and the reduced rate used for recovery errors. & \raggedright (E1)--(E2);\newline Eq.~\eqref{an:eq:parameters}\tabularnewline
\raggedright $J_H$, $C_\beta$, $g_1^\beta$ & \raggedright Microscopic strength, boundary-correction strength, and retained-region stability bound. & \raggedright Eq.~\eqref{an:eq:jh};\newline (E1)--(E2)\tabularnewline
\raggedright $\Delta_0$, $J_{\mathrm{ref}}$, $J_*$ & \raggedright Bounds for exact interactions, supplied lists, and the common bound after compression; all are below $\lambda$. & \raggedright Eqs.~\eqref{an:eq:strip},\newline\eqref{cr:eq:reference-strength}, \eqref{cr:eq:margins}\tabularnewline
\raggedright $\mathcal C_{\mathrm{ref}}(\Omega)$, $G_{\mathrm Q}$ & \raggedright Classical reference-generation cost per patch and total quantum gates, including normalization trials. & \raggedright Assump.~\ref{cr:ass:classical};\newline Thm.~\ref{thm:global}\tabularnewline
\bottomrule
\end{tabular}
\endgroup
\clearpage
\twocolumngrid
\section{Analytic hypotheses and locality of Petz recovery}
\label{sec:setting}\label{sec:analytic}\label{an:sec:analytic}

\subsection{Geometry, states, and recovery maps}

\begin{definition}[Geometry and Gibbs states]\label{an:def:geometry}
Let $\Lambda$ be the vertex set of a finite connected graph with uniformly
bounded degree, $N=|\Lambda|\ge1$, and shortest-path distance $d$.
The local Hilbert-space dimension $q\ge2$ is fixed. Assume polynomial volume growth
\begin{equation}\label{an:eq:growth}
\begin{aligned}
 |B_r(x)|&\le c_{\mathrm{vol}}(1+r)^D,\\
 B_r(x)&:=\{y\in\Lambda:d(x,y)\le r\}.
\end{aligned}
\end{equation}
The bound holds for every $x\in\Lambda$ and $r\ge0$, with fixed $D\ge1$ and $c_{\mathrm{vol}}>0$ independent of $N$.
For nonempty sets $A,B$, put $d(A,B):=\min_{a\in A,b\in B}d(a,b)$,
and set the distance to $+\infty$ if either region is empty.
Write $d(x,B):=d(\{x\},B)$ and
$\diam(A):=\max_{a,b\in A}d(a,b)$ for nonempty $A$, with
$\diam(\varnothing):=0$.
The Hilbert space of $A$ is
$\mathcal H_A:=\bigotimes_{x\in A}\mathbb C^q$, with
$\mathcal H_{\varnothing}:=\mathbb C$ and $d_A:=q^{|A|}$.
Write $\tr_A(O):=d_A^{-1}\Tr_A O$ for the normalized full trace.
Disjoint unions are abbreviated by juxtaposition, for example $YZ=Y\mathbin{\dot\cup}Z$.
Every operator on a smaller region is embedded by tensoring with identities.

For every $\Omega\subseteq\Lambda$, let
\begin{equation}\label{an:eq:gibbs}
\begin{gathered}
 H_\Omega=\sum_{\varnothing\ne S\subseteq\Omega}h_S,
 \qquad \gamma_\Omega=\frac{e^{-\beta H_\Omega}}{\Tr e^{-\beta H_\Omega}},\\
 h_S=h_S^\dagger,\qquad |S|\le k,\qquad\diam(S)\le r_0,
\end{gathered}
\end{equation}
with fixed $k,r_0$ and $\beta>0$.
For $A\subseteq\Omega$ define $\rho_A^\Omega:=\Tr_{\Omega\setminus A}\gamma_\Omega$.
All nonempty-system Gibbs marginals are positive definite; such states are also called faithful. The state on the empty system is the scalar $1$.
The normalized partial trace is $\mathbb E_B:=d_B^{-1}\Tr_B$; it is not used in place of
$\Tr_B$ when defining density-operator marginals.
\end{definition}

All logarithms are natural. The Euclidean norm of a state vector is $\norm{\cdot}_2$. Unsubscripted norms of operators are operator norms, $\norm{\cdot}_1$ is
the trace norm, and
\begin{equation}\label{an:eq:diamond}
 \norm{\mathcal T}_\diamond
 :=\sup_{k\ge1}\ \sup_{\norm{O}_1=1}
 \norm{(\mathcal T\otimes\id_k)(O)}_1.
\end{equation}
The symbol $\I_A$ denotes an identity operator and $\id_A$ an identity channel.
We write $\mathcal B(\mathcal H)$ for the linear operators on a
finite-dimensional Hilbert space $\mathcal H$.
The support $\supp(O)$ of an operator on a tensor-product register is
its smallest set of tensor factors outside which it acts as the identity;
the support of a scalar multiple of the identity is empty.
For Hermitian operators $A,B$, the order $A\preceq B$ means that
$B-A$ is positive semidefinite; $A\succeq0$ and $A>0$ mean positive
semidefinite and positive definite, respectively.
The notation $\poly$ denotes a polynomial of fixed degree, and
$\polylog x$ denotes a fixed polynomial in $\log x$.
Constants and polynomial degrees may depend on the fixed spatial dimension,
geometry, temperature, interaction bounds, and the fixed positive differences in the stated inequalities.
They are independent of the system size and desired accuracy.
Dimension dependence is carried by these constants and exponents.
We encode one $q$-dimensional site in
$n_q:=\lceil\log q/\log 2\rceil$ qubits.

\begin{definition}[Conditional mutual information]\label{an:def:cmi}
For a state $\rho_{WYZ}$, its conditional mutual information (CMI) is
\begin{equation}\label{an:eq:cmi}
\begin{aligned}
 I(W:Z\mid Y)_\rho
 &:=S(\rho_{WY})+S(\rho_{YZ})\\
 &\quad-S(\rho_Y)-S(\rho_{WYZ}),\\
 S(\sigma)&:=-\Tr\sigma\log\sigma.
\end{aligned}
\end{equation}
The exact quantum Markov condition is $I(W:Z\mid Y)_\rho=0$.
We will bound the decay of \eqref{an:eq:cmi} as the outer regions are separated.
\end{definition}

\begin{definition}[Ordinary Petz factor, channel, and isometry]\label{an:def:petz}
For disjoint $R,Z\subseteq\Omega$, set
\begin{align}
 Q^\Omega_{Z|R}
 &:= (\rho^\Omega_{RZ})^{1/2}
       \bigl((\rho^\Omega_R)^{-1/2}\otimes\I_Z\bigr),\label{an:eq:q}\\
 \PP^\Omega_{Z|R}(O_R)
 &:=Q^\Omega_{Z|R}(O_R\otimes\I_Z)(Q^\Omega_{Z|R})^\dagger.\label{an:eq:petz}
\end{align}
Fix one orthonormal basis on each site and use its product basis on every
region. For an environment $E$ that is a copy of $Z$ in this basis, define
\begin{equation}\label{an:eq:stinespring}
 \mathsf K_a:=Q^\Omega_{Z|R}(\I_R\otimes|a\rangle_Z),
 \qquad V^\Omega_{Z|R}:=\sum_{a=1}^{d_Z}\mathsf K_a\otimes|a\rangle_E.
\end{equation}
These formulas specify the factor, channel, and Stinespring isometry of
the ordinary, unrotated Petz recovery map \cite{Petz1986Sufficiency}.
\end{definition}

The map in Eq.~\eqref{an:eq:petz} is completely positive and trace
preserving, $V^\Omega_{Z|R}$ is an isometry, and
\begin{equation}\label{an:eq:exact-extension}
 \PP^\Omega_{Z|R}(\rho_R^\Omega)=\rho_{RZ}^\Omega.
\end{equation}
These are standard properties of Petz recovery and its Stinespring
representation \cite{Petz1986Sufficiency,Wilde2015}.
In the Kraus notation above, trace preservation follows directly from
$\Tr_Z\rho_{RZ}^\Omega=\rho_R^\Omega$:
\begin{equation}\label{an:eq:tp}
 \sum_a\mathsf K_a^\dagger\mathsf K_a
 =\Tr_Z[(Q^\Omega_{Z|R})^\dagger Q^\Omega_{Z|R}]=\I_R.
\end{equation}
Substitution of $\rho_R^\Omega$ in the channel gives
Eq.~\eqref{an:eq:exact-extension}.

Exact extension of a reference marginal is not yet local recovery of its
correlations with a distant region. The latter requires an estimate with an
identity channel on that distant region; this is established in
Sec.~\ref{an:sec:analytic}.

\begin{definition}[Canonical purification; {\cite{DuttaFaulkner2021}}]\label{def:canonical-purification}
For every site $x$ let $\mathcal E_x$ be a copy of its local Hilbert space,
and put $\mathcal E_A=\bigotimes_{x\in A}\mathcal E_x$.
For a state $\omega_A$ define
\begin{equation}\label{eq:canonical-state}
\begin{aligned}
 |\Upsilon_A\rangle
   &:=\sum_{\boldsymbol b}|\boldsymbol b\rangle_A
                          |\boldsymbol b\rangle_{\mathcal E_A},\\
 |\Gamma(\omega_A)\rangle
   &:=(\omega_A^{1/2}\otimes\I_{\mathcal E_A})
       |\Upsilon_A\rangle.
\end{aligned}
\end{equation}
The vector $|\Upsilon_A\rangle$ is unnormalized, whereas
$\norm{|\Gamma(\omega_A)\rangle}_2=1$ and its reduced state on $A$
is $\omega_A$. Indeed,
$\langle\Upsilon_A|(\omega_A\otimes\I)|\Upsilon_A\rangle=\Tr\omega_A$
and tracing its reference basis gives $\omega_A$.
For the full Gibbs state write
\begin{equation}\label{eq:canonical-gibbs}
\begin{gathered}
 |\Gamma_\beta\rangle
 =\frac{(e^{-\beta H_\Lambda/2}\otimes\I_{\mathcal E_\Lambda})
             |\Upsilon_\Lambda\rangle}{\sqrt{Z_\beta}},\\
 Z_\beta:=\Tr e^{-\beta H_\Lambda}.
\end{gathered}
\end{equation}
All reference bases and tensor-factor labels are fixed throughout the
construction. The state on the empty system is the scalar $1$.
\end{definition}

\subsection{Weighted effective interactions}

An interaction $\Phi=\{\Phi_S\}_{\varnothing\ne S\subseteq\Lambda}$ is a
specified family of Hermitian operators, with $\Phi_S$ supported on $S$.
Terms absent from a decomposition are zero. For $\lambda>0$ and
$\mu\ge0$, define
\begin{equation}\label{an:eq:weights}
\begin{aligned}
 w_{\lambda,\mu}(S)&:=e^{\lambda|S|+\mu\diam(S)},\\
 \norm{\Phi}_{\lambda,\mu}
 &:=\sup_x\sum_{S\ni x}\norm{\Phi_S}w_{\lambda,\mu}(S).
\end{aligned}
\end{equation}
Every use of this expression refers to a specified interaction decomposition.
Scalar multiples of the full identity are kept separately.
Fix positive decay parameters $\mu_0,\mu_1,\mu_2$, independent
of $N$. Assume a uniform bound on
\begin{equation}\label{an:eq:jh}
 J_H:=\sup_x\sum_{S\ni x}\norm{h_S}w_{\lambda,\mu_0}(S).
\end{equation}

Local effective Hamiltonians are used in Refs.~\cite{BluhmCapelPerezHernandez2025,ScaletEtAl2025}. We use the weighted norm of Ref.~\cite{BluhmCapelPerezHernandez2025}, Eq.~(3), and impose the following explicit retained-region estimates. In contrast to an approximate finite-range decomposition, E1--E2 refer to one simultaneously chosen exact family, with quantitative bounds for every retained subset.

\begin{assumption}[Effective interactions]\label{an:ass:effective}
For every $A\subseteq\Omega\subseteq\Lambda$ choose
\begin{equation}\label{an:eq:heff}
\begin{gathered}
 \He_A^\Omega=H_A+\sum_{\varnothing\ne S\subseteq A}h_S^{\beta,\Omega;A},\\
 \rho_A^\Omega=\frac{e^{-\beta\He_A^\Omega}}{\mathcal Z_A^\Omega},
 \qquad \mathcal Z_A^\Omega:=\Tr e^{-\beta\He_A^\Omega}.
\end{gathered}
\end{equation}
The correction terms are Hermitian, have the indicated supports, and are chosen consistently
for all pairs $A\subseteq\Omega$.  Set $\He_\Omega^\Omega=H_\Omega$.
Pure scalar corrections are absorbed into normalization.
For fixed nonnegative constants $C_\beta,g_1^\beta$, independent of $N,\Omega,A$, assume the following bounds
for every $x\in A$, with $A\cap B=\varnothing$ in the second inequality:
\begin{align}
 &\sum_{S\subseteq A:S\ni x}
 \norm{h_S^{\beta,\Omega;A}}w_{\lambda,\mu_0}(S)
 \le C_\beta e^{-\mu_1d(x,\Omega\setminus A)},
 \tag{E1}\label{an:eq:e1}\\
 &\sum_{S\subseteq A:S\ni x}
 \norm{h_S^{\beta,\Omega;A}-h_S^{\beta,\Omega;AB}}
 w_{\lambda,\mu_0}(S)\notag\\[-1mm]
 &\hspace{3cm}\le g_1^\beta e^{-\mu_2d(x,B)}.
 \tag{E2}\label{an:eq:e2}
\end{align}
The term carrying a support label absent from a decomposition is zero.  In \eqref{an:eq:e2},
$A,B\subseteq\Omega$, and the two terms carry the \emph{same} label $S\subseteq A$.
The bounds are required for all retained regions used below, including disconnected regions
and regions with holes, not only for connected intervals.
\end{assumption}

Local indistinguishability controls the effect of distant boundary changes
on Gibbs marginals~\cite{CapelMoscolariTeufelWessel2025}.
Lemma~\ref{lem:ambient-marginals} below derives this state-level estimate
from Assumption~\ref{an:ass:effective} and the high-temperature condition.
It is then used to compare Petz isometries on their Gibbs inputs.

The nontrivial requirement in Assumption~\ref{an:ass:effective} is uniform locality of the
marginal logarithms, rather than their existence, which follows from
positive definiteness at finite temperature. Concrete high-temperature
examples are supplied by the uniform effective-interaction bounds of
Ref.~\cite[Remark~3.1 and Theorem~3.6]{BluhmCapelPerezHernandez2025}, for
interactions belonging to a commutative algebra preserved by normalized
partial traces; this class includes finite-range classical spin models.
Grouping those coefficients by their retained support yields
\eqref{an:eq:e1}--\eqref{an:eq:e2}, with a small reduction of the spatial
decay exponent: only terms reaching the traced region contribute to the
boundary correction, and only terms reaching the changed region survive
in coefficient differences. For one-dimensional noncommuting systems,
quasi-local effective Hamiltonians for interval marginals are known at
every fixed finite temperature
\cite{Kuwahara2025Markov,ScaletEtAl2025}; those interval estimates alone
do not supply the simultaneous bounds on disconnected retained sets used
here. The temperature qualification is substantive. In dimensions
$D\ge2$, retaining a periodic sublattice of a $q$-state Potts model with
sufficiently large but fixed $q$ can produce, in the thermodynamic limit,
a marginal without a uniformly quasi-local Gibbs description even above
the bulk transition temperature
\cite[Theorem~4.2]{VanEnterFernandezKotecky1995}. The resulting conditional
probabilities retain a nonvanishing sensitivity to distant boundary
changes, which is incompatible with a volume-uniform exponentially
summable interaction description. This distinguishes the high-temperature
phase from the sufficiently high-temperature regime of the coefficient
bounds. For general noncommuting models, the full simultaneous estimates
remain explicit sufficient hypotheses; the construction cost of approximate
interaction lists is specified separately in
Sec.~\ref{sec:reference-input}.

The following sufficient high-temperature condition permits the
complex-time bounds of Ref.~\cite[Proposition~2.1]{BluhmCapelPerezHernandez2025}.
One-dimensional complex-time locality can hold in a broader temperature
regime~\cite{PerezGarciaPerezHernandez2023}.

\begin{assumption}[High-temperature condition]\label{an:ass:strip}
For the ordinary-Petz locality statements, let $\Delta_0$ be independent
of the system size and assume
\begin{equation}\label{an:eq:strip}
 \beta(J_H+C_\beta)\le\Delta_0<\lambda.
\end{equation}
\end{assumption}
The fixed separation $\lambda-\Delta_0>0$ ensures volume-independent
norm and truncation estimates throughout the required imaginary-time
interval. It is not a condition for the existence of finite-dimensional
matrix exponentials.

Put
\begin{equation}\label{an:eq:parameters}
 0<4\kappa<\min\{\mu_0,\mu_1,\mu_2\},
 \qquad \mu_*:=\mu_0-2\kappa>\kappa.
\end{equation}
Let
\begin{equation}\label{an:eq:k-circ}
 K_A^\circ:=\beta\He_A^\Omega,
 \qquad
 U^\Omega_{Z|R}:=K_{RZ}^\circ-K_R^\circ\otimes\I_Z.
\end{equation}
The superscript $\Omega$ on $K_A^\circ$ is suppressed while $\Omega$ is fixed.
The centered modular Hamiltonian used in the circuit construction is
$K_A=K_A^\circ-\tr_A(K_A^\circ)\I_A$.
In contrast, $-\log\rho_A^\Omega=K_A^\circ+\log\mathcal Z_A^\Omega\,\I_A$.
These three representatives must not be interchanged without accounting for scalars.

For a specified decomposition $O=c\I+\sum_S O_S$, its weighted total sum is
\begin{equation}\label{an:eq:total-mass}
 |c|+\sum_S\norm{O_S}w_{\lambda,\mu}(S).
\end{equation}
For a supplied list $O=c\I+\sum_S O_S$, let $a_S\ge\norm{O_S}$
be computable nonnegative bounds. Replacing $\norm{O_S}$ by $a_S$
in \eqref{an:eq:weights} or \eqref{an:eq:total-mass} gives a computable
upper bound on the corresponding weighted norm or total sum.
These upper bounds need not be tight. The non-scalar weighted coefficient
bound of the list will be denoted by
\begin{equation}\label{eq:list-weight-bound}
 J^{\mathrm{bd}}(O):=\sup_x\sum_{S\ni x}a_Sw_{\lambda,\mu}(S).
\end{equation}
The symbols $J^{\mathrm{bd}}_{L_j}$ refer to this quantity for the specified
list of $L_j$. Analytic decompositions need not be accessible to the
algorithm; supplied lists include supports, finite matrices, and norm bounds.

For a Hermitian $K$, put
\begin{equation}\label{an:eq:tau}
\begin{gathered}
 \tau_K^t(O):=e^{-tK/2}Oe^{tK/2},\qquad \ad_K^0(O):=O,\\
 \ad_K^{n+1}(O):=[K,\ad_K^n(O)].
\end{gathered}
\end{equation}
Scalar parts of $K$ do not affect $\tau_K^t$.

\subsection{Geometric sums and four-region cancellation}

Define $\Csh{a}:=c_{\mathrm{vol}}\sum_{n=0}^\infty(1+n)^D e^{-an}$
for $a>0$. Polynomial volume growth gives, for nonempty $Z$,
\begin{equation}\label{an:eq:shell}
\begin{aligned}
 \sum_xe^{-ad(x,Z)}&\le\Csh{a}|Z|,\\
 \sum_{x:d(x,Z)\ge r}e^{-ad(x,Z)}
 &\le\Csh{a-b}|Z|e^{-br}\quad(0<b<a).
\end{aligned}
\end{equation}
For a single site, bound the shell at distance $n$ by
$c_{\mathrm{vol}}(1+n)^D$ and sum. For a set, use
$e^{-ad(x,Z)}\le\sum_{z\in Z}e^{-ad(x,z)}$; extracting $e^{-br}$
gives the tail estimate. This is the geometric summability used in
Ref.~\cite[Remark~2.5]{BluhmCapelPerezHernandez2025}.

The four-region cancellation used for CMI decay in
Ref.~\cite[Lemma~1]{ScaletEtAl2025} also supplies the interaction
difference needed by the Petz-locality proof. We retain that
decomposition with its spatial weights; the CMI bound is a consequence.

\begin{lemma}[Four-region cancellation and conditional mutual information]
\label{an:lem:four-cmi}
Under Assumption~\ref{an:ass:effective}, let $W,Y,Z\subseteq\Omega$
be pairwise disjoint with $W,Z\ne\varnothing$, and put $r=d(W,Z)$.
The difference
\begin{equation}\label{an:eq:four}
\begin{aligned}
 \mathcal D^\circ
 &:=U^\Omega_{Z|WY}-\I_W\otimes U^\Omega_{Z|Y}\\
 &=K_{WYZ}^\circ+K_Y^\circ-K_{WY}^\circ-K_{YZ}^\circ
\end{aligned}
\end{equation}
has a specified decomposition $\sum_S d_S$ satisfying
\begin{equation}\label{an:eq:four-bound}
 \sum_S\norm{d_S}w_{\lambda,\mu_*}(S)
 \le C_{\mathrm{rel}}|Z|e^{-\kappa r},
\end{equation}
where $C_{\mathrm{rel}}$ is independent of the region volumes.
For the marginals of $\rho=\rho^\Omega_{WYZ}$, let
$\mathcal J=\log\rho_{WYZ}+\log\rho_Y-\log\rho_{WY}-\log\rho_{YZ}$.
Then
\begin{equation}\label{an:eq:cmi-main}
 0\le I(W:Z\mid Y)_\rho\le\norm{\mathcal J}
 \le2C_{\mathrm{rel}}\min\{|W|,|Z|\}e^{-\kappa r}.
\end{equation}
The CMI and logarithmic-combination bounds are zero if either outer
region is empty and do not require the high-temperature condition.
\end{lemma}

\begin{proof}
Use the common support labels in Eq.~\eqref{an:eq:heff} and anchor
each term at the first site of the indicated region in a fixed ordering.
Terms meeting both $W$ and $Z$ have diameter at least $r$; after
reducing the spatial weight from $\mu_0$ to $\mu_*$, their total is at
most $\beta|Z|(J_H+C_\beta)e^{-2\kappa r}$. For terms meeting $Z$
but not $W$, microscopic contributions cancel and (E2) bounds the
remaining total by $\beta g_1^\beta|Z|e^{-\mu_2r}$. For terms meeting
$W$ but not $Z$, (E2) and Eq.~\eqref{an:eq:shell} give
$\beta g_1^\beta\Csh{\mu_2-\kappa}|Z|e^{-\kappa r}$.

For supports contained in $Y$, pairing the four coefficients by adding
either outer region bounds the $w_{\lambda,\mu_0}$-weighted sum
anchored at $x\in Y$ by both
$2\beta g_1^\beta e^{-\mu_2d(x,Z)}$ and
$2\beta g_1^\beta e^{-\mu_2d(x,W)}$. Their geometric mean, together
with $d(x,Z)+d(x,W)\ge r$, is at most
$2\beta g_1^\beta e^{-\kappa r}
 e^{-(\mu_2/2-\kappa)d(x,Z)}$.
Summing with Eq.~\eqref{an:eq:shell} proves
Eq.~\eqref{an:eq:four-bound}, for example with
\begin{equation}\label{an:eq:crel}
\begin{split}
 C_{\mathrm{rel}}=\beta\bigl[&J_H+C_\beta+
 g_1^\beta\bigl(1+\Csh{\mu_2-\kappa}\\
 &\hspace{35mm}+2\Csh{\mu_2/2-\kappa}\bigr)\bigr].
\end{split}
\end{equation}

The logarithmic combination is $\mathcal J=-\mathcal D^\circ+c\I$,
where $c$ contains the partition-function scalars. Put
$\delta=C_{\mathrm{rel}}|Z|e^{-\kappa r}$.
Relative-entropy data processing~\cite{Lindblad1975}, with
$D(\sigma\|\omega)=\Tr\sigma(\log\sigma-\log\omega)$, gives
$\Tr\rho\mathcal J=I(W:Z\mid Y)_\rho\ge0$, hence $c\ge-\delta$.
For $\sigma=\rho_W\otimes\rho_{YZ}$, the same inequality gives
\[
 0\le D(\sigma\|\rho)-D(\sigma_{WY}\|\rho_{WY})
 =-\Tr\sigma\mathcal J,
\]
so $c\le\delta$ and $\norm{\mathcal J}\le2\delta$.
Interchanging $W$ and $Z$ yields Eq.~\eqref{an:eq:cmi-main}.
\end{proof}

\subsection{Imaginary-time interaction estimates}

The support-counting argument is that of Ref.~\cite[proof of Proposition~2.1]{BluhmCapelPerezHernandez2025}; we record its multilinear form with a separate interaction at each position.

\begin{lemma}[Sum over overlapping supports]\label{an:lem:connected}
Let $n\ge1$ and $S\ne\varnothing$, and put
$\mathcal S_j:=S\cup X_1\cup\cdots\cup X_j$ for $0\le j\le n$,
with $\mathcal S_0=S$. Let
$\Phi^{(1)},\ldots,\Phi^{(n)}$ be interactions with
$\norm{\Phi^{(j)}}_{\lambda,0}\le J_j$.  Then
\begin{equation}\label{an:eq:connected}
 \sum_{\substack{X_1\cap S\ne\varnothing\\
                  X_j\cap \mathcal S_{j-1}\ne\varnothing\ (2\le j\le n)}}
 \prod_{j=1}^{n}\norm{\Phi^{(j)}_{X_j}}
 \le e^{\lambda|S|}\frac{n!}{\lambda^n}\prod_{j=1}^{n}J_j.
\end{equation}
\end{lemma}

\begin{proof}
For every allowed support sequence,
\begin{equation}\label{an:eq:connected-proof}
\begin{aligned}
 \prod_j\norm{\Phi^{(j)}_{X_j}}
 &\le e^{\lambda|S|}e^{-\lambda|\mathcal S_n|}
       \prod_j\bigl(\norm{\Phi^{(j)}_{X_j}}e^{\lambda|X_j|}\bigr),\\
 e^{-\lambda|\mathcal S_n|}
 &\le\frac{n!}{\lambda^n}\prod_{j=1}^{n}\frac1{|\mathcal S_{j-1}|}.
\end{aligned}
\end{equation}
The second inequality uses $e^x\ge x^n/n!$ and
$|\mathcal S_{j-1}|\le|\mathcal S_n|$. Since
$\sum_{X:X\cap F\ne\varnothing}\norm{\Phi_X^{(j)}}e^{\lambda|X|}
\le|F|J_j$, summing backwards cancels the denominators and proves
Eq.~\eqref{an:eq:connected}.
\end{proof}

The single-Hamiltonian estimate is Ref.~\cite[Proposition~2.1]{BluhmCapelPerezHernandez2025} applied termwise. The second estimate allows effective Hamiltonians with different interaction coefficients; the proof below establishes that extension.

\begin{lemma}[Imaginary-time evolution under perturbations]\label{an:lem:itlr}
Fix nonempty regions $W,Z\subseteq\Lambda$. For $i=1,2$, let
$K_i=\sum_{S\ne\varnothing}\Psi_S^{(i)}+c_i\I$ be Hermitian
operators on the same finite register space, with interactions embedded
by identities as necessary and with
$\norm{\Psi^{(i)}}_{\lambda,\mu_*}\le\Delta_0<\lambda$.
For a specified decomposition $O=\sum_{\varnothing\ne S\subseteq\Lambda} O_S$
on that space, define
\begin{equation}\label{an:eq:masses}
\begin{aligned}
 M_O&:=\sum_S\norm{O_S}w_{\lambda,\mu_*}(S),\\
 M_{O,Z}&:=\sum_S\norm{O_S}w_{\lambda,\mu_*}(S)e^{\kappa d(S,Z)}.
\end{aligned}
\end{equation}
For $0\le t\le1$,
\begin{equation}\label{an:eq:single-it}
 \norm{\tau_{K_i}^t(O)}\le\frac{\lambda M_O}{\lambda-t\Delta_0}.
\end{equation}
If $\chi_S:=\Psi_S^{(1)}-\Psi_S^{(2)}$ satisfies
\begin{equation}\label{an:eq:chi-ass}
 \sup_x\sum_{S\ni x}\norm{\chi_S}w_{\lambda,\mu_*}(S)
 e^{\kappa d(S,W)}\le J_\chi,
\end{equation}
then
\begin{equation}\label{an:eq:two-it}
 \norm{\tau_{K_1}^t(O)-\tau_{K_2}^t(O)}
 \le
 \frac{t\lambda}{(\lambda-t\Delta_0)^2}
 J_\chi M_{O,Z}e^{-\kappa d(W,Z)}.
\end{equation}
\end{lemma}

\begin{proof}
Equation~\eqref{an:eq:single-it} follows termwise from
Ref.~\cite[Proposition~2.1]{BluhmCapelPerezHernandez2025}.

For the difference use the exact superoperator identity
\begin{equation}\label{an:eq:telescoping-ad}
 \ad_{K_1}^{n}-\ad_{K_2}^{n}
 =\sum_{j=1}^{n}\ad_{K_1}^{n-j}
       (\ad_{K_1}-\ad_{K_2})\ad_{K_2}^{j-1}.
\end{equation}
Each summand contains one distinguished term $\chi_{X_j}$.
For a nonzero sequence starting at $S$, the overlaps before that insertion imply
\begin{equation}\label{an:eq:bridge-distance}
 d(W,Z)\le d(S,Z)+\diam(S)+d(X_j,W)+\sum_{i=1}^{j}\diam(X_i).
\end{equation}
The supports $S,X_1,\ldots,X_j$ form a connected intersection graph,
because each new support meets their previous union. A path in this
intersection graph joins $S$ to $X_j$. Choose points in consecutive
intersections; the distance crossed inside a support is at most its
diameter. Summing along a simple path gives the displayed bound, which
is no larger than the sum over all $X_1,\ldots,X_j$. No individual
support needs to be connected in the original lattice.
Multiply a summand's norm by $e^{\kappa d(W,Z)}$ and absorb the right-hand side of
\eqref{an:eq:bridge-distance} into \eqref{an:eq:masses}, \eqref{an:eq:chi-ass}, and the
$e^{\mu_*\diam(X_i)}$ weights.  This is permitted because $\mu_*>\kappa$.
For each fixed distinguished position, Lemma~\ref{an:lem:connected} then gives the bound
\begin{equation}\label{an:eq:bridge-order}
 2^n\frac{n!}{\lambda^n}
 J_\chi M_{O,Z}\Delta_0^{n-1}e^{-\kappa d(W,Z)}.
\end{equation}
There are $n$ positions.  After multiplying by $(t/2)^n/n!$, their sum is
\begin{equation}\label{an:eq:bridge-sum}
 \sum_{n\ge1}\frac{n t^n\Delta_0^{n-1}}{\lambda^n}
 =\frac{t\lambda}{(\lambda-t\Delta_0)^2},
\end{equation}
which proves \eqref{an:eq:two-it}.  The displayed convergent majorants justify all sums.
\end{proof}

Equation~\eqref{an:eq:two-it} controls the effect of a perturbation near
$W$ on an interaction localized near $Z$ over the time interval allowed by Assumption~\ref{an:ass:strip}.

\subsection{Petz factors and normalization}

We now apply the preceding interaction estimates to the generator. This is the effective-Hamiltonian analogue of the expansional setup in Ref.~\cite[Proposition~2.2]{BluhmCapelPerezHernandez2025}.

\begin{lemma}[Interaction bounds for nested regions]\label{an:lem:relative}
Under Assumption~\ref{an:ass:effective}, let $R,Z\subseteq\Omega$
be disjoint with $Z\ne\varnothing$. The difference
$U^\Omega_{Z|R}=\sum_Su_S$ in \eqref{an:eq:k-circ} satisfies
\begin{equation}\label{an:eq:root-u}
\begin{gathered}
 \sum_S\norm{u_S}w_{\lambda,\mu_*}(S)e^{\kappa d(S,Z)}\le u_0|Z|,\\
 u_0:=\beta\left(J_H+C_\beta+g_1^\beta\Csh{\mu_2-\kappa}\right).
\end{gathered}
\end{equation}
For disjoint $W,Y,Z\subseteq\Omega$ with $W,Z\ne\varnothing$, let
\begin{equation}\label{an:eq:comparison-hamiltonians}
 K_1^\circ:=K_{WY}^\circ\otimes\I_Z,
 \qquad K_2^\circ:=\I_W\otimes K_Y^\circ\otimes\I_Z.
\end{equation}
Their interactions have weighted norm at most $\beta(J_H+C_\beta)$.
The difference of these interactions satisfies \eqref{an:eq:chi-ass} with
\begin{equation}\label{an:eq:jchi}
 J_\chi:=\beta(J_H+C_\beta+g_1^\beta).
\end{equation}
\end{lemma}

\begin{proof}
In $U^\Omega_{Z|R}$, microscopic terms meeting $Z$ have weighted total sum at most
$\beta|Z|J_H$; effective correction terms meeting $Z$ contribute at most
$\beta|Z|C_\beta$.  Both have $d(S,Z)=0$.
For $S\subseteq R$, use the difference in \eqref{an:eq:e2} for adding $Z$ and anchor
at a site of $S$ closest to $Z$.  Multiplication by $e^{\kappa d(S,Z)}$ leaves
$e^{-(\mu_2-\kappa)d(x,Z)}$, whose sum is bounded by \eqref{an:eq:shell}.
This proves \eqref{an:eq:root-u}.

For the two Hamiltonians, supports meeting $W$ occur only in $K_1^\circ$ and have
$d(S,W)=0$, so their local weighted sum is at most $\beta(J_H+C_\beta)$.
For $S\subseteq Y$, microscopic terms cancel.  Apply \eqref{an:eq:e2} to adding $W$,
and use $d(S,W)\le d(x,W)$ at any $x\in S$ to insert $e^{\kappa d(S,W)}$.
The resulting local weighted sum is at most $\beta g_1^\beta$.
This proves \eqref{an:eq:jchi} and the assertion about the individual weighted interaction norms.
\end{proof}

Fix a constant $b_*$ and for the remainder of this subsection require $1\le|Z|\le b_*$.
Define fixed constants
\begin{align}
 g_0&:=\frac{\lambda u_0b_*}{2(\lambda-\Delta_0)},\label{an:eq:g0}\\
 h_0&:=\frac{b_*}{2}\left[
 \frac{\lambda C_{\mathrm{rel}}}{\lambda-\Delta_0}
 +\frac{\lambda J_\chi u_0}{(\lambda-\Delta_0)^2}\right].\label{an:eq:h0}
\end{align}

Products of imaginary-time evolutions, called expansionals, are studied in Refs.~\cite{PerezGarciaPerezHernandez2023,BluhmCapelPerezHernandez2025}. The following form compares the two effective Hamiltonians used in the Petz factor.

\begin{lemma}[Perturbation bounds for exponential products]\label{an:lem:relative-evol}
Under Assumptions~\ref{an:ass:effective} and~\ref{an:ass:strip},
let $W,Y,Z\subseteq\Omega$ be disjoint, with $W\ne\varnothing$
and $1\le|Z|\le b_*$. Use $U_1:=U^\Omega_{Z|WY},U_2:=\I_W\otimes U^\Omega_{Z|Y}$ from
\eqref{an:eq:four} and the Hamiltonians \eqref{an:eq:comparison-hamiltonians}.  Set
\begin{equation}\label{an:eq:generator}
 G_i(t):=\tfrac12\tau_{K_i^\circ}^{t}(U_i),
 \qquad
 F_i(t):=e^{-t(K_i^\circ+U_i)/2}e^{tK_i^\circ/2}.
\end{equation}
Then, with $r=d(W,Z)$,
\begin{equation}\label{an:eq:ode}
 F_i(0)=\I,\qquad F_i'(t)=-F_i(t)G_i(t).
\end{equation}
\begin{equation}\label{an:eq:g-bounds}
\begin{aligned}
 \sup_{0\le t\le1}\norm{G_i(t)}&\le g_0,\\
 \sup_{0\le t\le1}\norm{G_1(t)-G_2(t)}&\le h_0e^{-\kappa r}.
\end{aligned}
\end{equation}
\begin{equation}\label{an:eq:f-bounds}
\begin{aligned}
 \norm{F_i(t)},\ \norm{F_i(t)^{-1}}&\le e^{g_0t},\\
 \norm{F_1(1)-F_2(1)}&\le e^{g_0}h_0e^{-\kappa r}.
\end{aligned}
\end{equation}
\end{lemma}

\begin{proof}
Differentiation gives \eqref{an:eq:ode}.  The individual generator bounds follow from
\eqref{an:eq:single-it} and \eqref{an:eq:root-u}.
For their difference write
\begin{equation}\label{an:eq:g-split}
 2(G_1-G_2)
 =\tau_{K_1^\circ}^{t}(U_1-U_2)
   +(\tau_{K_1^\circ}^{t}-\tau_{K_2^\circ}^{t})(U_2).
\end{equation}
For the first term, apply the single-Hamiltonian estimate
\eqref{an:eq:single-it} of Lemma~\ref{an:lem:itlr} and
Lemma~\ref{an:lem:four-cmi}. For the second, use the perturbation
estimate \eqref{an:eq:two-it} of the same lemma together with
Eqs.~\eqref{an:eq:root-u} and~\eqref{an:eq:jchi}.
This gives \eqref{an:eq:g-bounds}.

The integral equation for $F_i$ and the equation
$(F_i^{-1})'=G_iF_i^{-1}$ give the first bounds in \eqref{an:eq:f-bounds}.
If $D(t):=F_1(t)-F_2(t)$, then
$D'=-DG_1-F_2(G_1-G_2)$ and $D(0)=0$.
Variation of constants bounds $\norm{D(1)}$ by
\begin{equation}\label{an:eq:duhamel-bound}
 \int_0^1 e^{g_0(1-s)}e^{g_0s}h_0e^{-\kappa r}\,ds
 =e^{g_0}h_0e^{-\kappa r}.
\end{equation}
This proves all assertions.
\end{proof}

\begin{theorem}[Locality for a fixed Gibbs state]\label{an:thm:local-petz}
Under Assumptions~\ref{an:ass:effective} and \ref{an:ass:strip}, for disjoint $W,Y,Z\subseteq\Omega$
with $W\ne\varnothing$ and $1\le|Z|\le b_*$, set
\begin{equation}\label{an:eq:cq}
 C_Q:=h_0(e^{2g_0}+e^{6g_0}).
\end{equation}
Then
\begin{align}
 \norm{Q^\Omega_{Z|WY}-\I_W\otimes Q^\Omega_{Z|Y}}
 &\le d_Z^{-1/2}C_Qe^{-\kappa d(W,Z)},\label{an:eq:q-local}\\
 \norm{\PP^\Omega_{Z|WY}-\id_W\otimes\PP^\Omega_{Z|Y}}_\diamond
 &\le 2C_Qe^{-\kappa d(W,Z)},\label{an:eq:channel-local}\\
 \norm{(\id_W\otimes\PP^\Omega_{Z|Y})(\rho^\Omega_{WY})-\rho^\Omega_{WYZ}}_1
 &\le 2C_Qe^{-\kappa d(W,Z)}.\label{an:eq:fixed-recovery}
\end{align}
These constants are independent of $|W|,|Y|,|\Omega|$, and $N$.
\end{theorem}

\begin{proof}[Proof of Theorem~\ref{an:thm:local-petz}]
The unnormalized factors in Lemma~\ref{an:lem:relative-evol} are
$\widehat Q_1=F_1(1)$ and $\widehat Q_2=F_2(1)$, with identities on $W$ understood.
The normalized factors are $Q_i=c_i\widehat Q_i$ with
\begin{equation}\label{an:eq:ci}
 c_1=\left(\frac{\mathcal Z^\Omega_{WY}}{\mathcal Z^\Omega_{WYZ}}\right)^{1/2},
 \qquad
 c_2=\left(\frac{\mathcal Z^\Omega_Y}{\mathcal Z^\Omega_{YZ}}\right)^{1/2}.
\end{equation}
By trace preservation in \eqref{an:eq:tp},
\begin{equation}\label{an:eq:scalar-tp}
 \Tr_Z(\widehat Q_i^\dagger\widehat Q_i)=c_i^{-2}\I,
 \qquad
 d_Ze^{-2g_0}\le c_i^{-2}\le d_Ze^{2g_0}.
\end{equation}
Let $\epsilon_F:=\norm{\widehat Q_1-\widehat Q_2}$.
Since $\norm{\Tr_Z O}\le d_Z\norm{O}$,
\begin{equation}\label{an:eq:scalar-diff}
 |c_1^{-2}-c_2^{-2}|\le2d_Ze^{g_0}\epsilon_F.
\end{equation}
On the interval in \eqref{an:eq:scalar-tp}, the derivative of $x^{-1/2}$ has absolute
value at most $\tfrac12d_Z^{-3/2}e^{3g_0}$.
Thus $|c_1-c_2|\le d_Z^{-1/2}e^{4g_0}\epsilon_F$ and
\begin{equation}\label{an:eq:normalized-diff}
 \norm{Q_1-Q_2}
 \le d_Z^{-1/2}(e^{g_0}+e^{5g_0})\epsilon_F.
\end{equation}
Use \eqref{an:eq:f-bounds} to obtain \eqref{an:eq:q-local}.
No smallest eigenvalue of a marginal enters the argument.

Choose the same $Z$ basis and environment in \eqref{an:eq:stinespring} for both maps.
Then
\begin{equation}\label{an:eq:v-diff}
 \norm{V_1-V_2}\le\sqrt{d_Z}\norm{Q_1-Q_2}.
\end{equation}
The common-reference estimate of Sec.~\ref{sec:composition-errors}
gives channel diamond distance at most $2\norm{V_1-V_2}$.
Equations~\eqref{an:eq:q-local} and \eqref{an:eq:v-diff} prove \eqref{an:eq:channel-local}.
Finally, $\PP^\Omega_{Z|WY}(\rho^\Omega_{WY})=\rho^\Omega_{WYZ}$ by
Eq.~\eqref{an:eq:exact-extension}.  Evaluate the channel difference on $\rho^\Omega_{WY}$ to obtain
\eqref{an:eq:fixed-recovery}.
\end{proof}

\subsection{Changes in \texorpdfstring{$\Omega$}{Omega}}
\label{sec:ambient-comparison}

We compare the Gibbs marginals themselves before comparing their Petz
extensions. Local indistinguishability follows from correlation decay
and local perturbation estimates~\cite{CapelMoscolariTeufelWessel2025}.
The following proof uses the imaginary-time bounds already established
above and applies uniformly to the regions in the preparation scheme.

\begin{lemma}[Local indistinguishability of Gibbs marginals]
\label{lem:ambient-marginals}
Under Assumptions~\ref{an:ass:effective} and~\ref{an:ass:strip}, there is a
fixed $C_{\mathrm{marg}}>0$ such that, for every nonempty
$A\subseteq\Omega\subseteq\Omega'\subseteq\Lambda$,
\begin{equation}\label{eq:ambient-marginal-bound}
 \|\rho_A^{\Omega'}-\rho_A^\Omega\|_1
 \le C_{\mathrm{marg}}|A|^{3/2}
       e^{-\kappa d(A,\Omega'\setminus\Omega)/8}.
\end{equation}
The constant depends only on the stated physical and interaction bounds.
The difference is zero for $A=\varnothing$ or $\Omega'=\Omega$.
\end{lemma}

\begin{proof}
For disjoint nonempty $A,B\subseteq\Theta$, use
Lemma~\ref{an:lem:four-cmi} with the conditioning region empty and apply
quantum Pinsker's inequality~\cite{Watrous2018}. This gives
\begin{equation}\label{eq:separated-marginals}
 \|\rho_{AB}^{\Theta}-\rho_A^{\Theta}\otimes\rho_B^{\Theta}\|_1
 \le 2\sqrt{C_{\mathrm{rel}}}\,|A|^{1/2}
                           e^{-\kappa d(A,B)/2}.
\end{equation}
Consequently the covariance of operators on $A$ and $B$ is at most
the right-hand side times their operator norms. All bounds are uniform
in $\Theta$, including disconnected regions.

Fix $u\in\Theta\setminus A$ and work on $\mathcal H_\Theta$.
Put $K=\beta H_\Theta$ and
\[
 \begin{aligned}
 D_u&=-\beta\sum_{S\subseteq\Theta:\,u\in S}h_S,\\
 F_u(t)&=e^{-t(K+D_u)/2}e^{tK/2}.
 \end{aligned}
\]
Here $K+D_u=\beta H_{\Theta\setminus\{u\}}\otimes\I_u$.
For $s\ge r_0$, define $\Theta_s=\Theta\cap B_s(u)$,
$K_s=\beta H_{\Theta_s}\otimes\I_{\Theta\setminus\Theta_s}$, and
$F_{u,s}(t)=e^{-t(K_s+D_u)/2}e^{tK_s/2}$.
Every term of $D_u$ lies in $\Theta_s$ and meets $u$, so its weighted
total sum, also with the distance weight rooted at $u$, is at most
$J:=\beta J_H$. The interactions of $K$ and $K_s$ have weighted norm
at most $\Delta_0$. Every term of their difference meets
$\Theta\setminus\Theta_s$, so the difference satisfies
Eq.~\eqref{an:eq:chi-ass} with $J_\chi\le J$ and
$W=\Theta\setminus\Theta_s$.

Differentiation gives $F_u'=-F_uG_u$, where
$G_u(t)=\tfrac12\tau_K^t(D_u)$, and the same equation holds for
$F_{u,s}$ with $K_s$. Lemma~\ref{an:lem:itlr} therefore gives, uniformly
for $0\le t\le1$,
\[
 \|G_u(t)\|,\ \|G_{u,s}(t)\|\le g,
 \qquad
 \|G_u(t)-G_{u,s}(t)\|\le h e^{-\kappa s},
\]
where $g=\lambda J/[2(\lambda-\Delta_0)]$ and
$h=\lambda J^2/[2(\lambda-\Delta_0)^2]$.
If $\Theta_s=\Theta$, the difference vanishes.
The integral and inverse equations, as in
Lemma~\ref{an:lem:relative-evol}, imply at $t=1$
\begin{equation}\label{eq:site-removal-expansional}
\begin{gathered}
 \|F_u\|,\ \|F_u^{-1}\|,\ \|F_{u,s}\|,\ \|F_{u,s}^{-1}\|\le e^g,\\
 \|F_u-F_{u,s}\|\le e^g h e^{-\kappa s}.
\end{gathered}
\end{equation}

Write $\omega=\gamma_\Theta$. The exact normalized identity is
\begin{equation}\label{eq:site-removal-state}
 \gamma_{\Theta\setminus\{u\}}\otimes\frac{\I_u}{q}
 =\frac{F_u\omega F_u^\dagger}{\Tr(F_u\omega F_u^\dagger)}.
\end{equation}
Let $\omega_s=F_{u,s}\omega F_{u,s}^\dagger/
\Tr(F_{u,s}\omega F_{u,s}^\dagger)$. Both denominators are at least
$e^{-2g}$. For nonzero positive operators $P,Q$,
\(\|P/\Tr P-Q/\Tr Q\|_1\le2\|P-Q\|_1/\Tr P\).
Applying this inequality and Eq.~\eqref{eq:site-removal-expansional}
bounds the trace distance between the right-hand side of
Eq.~\eqref{eq:site-removal-state} and $\omega_s$ by
$4e^{4g}h e^{-\kappa s}$.

Put $r=d(A,u)$ and choose $r_0\le s<r$.
Then $F_{u,s}$ is supported on $\Theta_s$, disjoint from $A$.
For Hermitian $O_A$ with $\|O_A\|\le1$,
\[
\begin{split}
 \Tr[O_A(\omega_s-\omega)]
 &=\frac{\Tr(\omega O_AF_{u,s}^\dagger F_{u,s})}
       {\Tr(\omega F_{u,s}^\dagger F_{u,s})}\\
 &\quad-\Tr(\omega O_A).
\end{split}
\]
After putting the right-hand side over a common denominator,
Eq.~\eqref{eq:separated-marginals} bounds its absolute value by
$2e^{4g}\sqrt{C_{\mathrm{rel}}}\,|A|^{1/2}
 e^{-\kappa(r-s)/2}$.
Taking the supremum over $O_A$ and combining the two errors gives
\[
\begin{split}
 &\|\rho_A^\Theta-\rho_A^{\Theta\setminus\{u\}}\|_1\\
 &\quad\le4e^{4g}h e^{-\kappa s}
   +2e^{4g}\sqrt{C_{\mathrm{rel}}}\,|A|^{1/2}
                                e^{-\kappa(r-s)/2}.
\end{split}
\]
For $r\ge2r_0+2$, take $s=\lfloor r/2\rfloor$.
For smaller $r$, use the trace-distance upper bound $2$.
In both cases,
\begin{equation}\label{eq:one-site-marginal-change}
 \|\rho_A^\Theta-\rho_A^{\Theta\setminus\{u\}}\|_1
 \le C_{\mathrm{site}}|A|^{1/2}e^{-\kappa d(A,u)/4},
\end{equation}
where, for example, one may take
\[
 C_{\mathrm{site}}=
 e^{4g}(4he^\kappa+2\sqrt{C_{\mathrm{rel}}})
 +2e^{\kappa(r_0+1)/2}.
\]
Finally, delete the sites of $\Omega'\setminus\Omega$ in any order
and telescope Eq.~\eqref{eq:one-site-marginal-change}.
The distances are those of the fixed graph $\Lambda$.
Equation~\eqref{an:eq:shell}, with exponents $\kappa/4$ and
$\kappa/8$, yields
\[
 \sum_{u\in\Omega'\setminus\Omega}e^{-\kappa d(A,u)/4}
 \le\Csh{\kappa/8}|A|
                e^{-\kappa d(A,\Omega'\setminus\Omega)/8}.
\]
This proves Eq.~\eqref{eq:ambient-marginal-bound} with
$C_{\mathrm{marg}}=C_{\mathrm{site}}\Csh{\kappa/8}$.
\end{proof}

\begin{theorem}[Finite-patch Petz recovery on Gibbs inputs]
\label{thm:two-cut}
Under Assumptions~\ref{an:ass:effective} and~\ref{an:ass:strip}, let
$W,Y,Z\subseteq\Omega'$ be disjoint,
$YZ\subseteq\Omega\subseteq\Omega'$, and $1\le|Z|\le b_*$.
Use the same product basis on $Z$ and the same environment for the
isometries. Put $r=d(YZ,\Omega'\setminus\Omega)$ and
\begin{equation}\label{eq:gibbs-input-local-error}
 \delta_{W,Y,Z}=C_Qe^{-\kappa d(W,Z)}
       +2\sqrt{C_{\mathrm{marg}}}|YZ|^{3/4}e^{-\kappa r/16}.
\end{equation}
For any purification $|\psi\rangle_{WY\mathcal F}$ of
$\rho_{WY}^{\Omega'}$,
\begin{equation}\label{eq:two-cut-isometry}
 \bigl\|[(V^{\Omega'}_{Z|WY}-\I_W\otimes V^\Omega_{Z|Y})
           \otimes\I_{\mathcal F}]|\psi\rangle\bigr\|_2
 \le\delta_{W,Y,Z}.
\end{equation}
Equivalently, the Petz factors obey the input-weighted bound
\begin{equation}\label{eq:two-cut-factor}
\begin{split}
 \bigl\|&(Q^{\Omega'}_{Z|WY}-\I_W\otimes Q^\Omega_{Z|Y})\\
 &\qquad\times((\rho_{WY}^{\Omega'})^{1/2}\otimes\I_Z)
                        \bigr\|_{\mathrm{HS}}
 \le\delta_{W,Y,Z},
\end{split}
\end{equation}
where $\|O\|_{\mathrm{HS}}=(\Tr O^\dagger O)^{1/2}$.
The corresponding recovery error is
\begin{equation}\label{eq:two-cut-channel}
 \| (\id_W\otimes\PP^\Omega_{Z|Y})(\rho_{WY}^{\Omega'})
          -\rho_{WYZ}^{\Omega'}\|_1\le2\delta_{W,Y,Z}.
\end{equation}
All constants are independent of the region volumes.
\end{theorem}

\begin{proof}[Proof of Theorem~\ref{thm:two-cut}]
We first compare ordinary Petz isometries for two faithful states
$\rho,\sigma$ on $YZ$, with marginals $\rho_Y,\sigma_Y$.
The square-root trace inequality~\cite[Theorem~1, $s=1/2$]{AudenaertEtAl2007Chernoff}
gives
\[
 \||\Gamma(\rho)\rangle-|\Gamma(\sigma)\rangle\|_2^2
 =\|\sqrt\rho-\sqrt\sigma\|_{\mathrm{HS}}^2
 \le\|\rho-\sigma\|_1.
\]
Direct cancellation in the fixed-basis Petz isometry yields
$(V_\rho\otimes\I)|\Gamma(\rho_Y)\rangle=|\Gamma(\rho)\rangle$,
and likewise for $\sigma$; this is the identity recorded in
Proposition~\ref{prop:canonical-extension}.
The triangle inequality, isometry of $V_\sigma$, and trace-norm
contraction under partial trace imply
\begin{equation}\label{eq:petz-input-continuity}
\begin{split}
 &\|[(V_\rho-V_\sigma)\otimes\I]
                   |\Gamma(\rho_Y)\rangle\|_2\\
 &\quad\le\||\Gamma(\rho)\rangle-|\Gamma(\sigma)\rangle\|_2
       +\||\Gamma(\rho_Y)\rangle-|\Gamma(\sigma_Y)\rangle\|_2\\
 &\quad\le2\sqrt{\|\rho-\sigma\|_1}.
\end{split}
\end{equation}
For any purification of $\rho_Y$, the squared norm on the first line
is $\Tr[\rho_Y(V_\rho-V_\sigma)^\dagger(V_\rho-V_\sigma)]$.
It therefore has the same value for every reference system and every
purification of that marginal; no change of the fixed output basis is
needed.

Apply Eq.~\eqref{eq:petz-input-continuity} to
$\rho=\rho_{YZ}^{\Omega'}$ and $\sigma=\rho_{YZ}^{\Omega}$,
and use Lemma~\ref{lem:ambient-marginals} with $A=YZ$.
This bounds the change from $V^{\Omega'}_{Z|Y}$ to $V^\Omega_{Z|Y}$,
on $|\psi\rangle$ with $W\mathcal F$ included in its reference, by
$2\sqrt{C_{\mathrm{marg}}}|YZ|^{3/4}e^{-\kappa r/16}$.
The change from $V^{\Omega'}_{Z|WY}$ to
$\I_W\otimes V^{\Omega'}_{Z|Y}$ is at most
$C_Qe^{-\kappa d(W,Z)}$ in operator norm by
Theorem~\ref{an:thm:local-petz} and Eq.~\eqref{an:eq:v-diff}; it vanishes
when $W=\varnothing$.
The triangle inequality proves Eq.~\eqref{eq:two-cut-isometry}.
Expanding the canonical input in its fixed product basis shows that
its squared vector error equals the square of the Hilbert--Schmidt norm
in Eq.~\eqref{eq:two-cut-factor}.
Finally, both outputs in Eq.~\eqref{eq:two-cut-isometry} are unit
vectors. Their trace distance is at most twice their vector distance.
Tracing the reference and using exact Petz extension proves
Eq.~\eqref{eq:two-cut-channel}.
\end{proof}

\section{Petz composition and parallel recovery layers}
\label{sec:analytic-composition}

Locality is now combined with exact Petz composition, before any
reference input or elementary-gate construction is used. The result is an
explicit schedule of ideal finite-patch recovery isometries.

\subsection{Composition estimates}
\label{sec:composition-errors}

We use one telescoping estimate for all circuit and channel compositions.
For maps $A_j,\widetilde A_j:\mathcal H_{j-1}\to\mathcal H_j$ on
compatible finite register spaces,
\begin{equation}\label{eq:general-product-error}
\begin{split}
 &\|A_m\cdots A_1-\widetilde A_m\cdots\widetilde A_1\|\\
 &\quad\le\sum_{j=1}^m
 \left(\prod_{r>j}\|A_r\|\right)
 \|A_j-\widetilde A_j\|
 \left(\prod_{r<j}\|\widetilde A_r\|\right).
\end{split}
\end{equation}
For contractions, in particular initialized unitaries and isometries, this is
\begin{equation}\label{eq:product-error}
 \|A_m\cdots A_1-\widetilde A_m\cdots\widetilde A_1\|
 \le\sum_{j=1}^m\|A_j-\widetilde A_j\|.
\end{equation}
The same identity and diamond-norm one of CPTP maps give the standard
channel bound~\cite[Proposition~3.48]{Watrous2018},
\begin{equation}\label{an:eq:hybrid}
 \|\mathcal M_m\circ\cdots\circ\mathcal M_1
       -\mathcal N_m\circ\cdots\circ\mathcal N_1\|_\diamond
 \le\sum_{j=1}^m\|\mathcal M_j-\mathcal N_j\|_\diamond.
\end{equation}
Both estimates allow appended registers and identity actions on arbitrary
reference systems. A perturbed initial state adds its trace-norm error in
the channel case, or its vector error in the isometry case. Tracing the
outputs of two isometries $V,\widetilde V$ gives channels at diamond
distance at most $2\|V-\widetilde V\|$, by expanding the two outer factors.
This implication is used only for the specified dilations; a channel-error
bound alone does not fix their purifications. Equation~\eqref{eq:product-error}
also bounds the accumulated unitary synthesis error by the sum over all
calls.
For the finite-patch comparison we also use telescoping along an exact
input sequence. If $|\psi_j\rangle=A_j|\psi_{j-1}\rangle$ and all maps
are isometries, then
\begin{equation}\label{eq:input-sequence-telescoping}
\begin{split}
 &\|\widetilde A_m\cdots\widetilde A_1|\psi_0\rangle
                       -|\psi_m\rangle\|_2\\
 &\quad\le\sum_{j=1}^m
                \|(\widetilde A_j-A_j)|\psi_{j-1}\rangle\|_2.
\end{split}
\end{equation}
Indeed, insert the exact prefix before the differing factor and the
approximate suffix after it; the suffix is a contraction. The analogous
trace-norm inequality holds for CPTP maps along an exact density-operator
sequence. Thus input-dependent one-step estimates can be accumulated
without evaluating them on approximate intermediate states.

\subsection{Exact Petz composition}
\label{sec:canonical-telescoping}

\begin{proposition}[Exact one-site factorization; {\cite{Wilde2015,ParzygnatBuscemi2023}}]\label{an:thm:factorization}
Let $\rho_{YZ}>0$ be any faithful density operator; no Gibbs, locality, or temperature
assumption is needed.  Use Definition~\ref{an:def:petz} for this state and omit the superscript $\Omega$.  Write
\begin{equation}\label{an:eq:ordered-z}
\begin{gathered}
 Z=\mathop{\dot\bigcup}_{j=1}^{m}\{z_j\},\qquad \omega_j:=\rho_{YZ^{(j)}},\\
 Z^{(0)}:=\varnothing,\qquad Z^{(j)}:=\{z_1,\ldots,z_j\}.
\end{gathered}
\end{equation}
Define
\begin{equation}\label{an:eq:one-site}
 Q_j:=\omega_j^{1/2}(\omega_{j-1}^{-1/2}\otimes\I_{z_j}),
 \qquad \PP_j(O):=Q_j(O\otimes\I_{z_j})Q_j^\dagger.
\end{equation}
Then
\begin{align}
 Q_{Z|Y}
 &=Q_m(Q_{m-1}\otimes\I_{z_m})\cdots
       (Q_1\otimes\I_{z_2\cdots z_m}),\label{an:eq:factorization}\\
 \PP_{Z|Y}&=\PP_m\circ\cdots\circ\PP_1.\label{an:eq:channel-factorization}
\end{align}
The factorization is exact for every ordering of $Z$. For $Z=\varnothing$, the products and the channel are identities.
\end{proposition}
\begin{proof}
This is the serial-composition law of Petz recovery applied to successive
partial traces~\cite[Sec.~4.3]{Wilde2015}; see also
Ref.~\cite[Theorem~3.3, $t=0$]{ParzygnatBuscemi2023}.
Adjacent square roots and inverse square roots cancel in the ordered
factor product and on both sides of the channel input.
\end{proof}
\begin{proposition}[Canonical Petz extension]\label{prop:canonical-extension}
Let $\omega_{Yz}>0$ and $\omega_Y=\Tr_z\omega_{Yz}$. The ordinary
Petz isometry $V_{z|Y}$ for this pair satisfies
\begin{equation}\label{eq:canonical-extension}
 (V_{z|Y}\otimes\I_{\mathcal E_Y})
       |\Gamma(\omega_Y)\rangle
   =|\Gamma(\omega_{Yz})\rangle,
\end{equation}
with registers identified by their site labels.
For any faithful state $\omega_\Lambda$ and any site order, composing the
corresponding one-site isometries from the scalar state produces
$|\Gamma(\omega_\Lambda)\rangle$ exactly.
\end{proposition}

\begin{proof}
Write the input as
$\sum_{\boldsymbol a}\omega_Y^{1/2}|\boldsymbol a\rangle_Y
 |\boldsymbol a\rangle_{\mathcal E_Y}$. The output is
\begin{equation}
 \sum_{\boldsymbol a,b}
 \omega_{Yz}^{1/2}(|\boldsymbol a\rangle_Y\otimes|b\rangle_z)
 |\boldsymbol a\rangle_{\mathcal E_Y}|b\rangle_{\mathcal E_z},
\end{equation}
because $\omega_Y^{-1/2}\omega_Y^{1/2}=\I_Y$ cancels as adjacent
factors. This is \eqref{eq:canonical-extension}. Induction over the
site order gives the full-state identity. The $Y=\varnothing$ case
starts from the scalar $1$.
\end{proof}
For any faithful $\rho_\Lambda$ and ordering $z_1,\ldots,z_N$, put
$\rho_k=\rho_{\{z_1,\ldots,z_k\}}$, $\rho_0=1$, and let $\PP_k$
be the Petz channel for $(\rho_k,\rho_{k-1})$. Tracing the references
then gives exact sequential preparation:
\begin{equation}\label{an:eq:ideal-sequential}
 \rho_\Lambda=\PP_N\circ\cdots\circ\PP_1(1).
\end{equation}
This also follows directly from Eq.~\eqref{an:eq:exact-extension}.

\subsection{Arbitrary restored regions}

\begin{corollary}[Recovery of an arbitrary deleted region]\label{an:thm:large-recovery}
Under Assumptions~\ref{an:ass:effective} and~\ref{an:ass:strip}, let $W,Y,Z\subseteq\Omega$ be disjoint,
with $W,Z\ne\varnothing$. Enumerate $Z=\{z_1,\ldots,z_m\}$ in any order, where $m=|Z|$.
Let $C_Q^{(1)}$ denote $C_Q$ from
Theorem~\ref{an:thm:local-petz} with $b_*=1$. Then
\begin{align}
 \norm{\PP^\Omega_{Z|WY}-\id_W\otimes\PP^\Omega_{Z|Y}}_\diamond
 &\le 2C_Q^{(1)}
            \sum_{j=1}^{m}e^{-\kappa d(W,z_j)}\label{an:eq:large-diamond}\\
 &\le 2C_Q^{(1)}|Z|e^{-\kappa d(W,Z)},\nonumber\\
 \norm{(\id_W\otimes\PP^\Omega_{Z|Y})(\rho^\Omega_{WY})
                    -\rho^\Omega_{WYZ}}_1
 &\le 2C_Q^{(1)}|Z|e^{-\kappa d(W,Z)}.\label{an:eq:large-state}
\end{align}
Each right-hand side can also be replaced by its minimum with $2$.
\end{corollary}
\begin{proof}
For the ordering of Eq.~\eqref{an:eq:ordered-z}, compare
\begin{equation}\label{an:eq:two-one-site}
 \widehat\PP_j:=\PP^\Omega_{z_j|WYZ^{(j-1)}},\qquad
 \PP_j:=\PP^\Omega_{z_j|YZ^{(j-1)}}.
\end{equation}
Theorem~\ref{an:thm:local-petz}, uniformly in the growing retained region,
gives
\begin{equation}\label{an:eq:step-diff}
 \|\widehat\PP_j-\id_W\otimes\PP_j\|_\diamond
 \le 2C_Q^{(1)}e^{-\kappa d(W,z_j)}.
\end{equation}
Unprocessed sites and $\Omega\setminus(WYZ)$ are traced in the defining
marginals, which are included in Assumption~\ref{an:ass:effective}.
Proposition~\ref{an:thm:factorization} and Eq.~\eqref{an:eq:hybrid}
give the channel bound. Applying it to $\rho^\Omega_{WY}$ gives the
state bound by exact extension. Both distances are also at most $2$.
\end{proof}
\begin{corollary}[Additional implementation errors]\label{an:cor:impl}
In Corollary~\ref{an:thm:large-recovery}, let $\PP_j$ be the one-site
maps in \eqref{an:eq:two-one-site}, in the chosen ordering of $Z$.
Suppose $\widetilde{\PP}_j$ is CPTP and
$\norm{\widetilde{\PP}_j-\PP_j}_\diamond\le\eta_j$.
For an input state $\widetilde\rho_{WY}$ with
$\norm{\widetilde\rho_{WY}-\rho^\Omega_{WY}}_1\le\eta_{\rm in}$,
define
\[
 \widetilde\omega_{WYZ}:=
 \bigl[\id_W\otimes(\widetilde\PP_m\circ\cdots\circ
       \widetilde\PP_1)\bigr](\widetilde\rho_{WY}).
\]
Then
\begin{equation}\label{an:eq:all-errors}
\begin{split}
 \norm{\widetilde\omega_{WYZ}-\rho^\Omega_{WYZ}}_1
 \le{}&\eta_{\rm in}+\sum_j\eta_j\\
 &+2C_Q^{(1)}\sum_j e^{-\kappa d(W,z_j)}.
\end{split}
\end{equation}
\end{corollary}

\begin{proof}
Use trace-norm contraction for the input error, Eq.~\eqref{an:eq:hybrid} for the implemented
channels, and Corollary~\ref{an:thm:large-recovery} for the exact local recovery.
\end{proof}

\subsection{Local preparation and seed errors}
\label{sec:local-representation}

Choose a seed region $A_{\mathrm{seed}}\subseteq\Lambda$ and an ordering
$z_1,\ldots,z_n$ of its complement, where $n=N-|A_{\mathrm{seed}}|$.
Write $\Lambda_k=A_{\mathrm{seed}}\cup\{z_1,\ldots,z_k\}$ and $\Lambda_0=A_{\mathrm{seed}}$.
For two positive integer radii $R,\ell$ define
\begin{equation}\label{eq:patch-regions}
\begin{aligned}
 Y_k&:=\Lambda_{k-1}\cap B_R(z_k),\\
 W_k&:=\Lambda_{k-1}\setminus Y_k,\\
 \Omega_k&:=B_{R+\ell}(z_k),\qquad
 \PP_k^{\mathrm{loc}}:=\PP^{\Omega_k}_{z_k|Y_k}.
\end{aligned}
\end{equation}
The defining local states are $\rho^{\Omega_k}_{Y_kz_k}$ and
$\rho^{\Omega_k}_{Y_k}$, not generally $\gamma_{Y_kz_k}$ and its marginal.
In particular, $\Omega_k\setminus(Y_kz_k)$ is retained in the coefficient
construction as a traced complement, not silently removed from the
Hamiltonian.

Use $C_Q^{(1)}$ from Theorem~\ref{an:thm:local-petz} with $b_*=1$ and fix
\begin{equation}\label{eq:preparation-locality-constant}
 C_{\mathrm{loc}}^Q:=\max\{1,C_Q^{(1)},
                    2\sqrt{C_{\mathrm{marg}}}c_{\mathrm{vol}}^{3/4}\}.
\end{equation}
Since $|Y_kz_k|\le c_{\mathrm{vol}}(1+R)^D$, this constant converts
Theorem~\ref{thm:two-cut} to the geometric bounds below.

\begin{corollary}[Approximation by one-site channels and a seed marginal]
\label{thm:analytic-assembly}
Under the assumptions of Theorem~\ref{thm:two-cut},
\begin{equation}\label{eq:state-product}
\begin{split}
 &\Bigl\|\gamma_\Lambda
  -\bigl[(\id_{W_n}\otimes\PP_n^{\mathrm{loc}})\circ\cdots\circ
  (\id_{W_1}\otimes\PP_1^{\mathrm{loc}})\bigr]
  (\rho^\Lambda_{A_{\mathrm{seed}}})\Bigr\|_1\\
 &\qquad\le 2nC_{\mathrm{loc}}^Q
       \left[e^{-\kappa R}+(1+R)^{3D/4}e^{-\kappa\ell/16}\right].
\end{split}
\end{equation}
If the seed has trace-norm error $\eta_{\mathrm{seed}}$ and each local
channel is replaced by a CPTP map of diamond error at most $\eta_k$,
add $\eta_{\mathrm{seed}}+\sum_k\eta_k$ to the right-hand side.
In particular, $A_{\mathrm{seed}}=\varnothing$ gives preparation from
the scalar state $1$, without any seed-state preparation requirement.
\end{corollary}

\begin{corollary}[Preparation by local Petz isometries]\label{cor:local-purification}
Under the assumptions of Theorem~\ref{thm:two-cut}, use the empty seed and the regions in \eqref{eq:patch-regions}. Let
$V_k^{\mathrm{loc}}:=V^{\Omega_k}_{z_k|Y_k}$, extended by identities
on $W_k$ and on all previously retained reference systems. Starting
from the scalar state and keeping every new $\mathcal E_{z_k}$,
these isometries produce a vector $|\Psi_{\mathrm{loc}}\rangle$ with
\begin{equation}\label{eq:local-pure-error}
 \norm{|\Psi_{\mathrm{loc}}\rangle-|\Gamma_\beta\rangle}_2
 \le NC_{\mathrm{loc}}^Q
 \left[e^{-\kappa R}+(1+R)^{3D/4}e^{-\kappa\ell/16}\right].
\end{equation}
If the $k$th implementation has operator-norm isometry error $\xi_k$,
including leakage into work registers, add $\sum_k\xi_k$ to the
right-hand side and tensor the target with zero work registers.
\end{corollary}
\begin{proof}[Proof of Corollaries~\ref{thm:analytic-assembly} and~\ref{cor:local-purification}]
Let $\widehat\PP_k=\PP^\Lambda_{z_k|\Lambda_{k-1}}$ and let
$\widehat V_k$ be its specified isometry. Exact extension prepares
$\rho_{\Lambda_k}^\Lambda$; Proposition~\ref{prop:canonical-extension}
gives the corresponding canonical purifications in the empty-seed case.
The geometry of Eq.~\eqref{eq:patch-regions} satisfies
$d(W_k,z_k)>R$ and $d(Y_kz_k,\Lambda\setminus\Omega_k)\ge\ell$.
Theorem~\ref{thm:two-cut} and Eq.~\eqref{eq:preparation-locality-constant}
therefore bound the one-step vector error on each exact Gibbs
purification by
\[
 C_{\mathrm{loc}}^Q
       [e^{-\kappa R}+(1+R)^{3D/4}e^{-\kappa\ell/16}],
\]
and the one-step trace-norm error on each exact marginal by twice
this expression.
Apply Eq.~\eqref{eq:input-sequence-telescoping} for isometries, retaining
all earlier references, and its CPTP version for the mixed-state
construction from the exact seed marginal. This gives the two locality
bounds without requiring a channel-norm comparison across Gibbs volumes.
For implemented isometries, separately telescope their operator-norm
errors relative to the finite-patch isometries, with fresh zero work
registers included in the target. For implemented CPTP maps, telescope
their diamond-norm errors relative to the same finite-patch maps by
Eq.~\eqref{an:eq:hybrid}. Finally, trace-norm contraction adds the
seed error $\eta_{\mathrm{seed}}$.
\end{proof}
A seed marginal can contain correlations between separated patches; it
is not assumed to factorize. The empty seed avoids a separate seed-state
preparation requirement.
\subsection{Disjoint patches}
\label{sec:coloring}

For an integer $r\ge1$, join two distinct sites whenever their distance
is at most $2r$. The resulting conflict graph has maximum degree at most
$\max_x|B_{2r}(x)|-1$. Greedy coloring in a fixed vertex order partitions
$\Lambda$ into classes $\Lambda^{(1)},\ldots,
\Lambda^{(n_{\mathrm{col}})}$ with
\begin{equation}\label{eq:color-count}
 n_{\mathrm{col}}\le\max_{x\in\Lambda}|B_{2r}(x)|
                    \le c_{\mathrm{vol}}(1+2r)^D.
\end{equation}
Equal-colored sites have distance greater than $2r$, so their radius-$r$
balls are disjoint. 

For $r=R+\ell$, fix such a coloring and order the sites color by color,
using any fixed order within each class. Define
\begin{equation}\label{eq:color-regions}
\begin{aligned}
 \Lambda^{(<c)}&:=\bigcup_{a<c}\Lambda^{(a)},\\
 Y_x&:=\Lambda^{(<c)}\cap B_R(x)
                         &&(x\in\Lambda^{(c)}),\\
 \Omega_x&:=B_{R+\ell}(x),\qquad
 V_x^{\mathrm{loc}}:=V^{\Omega_x}_{x|Y_x}.
\end{aligned}
\end{equation}
Because sites of a common color have distance greater than $2(R+\ell)$,
none is in $B_R(x)$ for another. Thus the sets $Y_x$ agree exactly with
those in \eqref{eq:patch-regions} for the induced site order. Each map adds just $x$. In particular, its input does not change as
other sites of the same color are restored. The equal-colored maps act
on disjoint physical registers and retain distinct new reference systems,
so their parallel product is the serial product for this site order.

For $0<\varepsilon<1/2$, take
\begin{equation}\label{eq:global-radii}
\begin{aligned}
 R&=\left\lceil\frac1\kappa
          \log\frac{16NC_{\mathrm{loc}}^Q}{\varepsilon}\right\rceil,\\
 \ell&=\left\lceil\frac{16}{\kappa}
          \log\frac{16NC_{\mathrm{loc}}^Q(1+R)^{3D/4}}
                        {\varepsilon}\right\rceil.
\end{aligned}
\end{equation}
Both radii are $O(\log(N/\varepsilon))$; the additional
$\log(1+R)$ term in $\ell$ does not change the patch-volume scaling.
The layered product, with identities on earlier registers understood,
\begin{equation}\label{eq:ideal-layered-preparation}
 |\Psi_{\mathrm{loc}}\rangle=
 \left(\prod_{x\in\Lambda^{(n_{\mathrm{col}})}}V_x^{\mathrm{loc}}\right)
 \cdots
 \left(\prod_{x\in\Lambda^{(1)}}V_x^{\mathrm{loc}}\right)(1)
\end{equation}
obeys $\||\Psi_{\mathrm{loc}}\rangle-|\Gamma_\beta\rangle\|_2
\le\varepsilon/8$ by Corollary~\ref{cor:local-purification}.
There are $O([\log(N/\varepsilon)]^D)$ recovery layers and
$O([\log(N/\varepsilon)]^D)$ sites per patch. This is a statement about
specified isometries and their spatial supports, not their elementary-gate
cost. Appendix~\ref{sec:classical-construction} constructs these local
maps; Appendix~\ref{sec:global-complexity} accounts for implementation
errors and resources.

\section{Constructive implementation of one-patch Petz recovery}
\label{sec:classical-construction}\label{sec:circuit-primitives}

The preceding appendices specify a product of finite-patch Petz isometries
whose coherent application approximates the canonical purification. We now
compile one of those isometries. The first ingredient turns imaginary-time locality into exponential-product
circuits and normalization-controlled one-site contractions, encoding the
endpoint corrections from $K_0$ and fixed references. The second uses the
bounded-function integral to construct the noncommutative square-root
factors and complete the specified Petz isometry. The input is an interaction list,
unlike the supplied state encodings and spectral promises of the general
Petz algorithm in Ref.~\cite[Theorem~1]{GilyenEtAl2022}.

\subsection{Reference input and computational model}
\label{sec:reference-input}

A local construction acts on a supplied patch $\Omega$ of $M=|\Omega|$
sites. Its isometry-error budget is $\eta/2$, its induced channel-error budget is $\eta\in(0,1/4)$, and the failure probability
of its input-independent randomized preprocessing is
$\fail\in(0,1/4)$. In the reference and circuit construction the spatial
weight is $\mu=\mu_*$. We use linear combinations of unitaries (LCU) and
quantum singular value transformation (QSVT). A PREPARE unitary prepares a
coefficient-label state; SELECT applies the labelled operation conditionally
on the label.

\begin{definition}[Projected unitary encoding; {\cite{GilyenEtAl2019}}]\label{def:encoding}
Let $A:\mathcal H_{\mathrm{in}}\to\mathcal H_{\mathrm{out}}$.
Specify isometric embeddings $J_{\mathrm{in}}$ and $J_{\mathrm{out}}$
from these spaces into a common register space $\mathcal K$, and set
$\Pi_{\mathrm{in}}=J_{\mathrm{in}}J_{\mathrm{in}}^\dagger$ and
$\Pi_{\mathrm{out}}=J_{\mathrm{out}}J_{\mathrm{out}}^\dagger$.
A unitary $U$ on $\mathcal K$ is an $(\alpha,\zeta)$ encoding of $A$ if
\begin{equation}\label{eq:encoding}
 \norm{A-\alpha J_{\mathrm{out}}^\dagger UJ_{\mathrm{in}}}\le\zeta.
\end{equation}
Thus the encoded block acts from $\mathcal H_{\mathrm{in}}$ to
$\mathcal H_{\mathrm{out}}$, even when their dimensions differ.
The known number $\alpha>0$ is the normalization factor; it need not equal
$\norm A$. A block encoding is the case in which the embeddings append
specified zero ancillas. All embeddings used below have explicit unitary
implementations after adding zero registers. Their controlled versions,
inverses, and the reflections about the two images are included in the cost.
For a square Hermitian target, averaging a block encoding and its adjoint
using one control qubit gives an approximation with a Hermitian projected
block. The full implementing unitary need not be Hermitian.
\end{definition}

We count elementary one- and two-qubit gates, work qubits, classical bit
operations, and independent quantum measurements used in preprocessing.
A fixed $q$-dimensional site occupies $n_q$ qubits; operations outside its
encoded subspace are extended to unitaries on the padded registers.

A compiled local circuit is a unitary acting on its retained input,
new output sites, reference systems, and zero-initialized work qubits.
It uses no measurements during its application to the input. Randomized
list compression and independent measurements of normalization constants are preprocessing;
the resulting finite classical data fix the unitary before it is applied.
Circuit depth counts layers of one- and two-qubit gates with disjoint
register supports. Different patches can use independent work registers.
The gate model permits gates between these specified registers; a chosen
routing architecture with polynomial overhead per patch preserves the
asymptotic bounds below. Normalization-estimation trials, compiled quantum gates,
classical bit operations, and peak number of qubits in use are charged separately.

The local-term access model, including support arithmetic and a controlled
PREPARE operation over interaction terms incident on a specified site,
with all finite-precision costs included, is stated in
Assumption~\ref{ass:microscopic-access} in
Sec.~\ref{sec:microscopic-model}. The following deletion chain defines
the coefficients of one local recovery map. It is not an operation performed
on the physical state during the site-addition sequence.

\begin{definition}[Nested marginals]\label{cr:def:chain}
Choose disjoint nonempty blocks $B_1,\ldots,B_m$, where $1\le m\le M$ and $|B_j|\le b_*$ for a fixed constant $b_*$.  Put
\begin{equation}\label{cr:eq:chain}
 A_0=\Omega,\qquad A_j=A_{j-1}\setminus B_j,\qquad
 \rho_j=\Tr_{\Omega\setminus A_j}\gamma_\Omega.
\end{equation}
Define
\begin{equation}\label{cr:eq:K}
 K_j=-\log\rho_j-\trn{A_j}(-\log\rho_j)\I_{A_j}.
\end{equation}
Thus $\trn{A_j}(K_j)=0$ and $\rho_j=e^{-K_j}/\Tr(e^{-K_j})$;
$\E_B=d_B^{-1}\Tr_B$ remains the normalized partial trace.
Each preparation step restores one site $z$. For the map $Y\to Yz$,
let $C=\Omega\setminus(Yz)$, delete the sites of $C$ first, and choose
$B_m=\{z\}$. Then $A_{m-1}=Yz$, $A_m=Y$, and $m\le M$.
If $C=\varnothing$, this is the single deletion $m=1$.
Only the final recovery is applied to the physical input. Earlier deletions
are used to construct its coefficients, with temporary registers for the
traced sites. The local construction also allows a final block of fixed
size $|B_m|\le b_*$, but the global preparation uses $|B_m|=1$.
\end{definition}

The effective-interaction hypotheses give the following bounds for this
chain; the exact decompositions are needed only for analysis.

\begin{proposition}[Analytic bounds for the exact marginal family]
\label{cr:ass:analytic}
Under Assumptions~\ref{an:ass:effective} and \ref{an:ass:strip},
each $K_j$ has weighted interaction norm at most
$\Delta_0<\lambda$, and
\begin{equation}\label{cr:eq:exact-U}
 U_j^{\mathrm{ex}}:=K_{j-1}-K_j\otimes\I_{B_j}
\end{equation}
has a specified decomposition of weighted total sum at most
$2u_0b_*$. The decompositions are used only in proofs.
\end{proposition}

\begin{proof}[Proof of Proposition~\ref{cr:ass:analytic}]
The first bound follows from \eqref{an:eq:heff} and
\eqref{an:eq:strip}. Before centering, the second difference is
$U^\Omega_{B_j|A_j}$, whose weighted total sum is at most $u_0b_*$
by Lemma~\ref{an:lem:relative}. Centering subtracts the normalized trace
of this difference. That scalar has modulus at most its operator norm,
which is at most the same weighted total sum. This proves the stated bound.
\end{proof}

Here $u_0$ is specified in Eq.~\eqref{an:eq:root-u}.
The local effective-interaction assumptions already provide the existence
of a bounded-strength decomposition. On a finite patch, a strict strength
margin also permits finite-precision approximations of that decomposition.
Spatial decay and stability make finite-neighborhood calculations
a natural route to such coarse data, consistent with the local
reconstruction viewpoint of Refs.~\cite{ScaletEtAl2025,BluhmCapelPerezHernandez2025}.
The following input condition specifies this algorithmic realization
and its generation cost; it does not introduce a further correlation
observable.

\paragraph{Reference data.}
For $j=1,\ldots,m$, a reference list describes a centered Hermitian operator
\begin{equation}\label{cr:eq:reference-input}
 \overline L_j=\overline c_j\I+
       \sum_{X\in\mathcal I_j}\overline\ell_{j,X},\qquad
 \overline a_{j,X}\ge\|\overline\ell_{j,X}\|.
\end{equation}
It supplies finite nonempty supports $X\subseteq A_j$, finite-precision
matrices, and the indicated norm bounds. All numerical errors are included
in the reference accuracy. The cost $\mathcal C_{\mathrm{ref}}(\Omega)$
counts generation, finite-precision construction, and scanning of all
$m$ lists. The exact microscopic $K_0$ is supplied separately through
Assumption~\ref{ass:microscopic-access}. We write $L_0:=K_0$ only to
use uniform chain formulas below; it is never approximated or compressed.

\begin{assumption}[Reference input and local preprocessing]
\label{cr:ass:classical}
\textit{(i) Reference quality.}
The lists in Eq.~\eqref{cr:eq:reference-input} are supplied with
\begin{align}
 \sup_x\sum_{X\ni x}\overline a_{j,X}w_{\lambda,\mu}(X)
       &\le J_{\mathrm{ref}}<\lambda,
          \label{cr:eq:reference-strength}\\
 \|\overline L_j-K_j\|&\le r_{\mathrm{ref}},\qquad 1\le j\le m.
          \label{cr:eq:reference-error}
\end{align}
\textit{(ii) Local generation.}
When the classical generation bound is asserted, a specified procedure
produces these same lists, for every patch and deletion chain used, with
\begin{equation}\label{eq:classical-local-error}
 r_{\mathrm{ref}}\le C_{\mathrm{cl}}m
       e^{-c_{\mathrm{cl}}\ell_{\mathrm{cl}}},\qquad
 \ell_{\mathrm{cl}}\ge1,
\end{equation}
for fixed $C_{\mathrm{cl}},c_{\mathrm{cl}}>0$, within the local arithmetic
and finite-precision bounds of Sec.~\ref{sec:local-reference-model}.
\end{assumption}

Theorem~\ref{cr:thm:local-main} uses only clause (i) of
Assumption~\ref{cr:ass:classical}, with the reference-generation cost
$\mathcal C_{\mathrm{ref}}(\Omega)$ kept explicit. Clause (ii) is used in
Proposition~\ref{prop:classical-budget} to bound this cost; dense computations
are allowed with cost polynomial in each neighborhood's Hilbert-space dimension.
Neither clause requires references at the final recovery tolerance;
Eq.~\eqref{cr:eq:theta-p} fixes the required inverse-polynomial accuracy.

\paragraph{Finite-matrix synthesis.}\label{sec:finite-matrix-model}
We use the standard circuit and projected-encoding model
\cite{ShendeBullockMarkov2006,GilyenEtAl2019}. A known contraction on
$s$ sites has a unitary dilation compiled at cost
$\exp(O(s))\poly(\log(1/\zeta))$. Local matrices are Hermitian; taking
their Hermitian part when necessary increases neither support nor norm.
Products, linear combinations, and QSVT use the explicit encodings below.
Only the retained small matrices and short lists are used coherently after
compression. Classical dense arithmetic and every finite-precision
quantum operation are charged in their respective resource budgets.

Choose a fixed $\kappa_{\mathrm{ref}}>0$ such that
\begin{equation}\label{cr:eq:margins}
 J_*:=\max\{\Delta_0,J_{\mathrm{ref}}+2\kappa_{\mathrm{ref}}\}<\lambda.
\end{equation}
Choose a fixed $t_*>1$ with $t_*J_*<\lambda$, and set
\begin{equation}\label{cr:eq:theta-p}
\begin{gathered}
 \theta_{\mathrm{int}}:=1-t_*^{-1}>0,\qquad
 p_{\mathrm{ref}}>6/\theta_{\mathrm{int}},\\
 e_M:=c_e(M+1)^{-p_{\mathrm{ref}}}.
\end{gathered}
\end{equation}
The positive constant $c_e$ is chosen sufficiently small as specified in Lemma~\ref{lem:reference-stability}; it depends only on the fixed bounds and margins.

For the chain in Definition~\ref{cr:def:chain}, write
\begin{equation}\label{cr:eq:petz}
 Q_j=Q^\Omega_{B_j|A_j},\qquad
 V_j=V^\Omega_{B_j|A_j},\qquad
 \PP_j=\PP^\Omega_{B_j|A_j}.
\end{equation}
These use the fixed bases of Definition~\ref{an:def:petz}.
The family notation is convenient for the algebraic identities below;
only $V_m=V^\Omega_{z|Y}$ is compiled as the physical recovery.
In particular, intermediate $V_j$ and the auxiliary $X_j,T_j$ are not
inputs to a subsequent update.

\subsection{Microscopic access and elementary encoding rules}\label{sec:microscopic-model}

The input model specializes PREPARE/SELECT access~\cite{GilyenEtAl2019} to incident interaction terms and uses reversible bookkeeping~\cite{Bennett1973}. Its availability and cost are specified explicitly.

\begin{assumption}[Access to the microscopic Hamiltonian]
\label{ass:microscopic-access}
The bounded-degree adjacency list is supplied explicitly or its neighbor
queries are computable in $\poly(\log N)$ bit operations. Balls of a
specified integer radius can therefore be enumerated by breadth-first
search with polynomial cost in their volume and $\log N$.
Choose a fixed ordering of sites and nonzero interaction labels. Labels with
the same support and type may be combined, so a site has only a bounded
number of incident labels. For every $h_X$, a support list, a bound
$\alpha_{h_X}\ge\norm{h_X}$, and its fixed-dimensional matrix entries can be
computed to $b$ bits in $\poly(\log N,b)$ classical time. Put
\begin{equation}
 g:=\sup_x\sum_{X\ni x}\alpha_{h_X}=O(1),\qquad
 g_x:=\sum_{X\ni x}\alpha_{h_X}.
\end{equation}
For $g>0$, the controlled PREPARE unitary selects interaction terms
containing the specified site and acts as
\begin{equation}\label{eq:incidence}
 |x\rangle|0\rangle\longmapsto |x\rangle
 \left(\sum_{X\ni x}\sqrt{\frac{\alpha_{h_X}}{g}}|X\rangle
       +\sqrt{1-\frac{g_x}{g}}|\perp\rangle\right).
\end{equation}
$|\perp\rangle$ denotes an invalid label. Indexed SELECT applies an $\alpha_{h_X}$-normalized block encoding of $h_X$ and uses a zero projected block on invalid labels.
Both have controlled implementations of gate and ancilla cost
$\poly(\log N,\log(1/\zeta))$ at unitary error $\zeta$.
All nonzero input weights have finite polynomial-bit descriptions.
\end{assumption}

If $g=0$, the target is maximally mixed and each recovery simply appends
maximally entangled site--reference pairs; hence we consider $g>0$ below.

\paragraph{Microscopic encoding and support operations.}
To encode $H_A$ for a supplied nonempty region $A$, prepare $x\in A$
uniformly and apply Eq.~\eqref{eq:incidence}, accepting $X\subseteq A$
only when $x$ is the first site of $X$. The accepted squared amplitude
is $\alpha_{h_X}/(g|A|)$, so SELECT and inverse preparation encode
$H_A/(g|A|)$ at cost $\poly(|A|,\log N,\log(1/\zeta))$.
For the nested-commutator construction in
Sec.~\ref{app:relative-circuits}, supports are stored as sorted site-label
lists padded to the chosen truncation size. Intersection tests,
first-site selection, and unions are computed reversibly while retaining
their inputs and computation records~\cite{Bennett1973}; their costs are
polynomial in the support size and $\log N$. Patch tests may scan the
explicit $M$-site list, at charged cost $\poly(M,\log N)$. The scalar
functions and controlled rotations used below are evaluated reversibly
at cost polynomial in their argument bit lengths and desired precision.

\paragraph{Standard encoding rules.}
We use the product and linear-combination constructions of
Ref.~\cite[Lemmas~52--53]{GilyenEtAl2019}. For encoded operators $A_i$
with normalization factors $\alpha_i$, the factors multiply for a product,
while $\sum_i c_iA_i$ has normalization $\sum_i|c_i|\alpha_i$. Products use fresh projection ancillas and the
intermediate subspaces of Definition~\ref{def:encoding}; the Gram product
needed for the corrections is specified in
Eq.~\eqref{eq:gram-projection}. The error bounds are
Eqs.~\eqref{eq:general-product-error}--\eqref{eq:product-error}, including
$n_{\mathrm{calls}}\zeta$ for $n_{\mathrm{calls}}$ unitary calls of error
at most $\zeta$. QSVT applies bounded polynomials of the required parity
with query cost linear in their degree; even and odd parts are combined
by LCU. We use uniform singular-value amplification with its stated
output margin~\cite[Theorem~30]{GilyenEtAl2019}, allocating input errors
within that margin.

\paragraph{Classical phase synthesis.}
The bounded approximation polynomials used below have explicitly
computable coefficients and the required parity. Their phase sequences
are computed before recovery with classical cost and bit precision
polynomial in the degree $d_{\mathrm{pol}}$ and $\log(1/\xi)$, using
real QSP~\cite[Corollary~10]{GilyenEtAl2019} and the finite-precision
product decomposition of Ref.~\cite[Sec.~4]{Haah2019}.
Rescaling within the approximation budget provides a strict bound below
one when needed. Rotation errors $O(\xi/(d_{\mathrm{pol}}+1))$ contribute
$O(\xi)$ by Eq.~\eqref{eq:product-error}. For the polynomially many
quadrature nodes, precomputation and coherent selection of their phase
sequences also have polynomial cost. The Dyson time register instead
uses the binary arithmetic of Lemma~\ref{lem:time-ordered-lcu}.

\paragraph{Local implementation guarantee.}
The following is the central one-patch result. The reference tolerance
in Eq.~\eqref{cr:eq:theta-p} is fixed before the final quantum accuracy
is chosen. The subsequent sections prove the theorem by constructing
the marginal corrections and then the square-root factors.

\begin{theorem}[Efficient local Petz implementation]\label{cr:thm:local-main}
Under Assumptions~\ref{an:ass:effective}, \ref{an:ass:strip}, \ref{cr:ass:classical}(i), and~\ref{ass:microscopic-access}, assume
\begin{equation}\label{cr:eq:accuracy-input}
 r_{\mathrm{ref}}\le e_M.
\end{equation}
Input-independent randomized preprocessing constructs a unitary
$\mathscr U_m$ for the specified isometry $V_m$ in
Eq.~\eqref{cr:eq:petz}. Let $\mathsf A_m$ be its work register,
and define its initialized action by
\begin{equation}\label{eq:local-initialized-action}
 \widetilde V_m|\psi\rangle
 :=\mathscr U_m
   (|\psi\rangle_{A_m}|0\rangle_{B_m\mathcal E_m\mathsf A_m}).
\end{equation}
Except with probability at most $\fail$,
\begin{equation}\label{eq:local-isometry-accuracy}
 \norm{\widetilde V_m-V_m\otimes|0\rangle_{\mathsf A_m}}
 \le\eta/2.
\end{equation}
Tracing out $\mathcal E_m\mathsf A_m$ defines a channel
$\widetilde\PP_m$ with
$\norm{\widetilde\PP_m-\PP_m}_\diamond\le\eta$.
The quantum gate and work-qubit costs, including independent normalization-estimation
trials, are bounded by
\begin{equation}\label{cr:eq:main-local-cost}
 \poly\!\left(M,m,\log N,\log\frac1\eta,
                         \log\frac1\fail\right).
\end{equation}
The classical cost is $\mathcal C_{\mathrm{ref}}(\Omega)$ plus a
polynomial in the parameters of \eqref{cr:eq:main-local-cost}.
The references require only the precision \eqref{cr:eq:accuracy-input},
independently of $\eta$. The theorem includes a one-site output
$B_m=\{z\}$ and any traced complement $\Omega\setminus(Yz)$,
including the empty complement and $Y=\varnothing$.
\end{theorem}

\subsection{Reference compression and imaginary-time stability}
\label{app:references}

The supplied references must be both inexpensive to access and stable
under the nonunitary operations used in the compiler. Compression first
reduces each list to polynomially many small-support matrices while
preserving its accuracy and weighted interaction bound. We then show
that the remaining reference error stays small under imaginary-time
conjugation. The first guarantee supplies cheap term encodings; the
second controls the reference exponential products and, after
normalization, the cumulative Gram corrections.

\paragraph{Compression and quantum access.}
Matrix and scalar concentration~\cite{Tropp2012,Hoeffding1963} are applied
to a weighted sampling of the supplied interactions. The choice of
sampling weights below preserves both operator accuracy and the
per-site interaction bound needed by the imaginary-time estimates.

\begin{lemma}[Compressed references and quantum access]
\label{cr:lem:sparsify}
Under Assumption~\ref{cr:ass:classical}(i), suppose
$r_{\mathrm{ref}}\le e_M$, with the margins and reference tolerance in
Eqs.~\eqref{cr:eq:margins}--\eqref{cr:eq:theta-p}.
Classical randomized preprocessing produces centered references $L_j$,
$j=1,\ldots,m$, such that, with probability at least $1-\fail/4$,
\begin{equation}\label{cr:eq:L-properties}
 \|L_j-K_j\|\le 2e_M,\qquad
 J^{\mathrm{bd}}_{L_j}\le J_*<\lambda.
\end{equation}
Each list has $\poly(M,\log(m/\fail))$ terms, supported on
$O(\log(M+1))$ sites, and $\|L_j\|=O(M)$. We keep $L_0=K_0$ unchanged.
Controlled term encodings, incident-term PREPARE operations of the form
in Eq.~\eqref{eq:incidence}, and support updates have per-call cost
\begin{equation}\label{cr:eq:list-gates}
 \poly\!\left(M,\log N,\log\frac m\fail,\log\frac1\zeta\right)
\end{equation}
for $0<\zeta<1/2$. Generating and scanning the original lists is
charged to $\mathcal C_{\mathrm{ref}}(\Omega)$; the subsequent
small-matrix processing and circuit compilation have polynomial cost.
\end{lemma}

\begin{proof}
For one supplied list, suppress $j$ and write
$\overline L=c\I+\sum_X\ell_X$, with supplied bounds $a_X\ge\|\ell_X\|$.
Set $e=e_M/4$ and $w_X=\wt(X)$, so
$W_{\mathrm{all}}:=\sum_Xa_Xw_X\le MJ_{\mathrm{ref}}$.
Discard terms with
\[
 |X|>k_c:=\max\!\left\{1,\left\lceil
 \lambda^{-1}\log\frac{16(MJ_{\mathrm{ref}}+1)}e
 \right\rceil\right\},
\]
and, when $\mu>0$, those with
$\mu\diam(X)>\log(16(W_{\mathrm{all}}+1)/e)$.
Each cutoff costs at most $e/16$ in operator norm. The retained set
$\mathcal R$ has support size $O(\log(M+1))$ and weights polynomially
bounded in $M$.
Put $W=\sum_{X\in\mathcal R}a_Xw_X$; if $W\le e/16$, use the centered
zero operator. Otherwise sample $n$ independent labels according to
\begin{equation}\label{cr:eq:sampling}
 \begin{gathered}
 p_X=\frac{a_Xw_X}{W},\qquad
 Y_t=\frac{\ell_{X_t}}{np_{X_t}},\qquad
 \alpha_t=\frac{W}{nw_{X_t}}.
 \end{gathered}
\end{equation}
The mean of $\sum_tY_t$ is the retained interaction sum, and
$\|nY_t\|\le W$. At a site $x$, its sampled weighted coefficient sum is
$(W/n)\sum_t\mathbf1_{\{x\in X_t\}}$, with mean at most
$J_{\mathrm{ref}}$. Matrix Hoeffding~\cite[Theorem~1.3]{Tropp2012} and
scalar Hoeffding~\cite{Hoeffding1963}, respectively, control the operator
error and all sitewise excesses. Taking
\[
 n=O\!\left((1+MJ_{\mathrm{ref}})^2
 (e^{-2}+\kappa_{\mathrm{ref}}^{-2})
 \left[M+\log\frac m\fail\right]\right)
\]
makes the combined failure probability at most $\fail/(4m)$, with
sampling error at most $e/16$ and sitewise excess at most
$\kappa_{\mathrm{ref}}$. Only $\log(q^M)=M\log q$ enters the matrix
concentration bound.

Subtract the normalized trace to center the sampled sum; this at most
doubles its error and leaves the non-scalar weights unchanged.
Finite-probability sampling and matrix rounding are allocated the
remaining tolerance, with conservative norm bounds using a further
$\kappa_{\mathrm{ref}}$ of slack. The resulting error from
$\overline L$ is at most $2e$, and its weighted coefficient bound is at
most $J_{\mathrm{ref}}+2\kappa_{\mathrm{ref}}$. Combining with
$\|\overline L-K_j\|\le e_M$ and taking a union bound over all $m$
lists proves Eq.~\eqref{cr:eq:L-properties}.

Since $e_M$ is inverse-polynomial in $M$, $n$ is polynomial and every
retained matrix acts on $O(\log(M+1))$ sites. Its generic unitary dilation
and synthesis therefore cost $\poly(M,\log(1/\zeta))$
\cite{ShendeBullockMarkov2006}. Explicit multiplexing over the compressed
list and scans of its incident terms give Eq.~\eqref{cr:eq:list-gates}.
The retained matrices, weights, and rotations need only polynomial bit
precision; negligible weights can be rounded within the same budgets.
Finally, centering and the weighted bound imply $\|L_j\|=O(M)$.
\end{proof}

\paragraph{Stability under imaginary-time conjugation.}
Compression alone gives an additive bound on $L_j-K_j$.
The exponential-product circuit and its normalization require this
error to remain small after conjugation by reference evolutions.
The following interpolation argument supplies that stronger control,
without requiring access to the exact interaction coefficients.

The complex-time estimate of
Ref.~\cite[Proposition~2.1]{BluhmCapelPerezHernandez2025}, applied termwise,
gives the following bound. If $L=L^\dagger$ has weighted interaction norm
at most $J<\lambda$ and the specified decomposition of $O$ has weighted
total sum at most $B$, then for real $t$
\begin{equation}\label{cr:eq:complex}
 \|e^{-tL/2}Oe^{tL/2}\|\le\frac{\lambda B}{\lambda-|t|J},
 \qquad |t|J<\lambda.
\end{equation}
The same bound holds at time $t+iu$, independently of $u\in\mathbb R$,
because an imaginary shift is a unitary conjugation. Combined with the
three-lines theorem \cite{Hirschman1952,SutterBertaTomamichel2017}, this
estimate turns the supplied operator-norm accuracy into the stability
needed by the normalization experiments.

\begin{lemma}[Stability of the reference states]\label{lem:reference-stability}
On the successful event of Lemma~\ref{cr:lem:sparsify}, put
$\mathcal D_j^{\mathrm{err}}=L_j-K_j$, where
$\mathcal D_0^{\mathrm{err}}=0$, and
$\sigma_j=e^{-L_j}/\Tr(e^{-L_j})$. Let $L$ be any compressed reference
Hamiltonian used below with weighted coefficient bound at most $J_*$,
with identities appended to put $L$ and $\mathcal D_j^{\mathrm{err}}$
on the same space. For a fixed constant $C_2$,
\begin{equation}\label{cr:eq:h}
 \begin{gathered}
 \sup_{|t|\le1}\|e^{-tL/2}\mathcal D_j^{\mathrm{err}}e^{tL/2}\|
       \le h_M,\\
 h_M:=C_2(M+1)^{1-\theta_{\mathrm{int}}}(2e_M)^{\theta_{\mathrm{int}}}.
 \end{gathered}
\end{equation}
The fixed constant $c_e$ in Eq.~\eqref{cr:eq:theta-p} can be chosen so that
\begin{equation}\label{cr:eq:b-h}
 b_{\mathrm{iso}}:=\frac1{256m},\qquad
 h_M\le b_{\mathrm{iso}}/128\qquad(1\le m\le M).
\end{equation}
The exact normalized marginals then satisfy
\begin{equation}\label{cr:eq:state-order}
 e^{-2h_M}\sigma_j\preceq\rho_j\preceq e^{2h_M}\sigma_j.
\end{equation}
\end{lemma}

\begin{proof}
The analytic decomposition of $K_j$ and the supplied list of $L_j$
give a decomposition of their difference with weighted total sum $O(M)$,
including the centering scalar. This uses only the existence of the
analytic decomposition, not algorithmic access to its coefficients.
Also $\|\mathcal D_j^{\mathrm{err}}\|\le2e_M$ by
Eq.~\eqref{cr:eq:L-properties}. The same bounds hold after embedding in
the common space.

Apply the operator-norm three-lines estimate
\cite[Theorem~3.1, $p_0=p_1=\infty$]{SutterBertaTomamichel2017} to
$F(z)=e^{-zL/2}\mathcal D_j^{\mathrm{err}}e^{zL/2}$ on the strips with
real endpoints $0,\pm t_*$. The boundary norms are at most $2e_M$ and
$C_1(M+1)$ by Eq.~\eqref{cr:eq:complex}. For $0\le t\le1$,
\[
 \|F(\pm t)\|
 \le(2e_M)^{1-t/t_*}[C_1(M+1)]^{t/t_*}
 \le h_M,
\]
after increasing $C_2$ and choosing $c_e$ small enough that
$2e_M\le C_1(M+1)$. Here
$\theta_{\mathrm{int}}=1-1/t_*>0$. Since
$p_{\mathrm{ref}}\theta_{\mathrm{int}}>6$, the power of $M+1$ in $h_M$
is less than $-5$. A sufficiently small fixed $c_e$ therefore gives
Eq.~\eqref{cr:eq:b-h} uniformly for $m\le M$, independently of $\eta$.

For the state comparison, define
$F_j(t)=e^{-tK_j/2}e^{tL_j/2}$. Differentiation gives
\[
 F_j'(t)=-\tfrac12F_j(t)e^{-tL_j/2}(K_j-L_j)e^{tL_j/2}.
\]
Equation~\eqref{cr:eq:h} bounds the generator by $h_M/2$. Its integral
equation and the inverse equation give
$\|F_j(1)\|,\|F_j(1)^{-1}\|\le e^{h_M/2}$. Hence
\[
 e^{-h_M}\I\preceq
 e^{L_j/2}e^{-K_j}e^{L_j/2}
 \preceq e^{h_M}\I.
\]
Congruence by $e^{-L_j/2}$ and normalization of the traces prove
Eq.~\eqref{cr:eq:state-order}.
\end{proof}

\subsection{LCU and time-ordered evolution}
\label{app:relative-circuits}

The first constructive ingredient converts the imaginary-time
expansion into an explicit encoding from the interaction lists.
The analytic majorant is the overlapping-support bound of
Ref.~\cite{BluhmCapelPerezHernandez2025}; the construction below
prepares the required support sequences and phases while preserving
its normalization. Standard LCU~\cite{GilyenEtAl2019} then implements
the resulting sum.

\begin{lemma}[LCU for imaginary-time evolution]\label{lem:imaginary-time-lcu}
Let $L=L^\dagger$ and $U$ have explicit interaction lists with efficiently
implementable controlled term encodings and PREPARE operations over terms
containing a specified site, of the form in Eq.~\eqref{eq:incidence}.
Suppose the non-scalar list of $L$ has weighted coefficient bound at most
$0<J<\lambda$, and the weighted sum of the supplied norm bounds for $U$,
including its scalar term, is at most $B_U>0$. Then
\begin{equation}\label{cr:eq:tauU}
 \tau_L^t(U):=e^{-tL/2}Ue^{tL/2},\qquad 0\le t\le1,
\end{equation}
has a controlled encoding with normalization at most $B_U/(1-J/\lambda)$.
At fixed $J/\lambda<1$, operator error $\zeta\in(0,1/2)$ requires degree
$O(\log((B_U+1)/\zeta))$. The number of list and support operations is
polynomial in this degree and in the maximum support size. Their
bit-precision costs and the input-list access costs are included.
\end{lemma}

\begin{proof}
Write the non-scalar terms as $L=\sum_X\ell_X+c_L\I$ and
$U=\sum_Xu_X+c_U\I$, with supplied normalization factors
$a_X^{(L)}\ge\norm{\ell_X}$ and $a_X^{(U)}\ge\norm{u_X}$.
Here an interaction label includes its support and, if necessary, a tag
for distinct terms with the same support. The scalar $c_L$ has no
commutator, and $c_U$ contributes only at degree zero. Their direct
encodings handle the purely scalar cases; a zero operator is encoded
by a zero block with a positive normalization.

For the non-scalar part, choose $X_0$ with squared amplitude
$a_{X_0}^{(U)}e^{\lambda|X_0|}/B_U$, assigning any unused probability
to an invalid label. Set $\mathcal S_0=X_0$.
At extension $r$, choose a uniform site in $\mathcal S_{r-1}$ and use
the PREPARE operation with weights $a_X^{(L)}e^{\lambda|X|}$ and
denominator $J$ to propose a term $X_r$ of $L$. Retain the proposal only if the chosen site is the first
site of $X_r\cap\mathcal S_{r-1}$. Choose the left or right term of
the commutator with probability $1/2$, and compute
$\mathcal S_r=\mathcal S_{r-1}\cup X_r$ reversibly while retaining the
previous support. Put the minus sign for a right multiplication in
SELECT. For each sequence $\boldsymbol X=(X_0,\ldots,X_n)$, abbreviate
\[
 w(\boldsymbol X):=a_{X_0}^{(U)}\prod_{r=1}^n a_{X_r}^{(L)}.
\]
Before a final rejection step its squared amplitude is
\begin{equation}
 \frac{w(\boldsymbol X)e^{\lambda\sum_{r=0}^n|X_r|}}
 {B_U(2J)^n\prod_{r=1}^n|\mathcal S_{r-1}|}.
\end{equation}
The computable acceptance probability is
\begin{equation}\label{cr:eq:accept}
 p_{\mathrm{acc}}(\boldsymbol X)=
 \frac{\lambda^n\prod_{r=1}^n|\mathcal S_{r-1}|}
 {n!e^{\lambda\sum_{r=0}^n|X_r|}}\le1.
\end{equation}
Indeed, $|\mathcal S_{r-1}|\le\sum_{r=0}^n|X_r|$, and
$e^x\ge x^n/n!$ for $x\ge0$. The accepted squared amplitude is therefore
\[
 \frac{w(\boldsymbol X)}{B_U n!(2J/\lambda)^n}.
\]
SELECT multiplies the term encodings in the exact order of the selected
nested commutator, using separate projection ancillas for its factors
as in the product construction of Ref.~\cite[Lemma~53]{GilyenEtAl2019}. Invalid labels select a zero block.
Undoing the support computations and label preparation produces an
encoding of $\ad_L^n(U)$ with normalization
\begin{equation}\label{eq:commutator-normalization}
 \alpha_n=B_U n!(2J/\lambda)^n.
\end{equation}
The coefficient $(-t/2)^n/n!$ in the expansion of $\tau_L^t(U)$ cancels
the factorial. Its sign is implemented on the degree label in SELECT.
The total LCU weight is bounded by
$B_U\sum_{n\ge0}(tJ/\lambda)^n$, and the geometric tail gives the stated
cutoff. The degree labels, support sequences, and acceptance probabilities
have polynomial-size descriptions. Finite-precision preparation and
arithmetic use the assumed reversible routines. Diameter weights can be
dropped from the preparation probabilities, since doing so preserves the
upper bounds $J$ and $B_U$.
\end{proof}

We next use the truncated-Dyson block encoding
\cite{LowWiebe2018,KieferovaSchererBerry2019}, including its extension
to general matrix generators by Berry and Costa
\cite[Sec.~3.1]{BerryCosta2024}. The form needed here uses the unit
interval, a uniform binary clock, and earlier times on the left.

\begin{lemma}[Time-ordered encoding on the unit interval]
\label{lem:time-ordered-lcu}
Let $G:[0,1]\to\mathcal B(\mathcal H)$ be continuously differentiable,
with coherently controlled encodings of uniform normalization
$0<\alpha_G\le C_0$ for a fixed $C_0$, available at binary-specified
times to every desired accuracy. Let
$D_G\ge\sup_t\norm{G'(t)}$. For $0<\zeta<1/2$, the solution
$F_G'=-F_GG$, $F_G(0)=\I$, admits an error-$\zeta$ encoding of
$F_G(1)$ with normalization at most $e^{C_0}$. It uses
$O(\log(1/\zeta))$ generator calls and additional gates polynomial in
$\log(D_G+1)$ and $\log(1/\zeta)$, including binary clock arithmetic
and finite-precision synthesis. The generator-call costs are counted
separately.
\end{lemma}

\begin{proof}
Choose the least power-of-two number $n_t\ge\max\{1,4e^{2C_0}D_G/\zeta\}$
of uniform cells, with midpoints $\xi_\ell=(\ell+1/2)/n_t$.
Hadamards prepare the uniform clock on $\log n_t/\log2$ qubits;
reversible arithmetic computes each $\xi_\ell$. For the step generator
$G_{n_t}(t)=G(\xi_\ell)$, the order-$k_{\mathrm D}$ polynomial is
\begin{equation}\label{eq:step-simplex}
 \sum_{r=0}^{k_{\mathrm D}}\frac{(-1)^r}{r!n_t^r}
 \sum_{\ell_1,\ldots,\ell_r=0}^{n_t-1}
 \mathcal T_{\boldsymbol\ell}
 [G(\xi_{\ell_1})\cdots G(\xi_{\ell_r})],
\end{equation}
where $\mathcal T_{\boldsymbol\ell}$ orders the labels increasingly,
with earlier times on the left. Repeated labels have the same
multiplicity weights as the ordered integrals. Preparing independent
clocks and reversibly sorting them implements this polynomial by the
cited Dyson construction, with $O(k_{\mathrm D})$ generator calls and
normalization $\sum_{r=0}^{k_{\mathrm D}}\alpha_G^r/r!\le e^{C_0}$.
The factorial tail is at most $\zeta/4$ for
$k_{\mathrm D}=O_{C_0}(\log(1/\zeta))$.

Variation of constants for the right-ordered, possibly nonunitary
solutions gives
\begin{equation}\label{eq:duhamel-general}
\begin{aligned}
 \norm{F_G(1)-F_{G_{n_t}}(1)}
 &\le e^{2C_0}\int_0^1\norm{G(t)-G_{n_t}(t)}\,dt\\
 &\le\frac{e^{2C_0}D_G}{n_t}.
\end{aligned}
\end{equation}
This is at most $\zeta/4$. The remaining error is allocated to the
controlled encodings and elementary gates using the product bound.
Only the clock bit length enters the additional cost.
\end{proof}

For the reference exponential products, the important step is to
reduce the generator normalization before using the Dyson circuit.
The following estimate performs this reduction for the actual
one-site reference changes.

\begin{proposition}[Products of imaginary-time evolutions]\label{cr:prop:E}
Under the successful-event assumptions of Lemma~\ref{cr:lem:sparsify}, define
\begin{equation}\label{cr:eq:E}
 U_j=L_{j-1}-L_j\otimes\I_{B_j},\qquad
 E_j(t)=e^{-tL_{j-1}/2}e^{t(L_j\otimes\I_{B_j})/2}.
\end{equation}
Each $E_j(1)$ has a constant-normalization encoding with error $\zeta$ and quantum gate cost polynomial in $M,m,\log N,\log(1/\zeta),\log(1/\fail)$.  Its true norm and inverse norm are bounded by a fixed constant.  This construction does not call a high-precision circuit for any $K_i$, marginal state, or correction.
\end{proposition}

\begin{proof}
Differentiate $E_j(t)$ to obtain $E_j'(t)=-E_j(t)G_j(t)$, where
\begin{equation}\label{cr:eq:generator}
 G_j(t)=\tfrac12e^{-tL_j/2}U_je^{tL_j/2}.
\end{equation}
The two supplied lists have weighted total sum $O(M)$, so
Lemma~\ref{lem:imaginary-time-lcu} initially encodes $G_j(t)$ with
normalization $O(M)$. For its true norm, use
\[
 U_j=U_j^{\mathrm{ex}}+\mathcal D_{j-1}^{\mathrm{err}}
                    -\mathcal D_j^{\mathrm{err}}\otimes\I_{B_j}.
\]
Proposition~\ref{cr:ass:analytic} and
Eqs.~\eqref{cr:eq:complex}, \eqref{cr:eq:h} give
\begin{equation}\label{cr:eq:g-bound}
 \sup_{0\le t\le1}\|G_j(t)\|
 \le\frac{\lambda u_0b_*}{\lambda-J_*}+h_M=:g_*.
\end{equation}
Fix a constant upper bound for $g_*$ and apply uniform singular-value
amplification~\cite[Theorem~30]{GilyenEtAl2019} to normalization
$2(g_*+1)$, with fixed output slack and $\poly(M,\log(1/\zeta))$ calls.
This amplification precedes the time-ordered construction.

Since $\|G'_j(t)\|\le\|L_j\|\|G_j(t)\|=O(M)$,
Lemma~\ref{lem:time-ordered-lcu} applies with a power-of-two clock of
$n_{\mathrm{clk}}=\poly(M,1/\zeta)$ midpoint cells, requiring only
$O(\log M+\log(1/\zeta))$ clock qubits. The commutator circuit accepts
this binary time coherently. Its order-$r$ Dyson weight is at most
$[2(g_*+1)]^r/r!$, giving total normalization at most $e^{2(g_*+1)}$
and logarithmic truncation order. The explicit clock preparation,
reversible sorting, and Duhamel bound in that lemma include all time
arithmetic and discretization errors without a table of time cells.
Finally, the integral equations give
$\|E_j(1)\|,\|E_j(1)^{-1}\|\le e^{g_*}$.
\end{proof}

\subsection{Normalization estimates}

The rectangular one-step map appends a maximally entangled pair on the
deleted site and its reference. Contraction with that pair realizes the
normalized partial trace, allowing the cumulative Gram operator to retain
the exact marginal information without preparing or tracing a dense
Gibbs state. The scalar measurements below use independent inputs.

\label{app:normalization-estimation}

For each $j=1,\ldots,m$, append a normalized maximally entangled pair
on $B_j\mathcal E_j$ and apply $E_j(1)$ to the physical registers. This defines
\begin{equation}\label{cr:eq:W0}
 W_j^0|\psi\rangle=\frac1{\sqrt{d_{B_j}}}
 \sum_{b=1}^{d_{B_j}} E_j(1)(|\psi\rangle\otimes|b\rangle_{B_j})
                       \otimes|b\rangle_{\mathcal E_j}.
\end{equation}
For the positive estimate $a_j$ of
$\trn{A_j}((W_j^0)^\dagger W_j^0)$ obtained by the independent
measurements in Lemma~\ref{lem:normalization-estimation}, put
\begin{equation}\label{cr:eq:W-S}
 W_j=a_j^{-1/2}W_j^0,\qquad S_j=W_j^\dagger W_j-\I.
\end{equation}
Here $W_j:\mathcal H_{A_j}\to
\mathcal H_{A_{j-1}}\otimes\mathcal E_j$, and $S_j$ is an analytical
measure of its deviation from the isometry condition.
The following independent experiments determine $a_j$; their accuracy
controls $S_j$ without measuring the eventual recovery input.
\begin{lemma}[Normalization estimates and isometry bounds]
\label{lem:normalization-estimation}
Under the successful-event assumptions of Lemma~\ref{cr:lem:sparsify},
use the maps in \eqref{cr:eq:W0}--\eqref{cr:eq:W-S} and
$b_{\mathrm{iso}}=1/(256m)$. With total probability at least $1-\fail/4$,
independent experiments return positive finite numbers $a_j$ such that
\begin{equation}\label{eq:isometry-deviation}
 \norm{S_j}\le b_{\mathrm{iso}}.
\end{equation}
They require $\poly(m,\log(m/\fail))$ trials per step. An encoding of each
$W_j$ with normalization
\begin{equation}\label{eq:map-normalizations}
 \alpha_W=1+2b_{\mathrm{iso}}
\end{equation}
can then be constructed to operator error $\zeta$ at quantum and classical
compilation cost polynomial in
$M,m,\log N,\log(1/\zeta),\log(1/\fail)$.
\end{lemma}

\begin{proof}
Contracting the maximally entangled pair in $W_j^0$ gives, for every
operator $O$ on $A_{j-1}$,
\begin{equation}\label{cr:eq:compression-W}
 (W_j^0)^\dagger(O\otimes\I)W_j^0
 =e^{L_j/2}\E_{B_j}
   \bigl(e^{-L_{j-1}/2}Oe^{-L_{j-1}/2}\bigr)e^{L_j/2}.
\end{equation}
Let $Z_j=\Tr(e^{-L_j})$ and $\sigma_j=e^{-L_j}/Z_j$.
The exact relation $\rho_j=\Tr_{B_j}\rho_{j-1}$ and
Lemma~\ref{lem:reference-stability} imply
\begin{equation}
 e^{-4h_M}\sigma_j\preceq\Tr_{B_j}\sigma_{j-1}
 \preceq e^{4h_M}\sigma_j.
\end{equation}
Consequently,
\begin{equation}
 e^{-4h_M}\I\preceq
 \frac{d_{B_j}Z_j}{Z_{j-1}}(W_j^0)^\dagger W_j^0
 \preceq e^{4h_M}\I.
\end{equation}
Define the mean
$a_j^*:=\trn{A_j}((W_j^0)^\dagger W_j^0)$. Division by this mean yields
\begin{equation}\label{cr:eq:mean-defect}
 \left\|\frac{(W_j^0)^\dagger W_j^0}{a_j^*}-\I\right\|
 \le e^{8h_M}-1\le16h_M.
\end{equation}
The bounds on $E_j(1)$ and $E_j(1)^{-1}$ from Proposition~\ref{cr:prop:E}
give $e^{-2g_*}\le a_j^*\le e^{2g_*}$. These bounds do not require the
partition functions $Z_j$.

Fix a constant-normalization encoding of $W_j^0$ with known factor
$\alpha_{0,j}$. Prepare a normalized maximally entangled vector on
$A_j$ and an independent copy of $A_j$. Applying the encoding to the
first register gives ideal success probability
\begin{equation}\label{eq:normalization-probability}
 p_j^{(0)}=\frac{a_j^*}{\alpha_{0,j}^2}.
\end{equation}
A trial measures whether all output projection ancillas are in their
specified zero states. The output of $W_j^0$ and its environment are not measured for this probability. Each trial uses newly prepared inputs,
independent of the later recovery input. When $A_j=\varnothing$, the
entangled input is simply the scalar state.

Estimate \eqref{eq:normalization-probability} by the sample mean
$\widehat p_j$, and store a rounded value of
$a_j=\alpha_{0,j}^2\widehat p_j$. The known constant lower bound on
$a_j^*$ allows an additive probability tolerance proportional to
$b_{\mathrm{iso}}$ to give relative error at most $b_{\mathrm{iso}}/16$.
Hoeffding's inequality and a union bound over $j$ require
$O(b_{\mathrm{iso}}^{-2}\log(m/\fail))$ trials per step. Allocate
fixed fractions of that tolerance to sampling, circuit-implementation
bias, and rounding. On the resulting joint successful event,
\begin{equation}
 \norm{S_j}\le
 \frac{16h_M+b_{\mathrm{iso}}/16}{1-b_{\mathrm{iso}}/16}
 <b_{\mathrm{iso}}.
\end{equation}
The stored $a_j$ is used in the definition of $\widehat G_j$ in
\eqref{cr:eq:Ghat}. Therefore
$\Tr_{B_j}\widehat G_{j-1}=a_jd_{B_j}\widehat G_j$ is exact for that
stored value; estimation affects the bound on $S_j$, not the target
normalized marginal. If a stored value is nonpositive or outside
$[e^{-2g_*}/2,2e^{2g_*}]$, the construction returns the initialized
identity unitary specified in Sec.~\ref{app:completion}.

The initial encoding of $W_j$ has constant normalization. Since
$\norm{W_j}\le\sqrt{1+b_{\mathrm{iso}}}$, amplification to
$\alpha_W=1+2b_{\mathrm{iso}}$ leaves an output singular-value margin
proportional to $b_{\mathrm{iso}}$. It requires
$O(b_{\mathrm{iso}}^{-1}\log(1/\zeta))$ calls, with constants depending
only on the fixed input bounds~\cite[Theorem~30]{GilyenEtAl2019}.
The input encoding error is chosen smaller than a fixed fraction of
$b_{\mathrm{iso}}$ to preserve this margin, and smaller as required by $\zeta$ and final error budget. Phase synthesis and finite gate
errors are included as in Sec.~\ref{sec:microscopic-model}. The bound on
$S_j$ is used analytically in Proposition~\ref{cr:prop:residual}; it does
not require compiling a separate encoding of $S_j$.
\end{proof}

\subsection{Cumulative products and endpoint corrections}

The first constructive step now combines the exponential-product circuits
with one-site normalized partial traces. The target correction is a square
operator determined by the microscopic $K_0$ and the chosen reference.
The $W_j$ realize its contractions coherently with controlled normalization;
their flat product never calls a preceding correction circuit. Only the
two endpoints are compiled as corrections. The identities below connect
this implementation to Eqs.~\eqref{m:eq:square-correction} and
\eqref{m:eq:contraction-step}.

\begin{proposition}[Correction encoding]\label{cr:prop:residual}
Under the successful-event guarantees of
Lemma~\ref{cr:lem:sparsify} and
Lemma~\ref{lem:normalization-estimation}, define
\begin{equation}\label{cr:eq:Ghat}
\begin{gathered}
 \widehat G_0=e^{-K_0},\qquad
 \widehat G_j=a_j^{-1}\E_{B_j}(\widehat G_{j-1}),\\
 X_j=e^{L_j/2}\widehat G_je^{L_j/2}-\I_{A_j},\qquad X_0=0.
\end{gathered}
\end{equation}
With $E_i=E_i(1)$ embedded in $\Omega$, these corrections have the
square-operator representation
\begin{equation}\label{cr:eq:square-correction}
 \I_{A_j}+X_j=
 \frac{\E_{\Omega\setminus A_j}
    \!\left[(E_1\cdots E_j)^\dagger(E_1\cdots E_j)\right]}
      {\prod_{i=1}^{j}a_i}.
\end{equation}
The empty product gives the identity for $j=0$.
For their circuit implementation, form
\begin{equation}\label{cr:eq:correction-product}
 \mathbb W_0=\I_{A_0},\qquad
 \mathbb W_j=(\mathbb W_{j-1}\otimes\I_{\mathcal E_j})W_j.
\end{equation}
Then
\begin{equation}\label{cr:eq:correction-gram}
 X_j=\mathbb W_j^\dagger\mathbb W_j-\I_{A_j}.
\end{equation}
Here $\mathbb W_j:\mathcal H_{A_j}\to
\mathcal H_{A_0}\otimes\bigotimes_{i=1}^j\mathcal E_i$ is a flat
ordered product of independently compiled $W_1,\ldots,W_j$, with
identities on unused references.
Each $\widehat G_j$ is a positive scalar multiple of $\rho_j$.
For $0\le j\le m$, putting $b=b_{\mathrm{iso}}$ gives
\begin{equation}\label{cr:eq:correction-small}
 \|X_j\|\le(1+b)^j-1\le e^{1/256}-1<1/128.
\end{equation}
For any selected $j$, $X_j$ has an encoding with a Hermitian projected
block, arbitrary operator error $\zeta\in(0,1/2)$, and normalization
$\alpha_X=1/32$. For $j\ge1$, it uses $O(j\log(2/\zeta))$ calls to
independently compiled $W_i$ encodings and their adjoints. Their required
accuracies are inverse-polynomial in $m$ and $1/\zeta$, and the
compilation costs have the polynomial form of
Theorem~\ref{cr:thm:local-main}. The one-patch compiler selects only
$j=m-1,m$, giving $X_{Yz}$ and $X_Y$.
\end{proposition}

\begin{proof}
Adjacent exponentials cancel, so
\[
 E_1\cdots E_j=e^{-K_0/2}
       (e^{L_j/2}\otimes\I_{\Omega\setminus A_j}).
\]
Taking the normalized partial trace proves
Eq.~\eqref{cr:eq:square-correction} from Eq.~\eqref{cr:eq:Ghat}.
Because $E_j$ is supported on $A_{j-1}$, splitting that trace into the
successive $\E_{B_i}$ gives
\begin{equation}\label{cr:eq:square-contraction-step}
 \I_{A_j}+X_j=a_j^{-1}\E_{B_j}
  \!\left[E_j^\dagger(\I_{A_{j-1}}+X_{j-1})E_j\right].
\end{equation}
The $\widehat G_j$ are positive scalar multiples of the exact marginals.
Equation~\eqref{cr:eq:compression-W} identifies the right-hand side of
Eq.~\eqref{cr:eq:square-contraction-step} with the coherent contraction

\begin{equation}\label{cr:eq:correction-step}
 \I_{A_j}+X_j
 =W_j^\dagger\bigl((\I_{A_{j-1}}+X_{j-1})\otimes\I_{\mathcal E_j}\bigr)W_j.
\end{equation}
Induction from $X_0=0$ proves Eq.~\eqref{cr:eq:correction-gram}.
Using $(1-b)\|v\|^2\le\|W_iv\|^2\le(1+b)\|v\|^2$ at each factor,
\[
 (1-b)^j\I\preceq\mathbb W_j^\dagger\mathbb W_j\preceq(1+b)^j\I.
\]
Since $1-(1-b)^j\le(1+b)^j-1$ and $jb\le1/256$, this gives
Eq.~\eqref{cr:eq:correction-small}. The standard product rule~\cite[Lemma~53]{GilyenEtAl2019} gives normalization at most $(1+2b)^j$
for $\mathbb W_j$. Its Gram product minus the identity therefore has
normalization at most
\begin{equation}\label{cr:eq:correction-normalization}
 (1+2b)^{2j}+1\le e^{1/64}+1<3.
\end{equation}
For an exact projected encoding
$\mathbb W_j/\alpha=J_{\mathrm{out}}^\dagger UJ_{\mathrm{in}}$, the Gram
operator is the projected product
\begin{equation}\label{eq:gram-projection}
 \frac{\mathbb W_j^\dagger\mathbb W_j}{\alpha^2}
 =J_{\mathrm{in}}^\dagger U^\dagger
   \Pi_{\mathrm{out}}U J_{\mathrm{in}},
 \qquad \Pi_{\mathrm{out}}=J_{\mathrm{out}}J_{\mathrm{out}}^\dagger.
\end{equation}
The intermediate projector selects the encoded output before applying
its adjoint. Its known reflection supplies a normalization-one encoding;
standard block multiplication uses fresh projection ancillas and retains
all intermediate physical and reference registers. The same identity
holds for the encoded approximants used below.
Pad the subtracted block to normalization $3$. A real odd amplification polynomial of gain $96$
gives $\alpha_X=1/32$ in $O(\log(2/\zeta))$ calls, with fixed output
slack from Eq.~\eqref{cr:eq:correction-small}
\cite[Theorem~30]{GilyenEtAl2019}.

For encoded approximants of error $\delta_W\le b/8$,
Eq.~\eqref{eq:general-product-error} yields
\[
 \|\widetilde{\mathbb W}_j-\mathbb W_j\|
 \le j\delta_W(\sqrt{1+b}+\delta_W)^{j-1}\le Cj\delta_W.
\]
The Gram error has the same order because both product norms are bounded.
Take $\delta_W$ sufficiently smaller than $\zeta/j$ and allocate comparable
error to subtraction, scalar rounding, and Hermitian symmetrization.
The raw approximation then has norm at most $1/64$, retaining fixed
amplification slack. Its linear target transfers operator error without a
condition-number factor; polynomial approximation and phase synthesis
use the remaining budget. Precision costs are polynomial in $\log m+\log(1/\zeta)$. Each raw call uses $O(j)$ independent $W_i$
circuits, never an earlier correction circuit, proving the call bound.
For $j=0$ a zero block of normalization $1/32$ suffices.
\end{proof}

\subsection{Integral representation and bounded functions}
\label{app:inverse}
\label{sec:correction-factors}

The corrections are now encoded, but the Petz factor still requires
$T(L,X)$, which is not generally $(\I+X)^{-1/2}$. The second constructive
step converts this noncommutative square-root factor into bounded
functions of the accessible $L$ and the encoded $X$. The standard fractional-power resolvent representation
\cite{Balakrishnan1960} yields the following ordered identity. Its
implementation, rather than the scalar integral formula, is the
second constructive ingredient quantified in
Theorem~\ref{lem:inverse-circuit}.
\begin{lemma}[Integral representation of the correction factor]
\label{lem:relative-inverse}
Let $L=L^\dagger$, $X=X^\dagger$, and $\norm X<1$. Put $A=e^{-L/2}$ and
\begin{equation}\label{cr:eq:T-def}
 G=A(\I+X)A,\qquad T(L,X)=AG^{-1/2}.
\end{equation}
Write $T=T(L,X)$ and define
\begin{equation}\label{eq:bounded-functions}
 P_s=(\I+e^{2s+L})^{-1},\qquad f_s(L)=\sech(s+L/2).
\end{equation}
For complex arguments in a strip without poles, $P_z$ and $f_z(L)$
denote the analytic continuations of these same functions. For real $s$,
the norm-convergent identity is
\begin{equation}\label{eq:relative-inverse-integral}
 T=\frac1\pi\int_{\mathbb R}(\I+P_sX)^{-1}f_s(L)\,ds.
\end{equation}
Moreover $TT^\dagger=(\I+X)^{-1}$ and $T^{-1}=T^\dagger(\I+X)$.
\end{lemma}

\begin{proof}
For $G>0$, spectral calculus gives
$G^{-1/2}=\pi^{-1}\int_0^\infty v^{-1/2}(G+v\I)^{-1}dv$.
Write $P_v=A^2(A^2+v\I)^{-1}$. Direct multiplication verifies
\begin{equation}
 A(G+v\I)^{-1}=(\I+P_vX)^{-1}A(A^2+v\I)^{-1}.
\end{equation}
With $v=e^{2s}$,
$2\sqrt v A(A^2+v\I)^{-1}=f_s(L)$ and $P_v=P_s$.
This proves the integral. Inverting $G=A(\I+X)A$ between the two
outer $A$ factors proves the identity for $TT^\dagger$; multiplying it
by $\I+X$ gives the inverse formula. No commutation of $X$ and $L$ is used.
\end{proof}

The spectral functions in Eq.~\eqref{eq:bounded-functions} are
implemented by standard polynomial approximation and QSVT
\cite{Trefethen2019ATAP,GilyenEtAl2019}. The required additional estimate
is uniform in the integration variable $s$. For an encoding of $L$
with normalization $\Lambda_L\ge\norm L$, define the envelope
\begin{equation}\label{eq:envelope}
 \nu_{\Lambda_L}(s)=
 \begin{cases}
 1,&|s|\le\Lambda_L/2,\\
 \sech(|s|-\Lambda_L/2),&|s|>\Lambda_L/2.
 \end{cases}
\end{equation}

\begin{lemma}[Uniform encodings of the bounded functions]
\label{lem:bounded-functions}
Let $L=L^\dagger$ have an encoding with a Hermitian projected block
and normalization $\Lambda_L$. Uniformly for real $s$, the functions
$P_s$ and $f_s(L)/\nu_{\Lambda_L}(s)$ admit Hermitian contraction
approximants with operator error $\epsilon\in(0,1/2)$, encoding
normalizations $\alpha_P=3/2$ and $\alpha_f=2$, respectively, and
polynomial degree
\[
 O\!\left((\Lambda_L+1)\log\frac{\Lambda_L+1}{\epsilon}\right).
\]
Coefficient and phase synthesis have polynomial cost in this degree
and $\log(1/\epsilon)$; finite input-encoding and gate errors are
allocated by the common precision convention in
Sec.~\ref{sec:microscopic-model}.
\end{lemma}

\begin{proof}
The scalar case $\Lambda_L=0$ is immediate. Otherwise set
$g_s(y)=\tanh(s+\Lambda_Ly/2)$ and
$h_s(y)=\sech(s+\Lambda_Ly/2)/\nu_{\Lambda_L}(s)$.
Both are real and bounded by one on $[-1,1]$. On the Bernstein
ellipse with parameter $\varrho=1+[8(\Lambda_L+1)]^{-1}$,
$z=s+\Lambda_Ly/2$ has $|\operatorname{Im}z|\le1/16$ and its real
part is within $1/256$ of a point $s+\Lambda_Ly_0/2$, $y_0\in[-1,1]$.
The inequalities
\[
\begin{aligned}
 |\sech(x+iv)|&\le\frac{\sech x}{\cos v},\qquad
 |(\log\sech x)'|\le1,\\
 |\tanh(x+iv)|&\le1\qquad(|v|\le1/16)
\end{aligned}
\]
therefore give a constant bound for both analytic functions on that
ellipse, independent of $s$ and $\Lambda_L$. The Chebyshev estimate
\cite[Theorems~8.1--8.2]{Trefethen2019ATAP} yields truncation error
$C(\Lambda_L+1)e^{-cn/(\Lambda_L+1)}$ at degree $n$, proving the
stated degree.

Approximating $(1-\epsilon/2)g_s$ and $(1-\epsilon/2)h_s$ to error
$\epsilon/2$ gives real polynomials $p_g,p_h$ bounded by one.
Apply QSVT to their even and odd parts. Then
$\widetilde P_s=(\I-p_g(L/\Lambda_L))/2$ is a positive contraction
with LCU weight $3/2$, and $p_h(L/\Lambda_L)$ has LCU weight $2$.
The same coefficient decay gives polynomial-cost coefficient
computation on a uniform Fourier grid; stable evaluations of
$h_s$ use $\log\sech$. The known scalar envelope restores $f_s(L)$,
so its reciprocal never enters the quantum normalization.
Phase synthesis and coherent selection over the polynomially many
quadrature nodes use Sec.~\ref{sec:microscopic-model}.
\end{proof}

\subsection{Quadrature and circuit construction}

\paragraph{Positive quadrature.}
\label{sec:positive-quadrature}
We use positive Gauss--Legendre quadrature and its analytic convergence
bound~\cite[Theorem~4.5]{Trefethen2008}, keeping the envelope-weighted
sum that determines the LCU normalization. Suppose that $F$ is
operator-norm analytic on $|\operatorname{Im}z|<b_0$ and, for a fixed
$0<b<b_0$, satisfies
$\norm{F(s+iu)}\le M_b\nu_{\Lambda_L}(s)$ for $|u|\le b$, with
$M_b\ge1$. There are real nodes $s_l$ and positive weights $w_l$
whose operator-norm error for $\int_{\mathbb R}F(s)\,ds$ is at most
$\epsilon$, with
\begin{equation}\label{eq:quad-mass}
 \sum_lw_l\nu_{\Lambda_L}(s_l)\le\Lambda_L+\pi+1
\end{equation}
and
\begin{equation}\label{eq:quad-size}
\begin{split}
 K_{\mathrm{quad}}=O_b\biggl[&
 (\Lambda_L+\log(M_b/\epsilon)+1)\\
 &\times\log\frac{M_b(\Lambda_L+2)}{\epsilon}\biggr].
\end{split}
\end{equation}
To obtain this rule, truncate at
$B=\Lambda_L/2+\log(8M_b/\epsilon)$, giving tail error at most
$\epsilon/2$, and partition $[-B,B]$ at $\pm\Lambda_L/2$ and into
intervals shorter than the fixed strip width. The cited exponential
convergence applies to operators by testing against norm-one linear
functionals. A logarithmic number of nodes per interval gives the
remaining error and Eq.~\eqref{eq:quad-size}. Applying the same rule
to each analytic piece of $\nu_{\Lambda_L}$, with error at most one,
gives Eq.~\eqref{eq:quad-mass} because its integral is
$\Lambda_L+\pi$. All constants depend only on the fixed strip width.

The following circuit uses the ordered integral identity of
Lemma~\ref{lem:relative-inverse}, derived from the fractional-power
resolvent representation~\cite{Balakrishnan1960}, together with standard
LCU and QSVT~\cite{GilyenEtAl2019}. The normalization and complexity
bounds below quantify its implementation in the regime $\norm X<1$.

\begin{theorem}[Circuit for the correction factor]
\label{lem:inverse-circuit}
Let $L=L^\dagger$ and $X=X^\dagger$ act on the same space, and
assume $\norm X\le\delta_0$ for a fixed $0<\delta_0<2/3$, together with an
encoding of $X$ of normalization factor at most $\delta_0$ to arbitrary
desired accuracy. The fixed bound $\delta_0<2/3$ leaves a positive
margin in the Neumann coefficient sum.
Under the access of Lemma~\ref{lem:bounded-functions}, the operator $T$
in Lemma~\ref{lem:relative-inverse} has an error-$\zeta$ encoding with
\begin{equation}\label{eq:relative-inverse-alpha}
 \alpha_T\le\frac{2(\Lambda_L+\pi+1)}{\pi(1-3\delta_0/2)}
\end{equation}
and $\poly(\Lambda_L,\log(1/\zeta))$ calls and arithmetic operations.
The hidden constants depend on the fixed positive gaps in the convergence inequalities, not on
the smallest eigenvalue of $G$.
\end{theorem}

\begin{proof}
Write $A_L=\Lambda_L+\pi+1$ and
$\delta_*=(1+\delta_0)/2<1$. Let $\widetilde X$ be a Hermitian encoded approximation to $X$, and let
$\widetilde P_s$ be the polynomial contraction from
Lemma~\ref{lem:bounded-functions}. First bound the error for these
operator approximants; errors in the finite unitaries that implement
their encodings are allocated after fixing the degrees and quadrature.
Choose operator errors
$\epsilon_X,\epsilon_P$ such that
\begin{equation}\label{eq:relative-budget}
\begin{split}
 \epsilon_X+\delta_0\epsilon_P
 \le\min\biggl\{&\frac{1-\delta_0}{2},\\
 &\frac{\zeta(1-\delta_0)(1-\delta_*)}{64A_L}\biggr\}.
\end{split}
\end{equation}
Then $\norm{\widetilde P_s\widetilde X}\le\delta_*$, and the
resolvent identity gives
\begin{equation}\label{eq:resolvent-error}
 \norm{(\I+P_sX)^{-1}-(\I+\widetilde P_s\widetilde X)^{-1}}
 \le\frac{\epsilon_X+\delta_0\epsilon_P}{(1-\delta_0)(1-\delta_*)}.
\end{equation}
Take a Neumann cutoff $r_N$ with
$\delta_*^{r_N+1}/(1-\delta_*)\le\zeta/(64A_L)$.
The sum of LCU coefficient weights is controlled by
\begin{equation}
 \sum_{r=0}^{r_N}(\alpha_P\delta_0)^r\le(1-3\delta_0/2)^{-1}.
\end{equation}
Choose the normalized $f_s(L)$ error at most
$\zeta(1-\delta_*)/(64A_L)$.
The approximate integrand then differs from the exact integrand by
at most $3\zeta\nu_{\Lambda_L}(s)/(64A_L)$ and has normalization
$2\nu_{\Lambda_L}(s)/(1-3\delta_0/2)$.

For complex $z=s+iu$, the scalar spectral functions obey
$\norm{P_z}\le1/|\cos u|$ and
$\norm{f_z(L)}\le\nu_{\Lambda_L}(s)/|\cos u|$.
Choose fixed $b<\pi/2$ with $\cos b>\delta_0$.
The exact integrand is analytic on that strip and bounded by
$\nu_{\Lambda_L}(s)/(\cos b-\delta_0)$.
The positive quadrature above gives error $\zeta/4$
with the envelope-weight bound \eqref{eq:quad-mass}.
Positive weights add the primitive errors with that same envelope.
Allocate the remaining error to node/weight preparation, finite scalar
arithmetic, and QSVT synthesis; their call count is polynomial, so the
unitary hybrid bound applies. The sum of node normalization weights gives
\eqref{eq:relative-inverse-alpha}.
The degree in Lemma~\ref{lem:bounded-functions} is
polynomial in $\Lambda_L$ and logarithmic in the reciprocal primitive
error, including an additional $\log(\Lambda_L+1)$ factor.
All truncation orders and polynomial degrees therefore retain the stated
$\poly(\Lambda_L,\log(1/\zeta))$ bound at fixed margins.
No isolated encoding of $G^{-1/2}$ or $e^{L/2}$ occurs.
\end{proof}

\begin{lemma}[Correction factors for the Petz circuit]\label{cr:lem:T}
Let $L=L^\dagger$ and $X=X^\dagger$ have encodings of normalizations
$\Lambda_L=O(M)$ and $\alpha_X=1/32$, available to every desired
accuracy, and assume $\|X\|\le1/32$.
The operator $T$ of Eq.~\eqref{cr:eq:T-def} and its inverse have
error-$\zeta$ encodings using $\poly(M,\log(1/\zeta))$ calls.
Their normalization factors can be $O(M+1)$ and $2$, respectively, and
\begin{equation}\label{cr:eq:T-condition}
\begin{gathered}
 TT^\dagger=(\I+X)^{-1},\qquad T^{-1}=T^\dagger(\I+X),\\
 (1+1/32)^{-1/2}\le\sigma_k(T)\le(1-1/32)^{-1/2}.
\end{gathered}
\end{equation}
Here $\sigma_k(T)$ denotes a singular value.
\end{lemma}
\begin{proof}
Apply Theorem~\ref{lem:inverse-circuit} with $\delta_0=1/32$.
Lemma~\ref{lem:relative-inverse} gives the identities and spectral bounds.
The product of the adjoint encoding with an LCU encoding of $\I+X$
has normalization $\alpha_T(1+1/32)=O(M+1)$ and error bounded by a
fixed multiple of the errors in $T,X$. Since $\|T^{-1}\|<2$, uniform
singular-value amplification~\cite[Theorem~30]{GilyenEtAl2019} reduces
this normalization to $2$ with fixed output slack and
$O((M+1)\log((M+1)/\zeta))$ calls. Polynomially smaller input errors
and the phase-synthesis costs preserve the stated complexity.
\end{proof}
For comparison with Ref.~\cite[Theorems~1--2 and Appendix~C.1]{GilyenChenDoriguelloKastoryano2024}, use $T(L,X)$ from
Eq.~\eqref{cr:eq:T-def}. For
$\sigma_L=e^{-L}/\Tr(e^{-L})$ and $B_{\mathrm{rej}}\succeq0$ with $\|B_{\mathrm{rej}}\|<1$,
the square-root factor used there for a rejection Kraus operator is
\begin{equation}\label{eq:rejection-factor-comparison}
\begin{aligned}
 K_{\mathrm{rej}}
 &:=\bigl[\sigma_L^{1/2}(\I-B_{\mathrm{rej}})
                    \sigma_L^{1/2}\bigr]^{1/2}\sigma_L^{-1/2}\\
 &=T(L,-B_{\mathrm{rej}})^{-1}.
\end{aligned}
\end{equation}
The normalization of $\sigma_L$ cancels in the second line. That work
constructs such factors using operator series and time integrals.
Here the inputs are the compressed $L_j$ and the independently encoded
$X_j$ of Proposition~\ref{cr:prop:residual}; the bounded resolvent integral
in Sec.~\ref{app:inverse} gives the normalization and error bounds
needed for their composition into the specified Petz isometry.

The target is the ordinary Petz isometry of
Refs.~\cite{Petz1986Sufficiency,GilyenEtAl2022}. The identity below
verifies that the factorized circuit realizes that same map, including
cancellation of the normalization constants.

\begin{proposition}[Petz-factor and isometry identities]\label{cr:prop:V}
For \eqref{cr:eq:Ghat}, define
\begin{equation}
 T_j=e^{-L_j/2}\widehat G_j^{-1/2},\qquad T_0=\I.
\end{equation}
Then the exact Petz factor and isometry are
\begin{equation}\label{eq:petz-factor-corrections}
 Q_j=\frac{T_{j-1}^{-1}E_j(1)(T_j\otimes\I_{B_j})}
                  {\sqrt{a_jd_{B_j}}},
\end{equation}
\begin{equation}\label{cr:eq:V-factor}
 V_j=(T_{j-1}^{-1}\otimes\I_{\mathcal E_j})W_jT_j.
\end{equation}
\end{proposition}

\begin{proof}
Since $T_{j-1}^{-1}=\widehat G_{j-1}^{1/2}e^{L_{j-1}/2}$, adjacent exponentials cancel without commuting any two different operators.  The right-hand side of \eqref{cr:eq:V-factor} applied to $|\psi\rangle$ equals
\begin{equation}\label{cr:eq:V-expanded}
 \frac1{\sqrt{a_jd_{B_j}}}\sum_b
 \widehat G_{j-1}^{1/2}
 (\widehat G_j^{-1/2}|\psi\rangle\otimes|b\rangle)
 \otimes|b\rangle_{\mathcal E_j}.
\end{equation}
Equation~\eqref{cr:eq:Ghat} gives
\begin{equation}
 \Tr_{B_j}\widehat G_{j-1}=a_jd_{B_j}\widehat G_j.
\end{equation}
After normalizing the two positive operators to states, their square-root ratio consequently has exactly the coefficient in \eqref{cr:eq:V-expanded}.  This is Eq.~\eqref{cr:eq:petz}.  In particular the result is independent of the chosen scalar normalization convention.
\end{proof}

\subsection{Completion and finite precision}
\label{app:completion}

All factors of the specified Petz isometry are now encoded. We finish
the local theorem with standard singular-vector transformation
\cite{GilyenEtAl2019,QuekRebentrost2021}, keeping the error outside
the selected output space in the final guarantee.

\begin{proof}[Proof of Theorem~\ref{cr:thm:local-main}]
Apply Lemma~\ref{cr:lem:sparsify} and condition on its successful event.
The compressed lists satisfy Eq.~\eqref{cr:eq:L-properties}.
Lemma~\ref{lem:reference-stability} justifies the normalization experiments,
and Proposition~\ref{cr:prop:E} constructs each $W_j^0$ independently.
Lemma~\ref{lem:normalization-estimation} gives encodings of $W_j$ with the
normalization in Eq.~\eqref{eq:map-normalizations}.
Proposition~\ref{cr:prop:residual} constructs $X_{m-1}$ and $X_m$ from
ordered products and one fixed-margin amplification, with normalization
$1/32$. Lemma~\ref{cr:lem:T} constructs the integral factor and the inverse
$T_{m-1}^\dagger(\I+X_{m-1})$. The products in
Proposition~\ref{cr:prop:V} then encode $V_m$ with normalization $\poly(M)$.

For completeness, consider the whole initialized action rather than
only its selected block. Write the encoding unitary as $U$, with
$A=J_{\mathrm{out}}^\dagger UJ_{\mathrm{in}}$ and
$\norm{A-aV_m}\le\epsilon_{\mathrm{blk}}$, where
$a=\alpha^{-1}$ is inverse-polynomial in $M$.
Choose $\epsilon_{\mathrm{blk}}\le a\eta/48$.
The polar isometry $W_A=A(A^\dagger A)^{-1/2}$ satisfies
$\norm{W_A-V_m}\le3\epsilon_{\mathrm{blk}}/a$.
Singular-vector transformation
\cite[Theorem~26]{GilyenEtAl2019}, as used in quantum polar
decomposition~\cite{QuekRebentrost2021}, uses
$O(a^{-1}\log(1/\eta))$ calls to $U,U^\dagger$ and the subspace
reflections to approximate $W_A$ above singular value $a/2$.
Take its block error to be $\epsilon_{\mathrm{amp}}=\eta^2/512$.
If $B$ is the resulting selected block and $R_\perp$ the remaining
output, unitarity gives
$B^\dagger B+R_\perp^\dagger R_\perp=\I$ and hence
$\norm{R_\perp}\le\sqrt{2\epsilon_{\mathrm{amp}}}$.
The total error before finite-gate synthesis is bounded by
\[
 3\epsilon_{\mathrm{blk}}/a+\epsilon_{\mathrm{amp}}
   +\sqrt{2\epsilon_{\mathrm{amp}}}<\eta/4.
\]
The remaining budget covers elementary-gate and input-call errors.
The output projector selects zero work qubits while retaining the
physical output and its fixed-basis reference, giving
Eq.~\eqref{eq:local-isometry-accuracy}. The induced channel has
error at most $\eta$ after tracing out the reference and work
registers. No measurement is needed during this unitary application.

The truncation orders, retained matrix sizes, amplification degrees,
and clock bit lengths are polynomial in
Eq.~\eqref{cr:eq:main-local-cost}. The nesting depth of the circuit
transformations is fixed: increasing $m$ only lengthens the flat
product of independent $W_j$ circuits, whose normalization obeys
Eq.~\eqref{cr:eq:correction-normalization}.

To make the finite-precision claim explicit, first select all truncation degrees and quadratures for a preliminary tolerance polynomially smaller than $\eta$.  Expand the circuit into its individual controlled term and arithmetic calls; their number is a polynomial $P$ in \eqref{cr:eq:main-local-cost}.  Products and resolvents above have at most polynomial error factors, and singular transformations have polynomial degrees and fixed or inverse-polynomial margins.  Assign each intermediate operator approximation an error at most $\eta/(2P^c)$ for a sufficiently large fixed $c$, and each synthesized unitary call error at most $\eta/(32P)$. If a robust transformation bound uses a square root of an input error, use $(\eta/(2P^c))^2$ at that input.  These changes affect the costs only through $\log(1/\eta)+O(\log P)$.  The product perturbation bounds, resolvent identity, Duhamel estimate, and unitary hybrid bound then give total isometry error at most $\eta/2$. Every invocation has been counted down to the explicitly retained local
matrices. Coefficient evaluation, phase synthesis from
Sec.~\ref{sec:microscopic-model}, and finite gate-list construction have
polynomial classical cost in the same parameters.

Allocate $\fail/4$ each to reference sampling and normalization
estimation, reserving the remainder for other probabilistic subroutines.
Check all coefficient weights with conservative rounding and every
normalization estimate against its range. A failed check returns the
initialized identity $|\psi\rangle\mapsto
|\psi\rangle|0\rangle_{B_m\mathcal E_m\mathsf A_m}$, padded to the same
workspace. Thus every preprocessing outcome defines a unitary and a
CPTP reduced map; all sampled and measured data are independent of the
later recovery input.
\end{proof}
Algorithm~\ref{alg:local-petz-detailed} gives the detailed construction
corresponding to Algorithm~\ref{alg:local-petz}.

\begin{algorithm}
\caption{Detailed compilation of one Petz isometry}
\label{alg:local-petz-detailed}
\begin{algorithmic}[1]
\Require Patch $\Omega$, a chain of Definition~\ref{cr:def:chain},
 microscopic $K_0$, reference lists satisfying
 Assumption~\ref{cr:ass:classical}(i) and
 Eq.~\eqref{cr:eq:accuracy-input};
 $\eta,p_{\mathrm{fail}}\in(0,1/4)$.
\Ensure A unitary $\mathscr U_m$ satisfying
 Eq.~\eqref{eq:local-isometry-accuracy}, except with probability
 $p_{\mathrm{fail}}$.
\State Fix the internal tolerances and common workspace bound from the
 proof of Theorem~\ref{cr:thm:local-main}. Whenever a prescribed
 input-weight or estimate-range check fails, return its padded
 initialized identity fallback.
\State Initialize $\mathbb W_0=\I_\Omega$. Compress the supplied lists as in
 Lemma~\ref{cr:lem:sparsify}; set $L_0=K_0$ exactly.
 Each retained term has $O(\log(M+1))$ support, and each list has
 $\poly(M,\log(m/p_{\mathrm{fail}}))$ terms, where $M=|\Omega|$.
\For{$j=1,\ldots,m$}
 \State Encode $G_j(t)$ in Eq.~\eqref{cr:eq:generator}; amplify it to
 constant normalization; compile $E_j(1)$ in Eq.~\eqref{cr:eq:E}
 by the time-ordered construction.
 \State Encode $W_j^0$ in Eq.~\eqref{cr:eq:W0}; estimate
 $\trn{A_j}((W_j^0)^\dagger W_j^0)$ through the independent
 success-probability measurements of
 Lemma~\ref{lem:normalization-estimation}.
 \State Store $a_j>0$, form $W_j=a_j^{-1/2}W_j^0$, and amplify its
 encoding to the normalization in Eq.~\eqref{eq:map-normalizations}.
\State Append the independently compiled factor to the cumulative
  product $\mathbb W_j=(\mathbb W_{j-1}\otimes\I_{\mathcal E_j})W_j$,
  retaining a flat gate-sequence description.
\EndFor
\State For the two endpoints $j=m-1,m$ only, encode $\mathbb W_j$ from
 Eq.~\eqref{cr:eq:correction-product}, form
 $\mathbb W_j^\dagger\mathbb W_j-\I$ with fresh projection ancillas
 and the intermediate output projector, and amplify to normalization
 $1/32$; for $j=0$ use $X_0=0$.
\State Construct $T_{m-1},T_m$ by
 Eq.~\eqref{eq:relative-inverse-integral}, then encode
 $T_{m-1}^{-1}=T_{m-1}^\dagger(\I+X_{m-1})$ and amplify its
 normalization to $2$ as in Lemma~\ref{cr:lem:T}; for $m=1$ use
 $T_0=\I$.
\State Encode $V_m$ using Eq.~\eqref{cr:eq:V-factor} and amplify
 its initialized action as in the proof of
 Theorem~\ref{cr:thm:local-main}, with the fixed precision budgets.
\State \Return the padded unitary gate list $\mathscr U_m$.
\end{algorithmic}
\end{algorithm}

\section{Global circuit construction and resource bounds}
\label{sec:global-complexity}

Appendix~\ref{sec:analytic-composition} supplies the recovery layers;
Theorem~\ref{cr:thm:local-main} supplies each local unitary. We combine
these results and account separately for classical reference generation.

\subsection{Parallel and sequential preparation}
\label{sec:parallel-algorithm}

\begin{definition}[Parallel preparation circuit]\label{def:global-algorithm}
For the regions in Eq.~\eqref{eq:color-regions}, let $\mathscr U_x$
be the local unitary from Theorem~\ref{cr:thm:local-main}, with its
own reference $\mathcal E_x$ and work register $\mathsf A_x$.
Put $\mathsf A=\bigotimes_x\mathsf A_x$. After preprocessing, set
\begin{equation}\label{eq:parallel-unitary}
\begin{aligned}
 \mathscr U^{(c)}&=\prod_{x\in\Lambda^{(c)}}\mathscr U_x,\\
 \mathscr U_{\mathrm{prep}}&=
 \mathscr U^{(n_{\mathrm{col}})}\cdots\mathscr U^{(1)},\qquad
 |\Psi_{\mathrm{out}}\rangle=\mathscr U_{\mathrm{prep}}|0\rangle.
\end{aligned}
\end{equation}
Here $|0\rangle$ initializes $\Lambda\mathcal E_\Lambda\mathsf A$.
Same-color circuits have disjoint supports; earlier reference and work
registers remain untouched. Temporary coefficient registers belong to
$\mathsf A_x$, and all preprocessing outcomes return unitaries on the
same padded registers.
\end{definition}

\begin{theorem}[Parallel Gibbs-state preparation]\label{thm:global}
Assume the setup, microscopic access, (E1)--(E2), the high-temperature
condition~\eqref{an:eq:strip}, and
Assumption~\ref{cr:ass:classical}(i) with $r_{\mathrm{ref}}\le e_M$
for every color-ordered coefficient chain. For $0<\varepsilon<1/2$,
choose the radii in Eq.~\eqref{eq:global-radii} and
\begin{equation}\label{eq:global-errors}
 \eta_x=\frac{\varepsilon}{4N},\qquad
 p_{\mathrm{fail},x}=\frac{\varepsilon}{8N}.
\end{equation}
With probability at least $1-\varepsilon/8$ over preprocessing,
\begin{equation}\label{eq:global-pure-accuracy}
 \norm{|\Psi_{\mathrm{out}}\rangle
 -|\Gamma_\beta\rangle|0\rangle_{\mathsf A}}_2\le\varepsilon/4.
\end{equation}
On this event the full output and its system marginal have trace-norm
error at most $\varepsilon/2$ from their respective targets. The
preprocessing-averaged output obeys
\begin{equation}\label{eq:global-accuracy}
 \left\|\E|\Psi_{\mathrm{out}}\rangle\langle\Psi_{\mathrm{out}}|
 -|\Gamma_\beta\rangle\langle\Gamma_\beta|
   \otimes|0\rangle\langle0|_{\mathsf A}\right\|_1
 \le\varepsilon.
\end{equation}
For fixed physical parameters and dimension,
\begin{equation}\label{eq:global-cost}
\begin{aligned}
 G_{\mathrm Q}&\le C_{\mathrm g}N[\log(N/\varepsilon)]^{\nu_{\mathrm g}},\\
 \operatorname{depth}(\mathscr U_{\mathrm{prep}})
 &\le C_{\mathrm d}[\log(N/\varepsilon)]^{\nu_{\mathrm d}},\\
 n_{\mathrm{qubits}}&=2Nn_q+n_{\mathsf A},\\
 n_{\mathsf A}&\le C_{\mathrm w}N[\log(N/\varepsilon)]^{\nu_{\mathrm w}}.
\end{aligned}
\end{equation}
The gate budget includes normalization-estimation trials; the depth is
that of the fixed preparation unitary. Classical graph processing and
compilation, excluding reference generation, cost
$N\polylog(N/\varepsilon)$. Under
Assumption~\ref{cr:ass:classical}(ii), the total classical cost is
$N^{1+o(1)}$, or $N^{o(1)}$ per patch, when
$\log(1/\varepsilon)=O(\log N)$.
\end{theorem}

\begin{proof}
Equation~\eqref{eq:global-radii} ensures
\begin{equation}\label{eq:localization-budget}
 C_{\mathrm{loc}}^Q
 \left[e^{-\kappa R}+(1+R)^{3D/4}e^{-\kappa\ell/16}\right]
 \le\varepsilon/(8N).
\end{equation}
Corollary~\ref{cor:local-purification} and
Theorem~\ref{cr:thm:local-main} give
vector error $\varepsilon/8+\sum_x\eta_x/2=\varepsilon/4$ and
failure probability at most $\sum_xp_{\mathrm{fail},x}=\varepsilon/8$,
including the earlier registers. Trace-norm error is at most twice
vector error on success and at most two on failure; its average is
therefore at most $\varepsilon/2+2(\varepsilon/8)\le\varepsilon$.
With $M=O([\log(N/\varepsilon)]^D)$ and
Eq.~\eqref{eq:global-errors}, the local cost is
$\polylog(N/\varepsilon)$. Gates add over $N$ sites, while depths
multiply by the layer count in Eq.~\eqref{eq:color-count}.
The $N$ physical sites, $N$ references, and local workspaces give
Eq.~\eqref{eq:global-cost}. The fixed circuit has no measurements or
resets and has an explicit inverse. Neighbor enumeration, coloring,
and gate-list assembly cost $N\polylog(N/\varepsilon)$; the remaining
classical cost is Proposition~\ref{prop:classical-budget}.
\end{proof}

\begin{algorithm}
\caption{Global preparation from local circuits}
\label{alg:global-petz}
\begin{algorithmic}[1]
\Require Microscopic input, $\beta$, $\varepsilon$, and references
 satisfying Theorem~\ref{thm:global}.
\State Choose $R,\ell$ and local budgets using
 Eqs.~\eqref{eq:global-radii} and~\eqref{eq:global-errors}.
\State Greedily color sites at separation at most $2(R+\ell)$;
 set $\Omega_x,Y_x$ by Eq.~\eqref{eq:color-regions}.
\State For each $x$, obtain the references for
 $\Omega_x\to Y_xx\to Y_x$ and use Algorithm~\ref{alg:local-petz}
 to construct $\mathscr U_x$ with fresh $\mathcal E_x,\mathsf A_x$.
\State \Return the color-ordered circuit in
 Eq.~\eqref{eq:parallel-unitary}, with same-color circuits parallel
 and all reference and work registers retained.
\end{algorithmic}
\end{algorithm}

\begin{corollary}[Reusable workspace]\label{cor:sequential-workspace}
\label{sec:sequential-resources}
Under Theorem~\ref{thm:global}, execute the local circuits sequentially
in the same site order, discarding each reference and resetting its work
register after use. The mixed output satisfies
\begin{equation}\label{eq:sequential-resources}
\begin{aligned}
 \norm{\E\widetilde\rho_\Lambda-\gamma_\Lambda}_1&\le\varepsilon,\\
 G_{\mathrm Q}^{\mathrm{seq}},\ \operatorname{depth}_{\mathrm{seq}}
 &\le N\polylog(N/\varepsilon),\\
 n_{\mathrm{qubits}}^{\mathrm{seq}}
 &\le Nn_q+\polylog(N/\varepsilon).
\end{aligned}
\end{equation}
Classical costs are unchanged. Other fixed site orders have the same
bounds when all required regions and deletion chains satisfy the
stated hypotheses, including clause (ii) for the classical cost.
\end{corollary}

\begin{proof}
Tracing unused earlier registers leaves the system marginal unchanged.
Reusing local workspaces gives the space bound; summing local gates gives
sequential depth. For another order use
Corollary~\ref{thm:analytic-assembly} with the same error and failure sums.
\end{proof}

\subsection{Classical reference costs}
\label{sec:classical-resources}\label{app:classical-resources}

\paragraph{Local arithmetic and finite precision.}
\label{sec:local-reference-model}
The generation guarantee of Assumption~\ref{cr:ass:classical}(ii)
uses at most $\poly(M,\log N)$ dense calculations on neighborhoods
of at most $c_{\mathrm{nb}}(1+\ell_{\mathrm{cl}})^D$ sites. A
calculation on $s$ sites to $b_{\mathrm{num}}$ bits costs
$\exp(O(s))\poly(M,\log N,b_{\mathrm{num}})$ bit operations.
The reference-error and weighted-coefficient bounds are attained with
$b_{\mathrm{num}}=\poly(s,\log(M+1),\log N)$, including list assembly,
centering, scanning, and numerical errors, uniformly over the required
patches and deletion chains. At fixed local dimension, this permits
polynomial cost in the neighborhood Hilbert-space dimension $q^s$.

\begin{proposition}[Classical reference-generation cost]
\label{prop:classical-budget}
Under Assumption~\ref{cr:ass:classical}(ii), choose
\begin{equation}\label{eq:classical-radius}
 \ell_{\mathrm{cl}}=\max\!\left\{1,
 \left\lceil\frac1{c_{\mathrm{cl}}}
 \log\left(1+\frac{C_{\mathrm{cl}}m}{e_M}\right)\right\rceil\right\}.
\end{equation}
Then $r_{\mathrm{ref}}\le e_M$, $\ell_{\mathrm{cl}}=O(\log(M+1))$,
and
\begin{equation}\label{eq:classical-patch-cost}
 \mathcal C_{\mathrm{ref}}(\Omega)
 \le\poly(M,\log N)\exp\!\left(O([\log(M+1)]^D)\right).
\end{equation}
For $M=O([\log(N/\varepsilon)]^D)$ and
$\log(1/\varepsilon)=O(\log N)$, the complete classical costs,
including graph processing and circuit construction, satisfy
\begin{equation}\label{eq:classical-global-cost}
 \mathcal C_{\mathrm{cl,patch}}=N^{o(1)},\qquad
 \mathcal C_{\mathrm{cl,total}}=N^{1+o(1)}.
\end{equation}
\end{proposition}

\begin{proof}
The chosen radius gives
$C_{\mathrm{cl}}m e^{-c_{\mathrm{cl}}\ell_{\mathrm{cl}}}
 \le C_{\mathrm{cl}}m/(1+C_{\mathrm{cl}}m/e_M)\le e_M$.
Since $e_M=c_e(M+1)^{-p_{\mathrm{ref}}}$ and $m\le M$, each classical
neighborhood has $s=O([\log(M+1)]^D)$ sites. Substitution in the
arithmetic model proves Eq.~\eqref{eq:classical-patch-cost}.
At inverse-polynomial accuracy the logarithm of the per-patch cost is
$O([\log\log N]^D+\log\log N)=o(\log N)$ for every fixed $D$.
Summing over $N$ patches and adding graph and compilation costs gives
Eq.~\eqref{eq:classical-global-cost}.
\end{proof}

\end{document}